\documentclass[10pt,a4paper,reqno]{amsart}
\usepackage{amsthm,amsfonts,amssymb,amsmath,euscript,comment}
\usepackage{color}
\usepackage{graphicx}
\usepackage{color}
\usepackage{bbm}
\usepackage{geometry}
\usepackage{amsfonts}
\newtheorem{thm}{Theorem}[section]
\newtheorem{prp}[thm]{Proposition}
\newtheorem{cor}[thm]{Corollary}
\newtheorem{lm}[thm]{Lemma}
\newtheorem{df}[thm]{Definition}
\newtheorem{rk}[thm]{Remark}
\newtheorem{proposition}[thm]{Proposition}

\newtheorem{theorem}[thm]{Theorem}

\newcommand{\m}[1]{\mathbb{#1}}

\newcommand{\q}[1]{\mathcal{#1}}

\newcommand{\wht}[1]{\widetilde{#1}}

\newcommand{\ba}[1]{\underline{#1}}
\newcommand{\ep}{\varepsilon}

\newcommand{\f}{\frac}
\newcommand{\rd}{\partial}
\newcommand{\nab}{\nabla}
\newcommand{\alp}{\alpha}
\newcommand{\bt}{\beta}
\newcommand{\bA}{{\bf A}}
\newcommand{\bB}{{\bf B}}

\newcommand{\gi}{(g^{-1})}

\newcommand{\ls}{\lesssim}
\newcommand{\de}{\delta}

\numberwithin{equation}{section}

\begin{document}

\title
{Einstein equations under polarized $\mathbb U(1)$ symmetry in an elliptic gauge}

\begin{abstract}
We prove local existence of solutions to the Einstein--null dust system under polarized $\mathbb U(1)$ symmetry in an elliptic gauge. Using in particular the previous work of the first author on the constraint equations, we show that one can identify freely prescribable data, solve the constraints equations, and construct a unique local in time solution in an elliptic gauge. Our main motivation for this work, in addition to merely constructing solutions in an elliptic gauge, is to provide a setup for our companion paper in which we study high frequency backreaction for the Einstein equations. In that work, the elliptic gauge we consider here plays a crucial role to handle high frequency terms in the equations. The main technical difficulty in the present paper, in view of the application in our companion paper, is that we need to build a framework consistent with the solution being high frequency, and therefore having large higher order norms. This difficulty is handled by exploiting a reductive structure in the system of equations.
\end{abstract}

\author{C\'ecile Huneau}
\address{Institut Fourier, Université Grenoble-Alpes, 100 rue des maths, 38610 Gières, France}
\email{cecile.huneau@univ-grenoble-alpes.fr}
\author{Jonathan Luk}
\address{Department of Mathematics, Stanford University, CA 94304, USA}
\email{jluk@stanford.edu}
	
	\maketitle

\section{Introduction}
In this paper, we study the Einstein equation
$$R(g)_{\mu\nu}-\f 12 g_{\mu\nu} R(g) = T_{\mu\nu}.$$
under polarized $\mathbb U(1)$ symmetry in an elliptic gauge. We will consider the case where the stress-energy-momentum tensor $T_{\mu\nu}$ is either that of vacuum or a finite number of families of null dust. Previously, it was known that 
\begin{itemize}
\item given freely prescribable initial data, the constraint equations in vacuum for small data can be solved \cite{Huneau.constraints}, and 
\item a local, geometrically unique (large data) solution to the Einstein--null dust system exists in a wave coordinate gauge, even without the polarized $\mathbb U(1)$ symmetry assumption\footnote{Although strictly speaking, \cite{CBF} deals with the case where the dust is massive, the methods apply to the null case with little modifications, cf. \cite{livrecb}} \cite{CBF}.
\end{itemize}
Our main result in this paper is that in a small data regime, the constraints can be solved and that local existence can be established in an \underline{elliptic} gauge; see the precise statement in Section~\ref{sec.gauge}. Alternatively, this means that at least for a short time, in the solution that we already know exists by \cite{CBF}, an elliptic gauge can be constructed. In particular, under {suitable} conditions on the {initial} data, our {result constructs a local-in-time} maximal foliation.

Our motivation for studying the Einstein equation in an elliptic gauge is that under such a gauge condition, one can obtain additional regularity for some metric components and is therefore useful for low-regularity problems. An elliptic gauge is especially advantageous under polarized $\mathbb U(1)$ symmetry\footnote{Such an effective decoupling in fact occurs under $\mathbb U(1)$ symmetry without the polarization assumption. For simplicity, however, we only consider the polarized case in this paper.} since in this case the ``dynamical part'' of the solution and the ``elliptic part'' of the solution (which is more regular) essentially decouple (cf. \eqref{sys} and \eqref{back}). A specific application, which we discuss in our companion paper \cite{HL.HF}, is to study high-frequency backreaction for the Einstein equations. Precisely, we show in \cite{HL.HF}, using the elliptic gauge studied in the present paper, that any generic small data smooth polarized $\mathbb U(1)$ symmetric solution to the Einstein--\underline{null dust} system arise as suitable weak limits of solutions to the Einstein \underline{vacuum} equation. Physically, as we discuss at length in \cite{HL.HF}, this can be thought of as meaning that ``high frequency gravitational waves give rise to an effective stress-energy-momentum tensor of null dust in the limit''. Notice that it is in view of the application in \cite{HL.HF} that we also include null dust in our equations in this paper.

Since one of the purposes of this paper is to provide a setup for \cite{HL.HF}, our local existence and uniqueness statement is in particular consistent with the initial data being ``high-frequency'' in a suitable sense. In particular, it has the following features:
\begin{itemize}
\item (Choice of elliptic gauge) The elliptic gauge condition that we impose is such that the spacetime is foliated by maximal spacelike hypersurfaces $\Sigma_t$ and that on each $\Sigma_t$, the intrinsic metric is conformal to the Euclidean metric.\footnote{Strictly speaking, the condition that the initial hypersurface is maximal is a condition on the geometric data and is \underline{not} a gauge condition. One may in principle also consider that setting where the mean curvature is a prescribed regular function, cf. discussions in \cite{AM}. We will however be content with the restriction that the initial hypersurface is maximal and not pursue a general result.} As a consequence, all metric components satisfy semilinear elliptic equations.
\item (Lack of decay at infinity) One main technical challenge in our setting arises from the fact that we work in two spatial dimensions in all of $\mathbb R^2$. In this case, one needs to carefully control the logarithmically divergent terms arising from the inversion of the Laplacian on $\mathbb R^2$, and place the remaining terms in appropriate weighted Sobolev spaces.
\item (Large data for high norms) Another, more serious, technical challenge concerns the smallness that we can choose in this problem. In order to solve the constraints and handle the nonlinearity in the elliptic part of the system, one needs some smallness of the solution (in addition to the smallness in time). Nevertheless, in view of the application in \cite{HL.HF}, where we study \emph{high frequency} solutions, the solutions are necessarily large in any $W^k_p$ norms (cf. Definition~\ref{def.spaces}) for $p \in [1,\infty)$, $k>1$. We therefore study in this paper a solution regime where the $W^{1}_{\infty}$ norm of the initial data are required to be small, yet higher norms can be arbitrarily large. The main technical challenge of this paper is therefore to treat the elliptic part of the system -- where one cannot exploit the small time parameter -- using only the smallness of the low order norms.
\end{itemize}
As is standard, to obtain a solution to the Einstein equations in our gauge, we first introduce and solve a reduced system, and a posteriori show that the solution to the reduced system is indeed a solution to the Einstein equations. In our case, the reduced system is a coupled system of elliptic, wave and transport equations. Let us note that in order to handle both the issue of the lack of decay at infinity and the largeness of the higher order norms, we exploit a \emph{reductive} structure of the reduced system. By this we mean that one can introduce a hierarchy of estimates (both in terms of weights and in terms of size) so that when considered in an appropriate sequence, one can bound the terms one by one in order to obtain the desired estimates.

The remainder of this paper will be organized as follows: 
\begin{itemize}
\item In \textbf{Section~\ref{sec.notations}}, we introduce the notations for this paper.
\item In \textbf{Section~\ref{sec:ENDU1}}, we introduce the class of polarized $\mathbb U(1)$ spacetimes and the system of equations to be studied. 
\item In \textbf{Section~\ref{sec.elliptic.gauge}}, we introduce our elliptic gauge condition.
\item In \textbf{Section~\ref{sec.gauge}}, we give the main result of the paper.
\item In \textbf{Section~\ref{sec.reduced}}, we introduce a reduced system.
\item In \textbf{Section~\ref{sec.data.constraints}}, we study the constraint equations, following \cite{Huneau.constraints}.
\item In \textbf{Section~\ref{sec.solve.reduced}}, we prove existence and uniqueness of solutions to the reduced system (introduced in Section~\ref{sec.reduced})
\item In \textbf{Section~\ref{sec.final}}, we show that a solution to the reduced system is a solution to the original system.
\item In \textbf{Section~\ref{sec:improved.regularity}}, we conclude the proof of the main theorem (Theorem~\ref{lwp}) by proving all the estimates stated in Theorem~\ref{lwp}.
\item Finally, we have three appendices:
\begin{itemize}
	\item In \textbf{Appendix~\ref{weightedsobolev}}, we collect some results about Sobolev embedding, product estimates and elliptic estimates in weighted Sobolev spaces in $\mathbb R^2$.
	\item In \textbf{Appendix~\ref{app.B}}, we collect some computations in the elliptic gauge.
	\item In \textbf{Appendix~\ref{app.C}}, we collect some computations for the eikonal functions.
	\end{itemize}
\end{itemize}

\subsection*{Acknowledgements} Most of this work was carried out when both authors were at Cambridge University. {C. Huneau is supported by the ANR-16-CE40-0012-01.} J. Luk is supported in part by a Terman fellowship. 

\section{Notations and function spaces}\label{sec.notations}

{\bf Ambient space and coordinates}
In this paper, we will be working on the ambient manifold $\q M:=I\times \mathbb R^2$, where $I\subset \mathbb R$ is an interval. The space will be given a system of coordinates $(t,x^1,x^2)$. We will use $x^i$ with the lower case Latin index $i,j=1,2$ and will also sometime denote $x^0=t$.

{\bf Conventions with indices} We will use the following conventions:
\begin{itemize}
\item Lower case Latin indices run through the spatial indices $1,2$, while lower case Greek indices run through all the spacetime indices.
\item Repeat indices are always summed over: where lower case Latin indices sum over the spatial indices $1,2$ and lower case Greek indices sum over all indices $0,1,2$.
\item Unless otherwise stated, lower case Latin indices are always raised and lowered with respect to the standard Euclidean metric $\delta_{ij}$.
\item In contrast, lower case Greek indices are raised and lowered with respect to the spacetime metric $g$. In cases where there are more than one spacetime metric in the immediate context, we will not use this convention but will instead spell out explicitly how indices are raised and lowered.
\end{itemize}

{\bf Differential operators} We will use the following conventions for differential operators:
\begin{itemize}
\item $\rd$ denotes partial derivatives in the coordinate system $(t,x^1,x^2)$. We will frequently write $\rd_i$ for $\rd_{x^i}$. In particular, we denote
$$|\rd \xi|^2=(\rd_t\xi)^2+\sum_{i=1}^2(\rd_{x^i}\xi)^2.$$
\item The above $\rd$ notation also applied to rank-$r$ covariant tensors $\xi_{\mu_1\dots\mu_r}$ tangential to $I\times \mathbb R^2$ to mean
$$|\rd\xi|^2=\sum_{\mu_1,\dots,\mu_r=t,x^1,x^2}|\rd \xi_{\mu_1\dots\mu_r}|^2$$
and to rank-$r$ contravariant tensors $\xi_{i_1\dotsi_r}$ tangential to $\mathbb R^2$ to mean
$$|\rd\xi|^2=\sum_{i_1\dotsi_r=x^1,x^2}|\rd \xi_{i_1\dotsi_r}|^2.$$
\item $\Delta$ and $\nabla$ denotes the spatial Laplacian and the spatial gradient on $\m R^2$ with the standard \underline{Euclidean metric}. In particular, we use the convention
$$|\nabla\xi|^2=\sum_{i=1}^2|\rd_{x^i}\xi|^2.$$
\item $D$ denotes the Levi--Civita connection associated to the \underline{spacetime metric $g$}.
\item $\Box_g$ denotes the Laplace--Beltrami operator on functions, i.e., 
$$\Box_g\xi:=\f{1}{\sqrt{|\det g|}}\rd_\mu(\gi^{\mu\nu}\sqrt{|\det g|}\rd_\nu\xi).$$
\item $\q L$ denotes the Lie derivatives.
\item $e_0$ defines the vector field $e_0=\rd_t-\beta^i\rd_{x^i}$ (where $\beta$ will be introduced in \eqref{g.form}). We will often use the differential operator $\q L_{e_0}$.
\item $L$ denotes the Euclidean conformal Killing operator acting on vectors on $\m R^2$ to give a symmetric traceless (with respect to $\delta$) covariant $2$-tensor, i.e., 
$$(L\xi)_{ij}:=\delta_{j\ell}\rd_i\xi^\ell+\delta_{i\ell}\rd_j\xi^\ell-\delta_{ij}\rd_k\xi^k.$$
\end{itemize}

{\bf Functions spaces} We will work with standard function spaces $L^p$, $H^k$, $C^m$, $C^\infty_c$, etc. and assume the standard definitions. The following conventions will be important:
\begin{itemize}
\item \underline{Unless otherwise stated, all function spaces will be taken on $\m R^2$} and the measures will be taken to be the 2D Lebesgue measure $dx$. 
\item When applied to quantities defined on a spacetime $I \times \m R^2$, the norms $L^p$, $H^k$, $C^m$ denote \underline{fixed-time} norms (unless otherwise stated). In particular, if in an estimate the time $t\in I$ in question is not explicitly stated, then it means that the estimate holds for \underline{all} $t\in I$ for the time interval $I$ that is appropriate for the context.
\end{itemize}
We will also work in weighted Sobolev spaces, which are well-suited to elliptic equations. We recall here the definition, together with the definition of weighted H\"older space. The properties of these spaces that we need are listed in Appendix \ref{weightedsobolev}.
\begin{df} \label{def.spaces}
Let  $m\in \m N$, $1<p<\infty$, $\delta \in \mathbb{R}$. The weighted Sobolev space $W^m_{\delta,p}$ is the completion of $C^\infty_0$ under the norm 
	$$\|u\|_{W^m_{\delta,p}}=\sum_{|\beta|\leq m}\|(1+|x|^2)^{\frac{\delta +|\beta|}{2}}\nab^\beta u\|_{L^p}.$$
We will use the notation $H^m_\delta = W^m_{\delta,2}$, $L^p_{\de}=W^0_{\de,p}$ and $W^s_p=W^s_{p,0}$.

	The weighted H\"older space $C^m_{\delta}$ is the complete space of $m$-times continuously differentiable functions under the norm 
	$$\|u\|_{C^m_{\delta}}=\sum_{|\beta|\leq m}\|(1+|x|^2)^{\frac{\delta +|\beta|}{2}}\nab^\beta u\|_{L^\infty}.$$
\end{df}
Finally, let us introduce the convention that we will use the above function spaces for both tensors and scalars on $\m R^2$, where the norms in the case of tensors are understood componentwise.

\section{Einstein--null dust system and reduction under polarized $\mathbb U(1)$ symmetry}\label{sec:ENDU1}
From now on, we consider a Lorentzian manifold $(I\times \m R^{3}, ^{(4)}g)$, where $I\subset \mathbb R$ is an interval, and $^{(4)}g$ is a Lorentzian metric that takes the following form,
$$^{(4)}g= e^{-2\phi}g+ e^{2\phi}(dx^3)^2,$$
where $\phi:I\times \m R^{2}\to \m R$ is a scalar function and $g$ is a Lorentzian metric on $I\times \m R^{2}$. Abusing notation, we also extend $\phi$ to a function $\phi:I\times \m R^3$ in such a way that $\phi$ is independent of $x^3$. Given this ansatz of the metric, the vector field $\partial_{x_3}$ is Killing and hypersurface orthogonal.

On the manifold $I\times \m R^{3}$, we introduce the null dust variables $(F_{\bA}, u_{\bA})$, where $\bA\in \mathcal A$ for some finite set $\mathcal A$ with $|\mathcal A|=N$, $F_{\bA}: I \times \m R^2 \to \m R$, $u_{\bA}: I\times \m R^2 \to \m R$, (again also extended to $I\times \m R^3$ in a manner independent of $x^3$) so that 
$$(^{(4)}g^{-1})^{\alp\bt} \rd_\alp u_{\bA} \rd_\bt u_{\bA} = 0.$$
Define the stress-energy-momentum tensor 
$$^{(4)}T_{\mu\nu} = \sum_{\bA} (F_{\bA})^2\rd_\mu u_{\bA} \rd_\nu u_{\bA}.$$
The Einstein--null dust system is given by
\begin{equation}\label{END4D}
\left\{\begin{array}{l}
R_{\mu \nu}(^{(4)} g)= \sum_{{\bA}} (F_{{\bA}})^2\partial_\mu u_{{\bA}} \partial_\nu u_{{\bA}},\\
2(^{(4)}g^{-1})^{\alpha \beta}\partial_{\alpha} u_{{\bA}} \partial_{\beta} F_{{\bA}} + (\Box_{^{(4)}g} u_{{\bA}}) F_{{\bA}} = 0,\\
(^{(4)}g^{-1})^{\alpha \beta}\partial_\alpha u_{{\bA}} \partial_\beta u_{{\bA}}=0.
\end{array}
\right.
\end{equation}

Notice that the Einstein vacuum equations $R(^{(4)}g)_{\mu\nu}=0$ for the $(3+1)$-dimensional metric is included as a particular case. 

The above symmetry assumptions (for $^{(4)}g$, $u_\bA$ and $F_\bA$) are known as polarized $\mathbb U(1)$ symmetry. Under polarized $\mathbb U(1)$ symmetry, the system \eqref{END4D} reduces to the following equivalent system in $(2+1)$ dimensions:
\begin{equation}\label{back}
\left\{\begin{array}{l}
R_{\mu \nu}(g)= 2\partial_\mu \phi \partial_\nu \phi + \sum_{{\bA}} (F_{{\bA}})^2\partial_\mu u_{{\bA}} \partial_\nu u_{{\bA}},\\
\Box_{g}\phi= 0,\\
2(g^{-1})^{\alpha \beta}\partial_{\alpha} u_{{\bA}} \partial_{\beta} F_{{\bA}} + (\Box_{g} u_{{\bA}}) F_{{\bA}} = 0,\\
(g^{-1})^{\alpha \beta}\partial_\alpha u_{{\bA}} \partial_\beta u_{{\bA}}=0.
\end{array}
\right.
\end{equation}

In particular, the Einstein vacuum equations $R(^{(4)}g)_{\mu\nu}=0$ are equivalent to the following system for $(g,\phi)$:
\begin{equation}
\label{sys}\left\{
\begin{array}{l}
\Box_g \phi = 0,\\
R_{\mu \nu}(g)= 2\partial_\mu \phi \partial_\nu \phi.
\end{array}
\right.
\end{equation}

\section{Elliptic gauge}\label{sec.elliptic.gauge}
We write the $(2+1)$-dimensional metric $g$ on $\q M:=I\times \mathbb R^2$ in the form
\begin{equation}\label{g.form.0}
g=-N^2dt^2 + \bar{g}_{ij}(dx^i + \beta^i dt)(dx^j + \beta^jdt).
\end{equation}
Let $\Sigma_t:=\{(s,x^1,x^2): s=t\}$ and $e_0= \partial_t -\beta^i\rd_i$, which is a future directed normal to $\Sigma_t$.
We introduce the second fundamental form of the embedding $\Sigma_t \subset \q M$
\begin{equation}\label{K}K_{ij}=-\frac{1}{2N}\q L_{e_0} \bar{g}_{ij}.
\end{equation}
 We decompose $K$ into its trace and traceless parts.
\begin{equation}\label{K.tr.trfree} 
K_{ij}=:H_{ij}+\frac{1}{2}\bar{g}_{ij}\tau.
\end{equation}
Here, $\tau:=\mbox{tr}_{\bar{g}} K$ and $H_{ij}$ is therefore traceless with respect to $\bar{g}$.

Introduce the following \emph{gauge conditions}:
\begin{itemize}
	\item $\bar{g}$ is conformally flat, i.e., for some function $\gamma$,
	\begin{equation}\label{uniformized.g}
	\bar{g}_{ij}=e^{2\gamma}\delta_{ij};
	\end{equation}
	\item The constant $t$-hypersurfaces $\Sigma_t$ are maximal
	$$\tau=0.$$
\end{itemize}
By \eqref{g.form.0}, it follows that 
\begin{equation}\label{g.form}
g=-N^2dt^2 + e^{2\gamma}\delta_{ij}(dx^i + \beta^i dt)(dx^j + \beta^jdt).
\end{equation}
Hence the determinant of $g$ is given by
\begin{equation}\label{g.det}
\det(g)=e^{2\gamma}\beta^2(-e^{4\gamma}\beta^2)+e^{2\gamma}(e^{2\gamma}(-N^2+e^{2\gamma}|\beta|^2)-e^{4\gamma}\beta^1\beta^1)=-e^{4\gamma}N^2.
\end{equation}
Moreover, the inverse $g^{-1}$ is given by
\begin{equation}\label{g.inverse}
g^{-1}=\frac{1}{N^2}\left(\begin{array}{ccc}-1 & \beta^1 & \beta^2\\
\beta^1 & N^2e^{-2\gamma}-\beta^1\beta^1 & -\beta^1 \beta^2\\
\beta^2 & -\beta^1 \beta^2 & N^2e^{-2\gamma}-\beta^2\beta^2
\end{array}
\right).
\end{equation}

\section{Main results}\label{sec.gauge}

\subsection{Initial data}\label{sec.id}

In this section, we describe the initial data for \eqref{sys} and \eqref{back}. We will focus our discussions on \eqref{back} as the local well-posedness of \eqref{sys} clearly follows from that of \eqref{back}.

The initial data for \eqref{back} consist of the prescription of the geometry (first and second fundamental forms of $\Sigma_0$) as well as the matter fields. For convenience, we will require $\nab\phi$, its normal derivative and $F_{\bA}$ to be initial compactly supported. By the finite speed of propagation, they will remain compactly supported. 

To completely specify the initial data, we also need to prescribe the initial values for solutions to the eikonal equation $\gi^{\mu\nu}\rd_\mu u_\bA \rd_\nu u_\bA=0$. To this end, we will prescribe the initial values for $u_{\bA}\restriction_{\Sigma_0}$ and will require also that 
\begin{enumerate}
\item $\min_{\bA} \inf_{x\in \mathbb R^2}|\nabla u_{\bA}\restriction_{\Sigma_0}|(x) > C_{eik}^{-1}$ for some $C_{eik}>0$,
\item $(du_{\bA})^\sharp\restriction_{\Sigma_0}$ to be past-directed, $\forall \bA$.
\end{enumerate}
 The former condition in particular implies that $u_{\bA}$ has no critical points. The latter condition is equivalent to requiring that $(e_0 u_{\bA})\restriction_{\Sigma_0}>0$ (or, equivalently, by \eqref{g.form}, $(e_0 u_{\bA})\restriction_{\Sigma_0}=Ne^{-\gamma}|\nabla u_{\bA}|\restriction_{\Sigma_0}$). Moreover, while $u_{\bA}$ only becomes relevant in a compact subset\footnote{Indeed, $u_{\bA}$ only influences the metric according to \eqref{back} on the support of $F_{\bA}$.}, we will for technical convenience define $u_{\bA}$ globally and also require the level sets of $u_{\bA}$ to be asymptotic to planes in $\m R^2$, or more precisely, for each $\bA\in \mathcal A$, there exists a constant vector field\footnote{i.e., $\overrightarrow{c_{\bA}}=((c_1)_{\bA},(c_2)_{\bA})$, where $(c_1)_{\bA},(c_2)_{\bA}\in \m R$ are constants.} $\overrightarrow{c_{\bA}}$ such that $\nab u_{\bA}-\overrightarrow{c_{\bA}}$ is in an appropriate weighted Sobolev space.

Before we proceed to define the notion of admissible initial data, we need to fix a cutoff function for the rest of the paper:
\begin{df}[Cutoff function $\chi$]\label{def.cutoff}
From now on, let $\chi(|x|)$ be a fixed smooth cutoff function with $\chi=0$ for $|x|\leq 1$ and $\chi=1$ for $|x|\geq 2$.
\end{df}

We now make precise the discussions on the initial data set in the following definition:
\begin{df}[Admissible initial data]\label{def.data}
For $-\frac{1}{2}<\delta<0$, $k \geq 3$, $R>0$ and $\mathcal A$ a finite set, an {\bf admissible initial data set} with respect to the elliptic gauge for \eqref{back} consists of 
\begin{enumerate}
\item {a} conformally flat intrinsic metric $e^{2\gamma}\delta_{ij}\restriction_{\Sigma_0}$ which admits a decomposition
$$\gamma=-\alp \chi(|x|) \log (|x|) +\widetilde{\gamma},$$
where $\alp\geq 0$ is a constant, $\chi(|x|)$ is as in Definition~\ref{def.cutoff}, and $\tilde{\gamma}\in H^{k+2}_{\delta}${;}
\item {a} second fundamental form $(H_{ij})\restriction_{\Sigma_0} \in H^{k+1}_{\delta+1}$ which is traceless;
 \item $(\f{1}{N}(e_0\phi),\nabla \phi) \restriction_{\Sigma_0}\in H^k$, compactly supported in $B(0,R)$;
		\item $F_{\bA} \restriction_{\Sigma_0} \in H^k$, compactly supported in $B(0,R)$ for every $\bA\in \mathcal A$;
		\item $u_{\bA}\restriction_{\Sigma_0}$ such that $\inf_{x\in \mathbb R^2}|\nabla u_{\bA}\restriction_{\Sigma_0}|(x)> C_{eik}^{-1}$ for some $C_{eik}>0$ and $\left(\nabla u_{\bA}\restriction_{\Sigma_0}-\overrightarrow{c_{\bA}} \right)\in H^{k+1}_\delta$, where $\overrightarrow{c_{\bA}}$ is a constant vector field for every $\bA\in \mathcal A$.
\end{enumerate}
$\gamma$ and $H$
are required to satisfy the following {\bf constraint equations}:
\begin{align}
&\de^{ik}\partial_k H_{ij}=-\f{2e^{2\gamma}}{N}(e_0\phi)\partial_j \phi -\sum_{\bA} e^\gamma F_\bA^2|\nabla u_\bA|\partial_j u_\bA,\label{mom}\\
&\Delta \gamma + e^{-2\gamma}\left(\f{e^{4\gamma}}{N^2}(e_0\phi)^2+\frac{1}{2}|H|^2\right)+|\nabla \phi|^2+\sum_\bA F_\bA^2|\nabla u_\bA|^2=0.\label{ham}
\end{align}
\end{df}

It turns out that we can find freely prescribable initial data, from which (under suitable smallness assumptions) one can construct admissible initial data satisfying the constraint equations. To this end, it will be convenient not to prescribe the unit normal derivative $\f{1}{N}e_0$ of the scalar field $\phi$ and the density of the null dusts $F_{\bA}$, but instead prescribe appropriately rescaled versions as defined in \eqref{data.rescaled}. We define the notion of admissible free initial data as follows:
\begin{df}[Admissible free initial data]\label{def.free.data}
Define $\dot{\phi}$, $\breve{F}_\bA$ as follows:
\begin{equation}\label{data.rescaled}
\dot{\phi}=\frac{e^{2\gamma}}{N} (e_0 \phi),\quad\breve{F}_\bA= F_{\bA} e^{\frac{\gamma}{2}},
\end{equation}
where $\gamma$ is as in \eqref{g.form}.

For $-\frac{1}{2}<\delta<0$, $k \geq 3$, $R>0$ and $\mathcal A$ a finite set, an {\bf admissible free initial data set} with respect to the elliptic gauge is given by the following:
\begin{enumerate}
		\item $(\dot{\phi},\nabla \phi) \restriction_{\Sigma_0}\in H^k$, compactly supported in $B(0,R)$;
		\item $\breve{F}_{\bA} \restriction_{\Sigma_0} \in H^k$, compactly supported in $B(0,R)$ for every $\bA\in \mathcal A$;
		\item $u_{\bA}\restriction_{\Sigma_0}$ such that $\inf_{x\in \mathbb R^2}|\nabla u_{\bA}\restriction_{\Sigma_0}|(x)> C_{eik}^{-1}$ for some $C_{eik}>0$ and $\left(\nabla u_{\bA}\restriction_{\Sigma_0}-\overrightarrow{c_{\bA}} \right)\in H^{k+1}_\delta$, where $\overrightarrow{c_{\bA}}$ is a constant vector field for every $\bA\in \mathcal A$.
	\end{enumerate}
Moreover, $(\dot{\phi},\nabla \phi,\breve{F}_\bA, u_\bA)\restriction_{\Sigma_0}$ is required to satisfy
\begin{equation}\label{main.data.cond}
\int_{\mathbb R^2} \left(-2\dot{\phi}\rd_j \phi- {\sum_{\bA}}\breve{F}_{\bA}^2|\nab u_{\bA}|\rd_j u_{\bA}\right) \, dx{=0}.
\end{equation}
\end{df}

The fact that we claimed above, i.e., that an admissible free initial data set gives rise to an actual admissible initial data satisfying the constraint equations, will be the content of Lemma~\ref{lm.constraint}.

\subsection{Local well-posedness}\label{seclwpintro}

The following is our main result on local well-posedness for \eqref{back} (and therefore also \eqref{sys}). As we already mentioned in the introduction, we need a smallness assumption \eqref{smallness.fd}, but importantly for the applications in \cite{HL.HF}, it is required only for the lower norms but not the high norms.
\begin{thm}\label{lwp}
Let $-\frac{1}{2}<\delta<0$, $k \geq 3$, $R>0$ and $\mathcal A$ be a finite set. Given a free initial data set as in Definition \ref{def.free.data} such that
	\begin{equation}\label{smallness.fd}
\|\dot{\phi}\|_{L^\infty}+\|\nabla \phi\|_{L^\infty} + {\max_{\bA}}\|\breve{F}_{\bA}\|_{L^\infty}\leq \ep{,}
	\end{equation}
	{and}
	\begin{equation}\label{init.u.bd}
	{C_{eik}:= {\left(\min_{\bA} \inf_{x\in \mathbb R^2} |\nab u_{\bA}|(x)\right)^{-1} + \max_{\bA}}\|\nab u_{\bA}-\overrightarrow{c_{\bA}}\|_{H^{k+1}_{\delta}}<\infty,}
	\end{equation}
	and 
	\begin{equation}\label{C.high.def}
	C_{high}:=\|\dot{\phi}\|_{H^k}+\|\nabla \phi\|_{H^k}+ \|\breve{F}_{\bA}\|_{H^k} <\infty.
	\end{equation}
	Then{, for any $C_{eik}$ and $C_{high}$,} there exists a constant $\ep_{low}=\ep_{low}({C_{eik}, k,\delta,R})>0$ \underline{independent of $C_{high}$} and a $T=T(C_{high},{C_{eik},} k, \delta, R)>0$ such that if $\ep<\ep_{low}$, there exists a unique solution to \eqref{back} in elliptic gauge on $[0,T]\times \m R^2 $.
	Moreover, defining $\delta'=\delta-\ep,$ $\delta''=\delta-2\ep$, $\delta'''=\delta-3\ep$, the following holds for some constant $C_h = C_h(C_{eik}, C_{high}, k,\de,R)>0$:
	\begin{itemize}
	\item {The following estimates hold for $\phi$, $F_{\bA}$ and $u_{\bA}$ for all $\bA\in \mathcal A$ for $t\in [0,T]$:}
	\begin{align*}
	\|\nabla \phi\|_{H^k} +\|\partial_t \phi \|_{H^k} +\|\partial^2_t \phi\|_{H^{k-1}}\leq &C_h, \\
	\max_\bA\left(\|F_\bA \|_{H^k} + \|\partial_t F_{\bA}\|_{H^{k-1}}+\|{\partial^2_t} F_{\bA}\|_{H^{k-2}}\right)\leq &C_h, \\
	\left(\min_{\bA} \inf_{x\in \mathbb R^2} |\nab u_{\bA}|(x)\right)^{-1} + \max_\bA\left(\|\nabla u_{\bA}-\overrightarrow{c_{\bA}} \|_{ H^{k}_{\delta''}} +\|e^{\gamma}N^{-1} (e_0 u_{\bA}) -|\overrightarrow{c_{\bA}}|\|_{H^{k}_{\delta''}}\right)\leq & C_h,\\
	\max_\bA\left(\|\partial_t\nabla u_{\bA} \|_{ H^{k-1}_{\delta'''}}+\|\partial_t^2\nabla u_{\bA} \|_{ H^{k-2}_{\delta''''}}+\| \partial_t \left(\f{e^{\gamma}}{N} e_0 u_\bA \right)\|_{ H^{k-1}_{\delta'''}}
			+\|\partial_t^2\left(\f{e^{\gamma}}{N} e_0 u_\bA \right)\|_{H^{k-2}_{\delta''''}}\right)\leq &C_h.
				\end{align*}				
	\item The metric components $\gamma$ and $N$ can be decomposed as
				$$\gamma = {\alp}\chi(|x|){\log}(|x|) + \wht \gamma,\quad N = 1+ N_{asymp}(t)\chi(|x|){\log}(|x|) + \wht N,$$
				with ${\alp}\leq 0$ a constant, $N_{asymp}(t)\geq 0$ a function of $t$ alone and $\chi(|x|)$ is as in Definition~\ref{def.cutoff}.
	\item {$\gamma$, $N$ and $\beta$ obey the following estimates for $t\in [0,T]$:}
	\begin{align*}
	{|{\alp}|} +\|{\wht \gamma}\|_{ H^{k+2}_{\delta}} + \|\partial_t {\wht \gamma}\|_{H^{k+1}_{\delta}} + \|\partial_t^2 {\wht \gamma} \|_{H^k_{\delta}}\leq &C_h,\\
	|N_{asymp}| +|\partial_t N_{asymp}|+|\partial^2_t N_{asymp}|\leq &C_h,\\
	\|\wht N\|_{H^{k+2}_\delta}+ \|\partial_t \wht N\|_{H^{k+1}_\delta}+\|\partial^2_t \wht N\|_{H^{k}_\delta}\leq &C_h,\\
	\|\beta \|_{H^{k+2}_{\delta'}} + \|\partial_t \beta\|_{H^{k+1}_{\delta'}}+\|\partial^2_t \beta\|_{H^{k}_{\delta'}} \leq &C_h.
	\end{align*} 
	\item The support of $\phi$ and $F_{\bA}$ satisfies $$supp(\phi,F_{\bA}) \subset J^+(\{t=0\}\cap B(0, R)),$$
	where $J^+$ denotes the causal future.
	\end{itemize}	\end{thm}

\begin{rk}[The $|\mathcal A|\to \infty$ limit]
We observe from the proof that in Theorem~\ref{lwp}, we do not need $\mathcal A$ to be a finite set\footnote{We also remark that this is in contrast to the problem in our companion paper \cite{HL.HF}, where the assumption that $|\mathcal A|$ is finite is necessary}. Instead, in the case $|\mathcal A|=\infty$, as long as we replace the estimates for $\breve{F}_\bA$ in \eqref{smallness.fd} and \eqref{C.high.def} with appropriate $\ell^2$ norms, i.e.,
$$\left(\sum_\bA \|\breve{F}_{\bA}\|_{L^\infty}^2\right)^{\f 12}\leq  \ep,\quad \left(\sum_{\bA} \|\breve{F}_\bA\|_{H^k}^2 \right)^{\f 12}<\infty,$$
with all other assumptions unchanged, then the conclusion in Theorem~\ref{lwp} still holds.
\end{rk}

\begin{rk}
	We {remark on} the following fact{s regarding} the maximal foliation:
	\begin{itemize}
		\item The lapse function {$N$} has a logarithmic growth {as $|x|\to \infty$}.
		\item The following conservation laws hold{:}
\begin{align}
&\int_{\m R^2}\left(\f{2e^{2\gamma}}{N}(e_0\phi)\partial_j \phi +\sum_{\bA} e^\gamma F_\bA^2|\nabla u_\bA|\partial_j u_\bA\right)dx=0, \label{cons1}\\
&\int_{\m R^2}\left( e^{-2\gamma}\left(\f{e^{4\gamma}}{N^2}(e_0\phi)^2+\frac{1}{2}|H|^2\right)+|\nabla \phi|^2+\sum_\bA F_\bA^2|\nabla u_\bA|^2\right)dx=\alpha \label{cons2}.
\end{align}
	\end{itemize}
\end{rk}

It will be useful to note the following easy consequence of the proof of Theorem \ref{lwp}, which states that in the genuinely small data regime, the time of existence can be taken to be $T=1$. Since its proof is simplier than the general case in Theorem~\ref{lwp}, we omit its proof.
\begin{cor}\label{lwp.small}
Suppose the assumptions of Theorem~\ref{lwp} hold and let $\{c_{\bA}\}_{\bA\in \mathcal A}$ be a collection of constant vector fields on the plane. There exists $\ep_{small}=\ep_{small}(\delta,k,R,c_{\bA})$ such that if $C_{high}$ and $\ep$ in Theorem \ref{lwp} both satisfy
$$C_{high},\ep\leq \ep_{{small}}$$
and moreover
$$\sum_{\bA}\|\nabla u_{\bA}- \overrightarrow{c_{\bA}}\|_{H^{k+1}_\delta}\leq \ep_{{small}},$$
then the unique solution exists in $[0,1]\times \m R^2$. Moreover, there exists $C_0{= C_0(\de, k, R,c_{\bA})}$ such that {all the estimates in Theorem~\ref{lwp} hold with $C_h$ replaced by $C_0\ep$.}

\end{cor}

As we mentioned above, we will omit the details of the proof of Corollary~\ref{lwp.small}. We will focus on the proof of Theorem~\ref{lwp}, which will occupy most of the remainder of the paper. In order to simplify the exposition, for most of the paper, \textbf{we will assume $k=3$.} Higher derivatives estimates, i.e., the case $k>3$, follows straightforwardly from the ideas presented here.

To prove Theorem \ref{lwp}, we first introduce in Section \ref{sec.reduced} a \emph{reduced system} of equations \eqref{tau}--\eqref{chi}, which is an elliptic-hyperbolic-transport system. We then discuss the initial data appropriate for this system in Section~\ref{sec.data.constraints}. In Section~\ref{sec.solve.reduced}, we solve the reduced system using an iteration scheme. Then in Section~\ref{sec.final}, we prove that the solution to the reduced system indeed is a solution to \eqref{back}. Finally, in Section~\ref{sec:improved.regularity}, we conclude by proving all the estimates as stated in Theorem~\ref{lwp}

\section{The reduced system}\label{sec.reduced}
We consider the following system of equations, which we will call the {\bf reduced system}. Let us recall that lower case Latin indices are raised and lowered with respect to $\delta_{ij}$. We will also denote by $\Gamma^{\alp}_{\mu\nu}$ the Christoffel symbols associated to $g$.
\begin{align}
 \label{tau} & Ne^{2\gamma}\tau = -2e_0 \gamma +\rd_i\bt^i,\\
\label{gamma} &
2\Delta \gamma=2\left(\frac{\tau^2}{2}e^{2\gamma}-\frac{1}{2N}e^{2\gamma}e_0 \tau-\frac{1}{2N}\Delta N\right)- 2\delta^{ij}\rd_i\phi\rd_j\phi-\sum_{\bA} \f{e^{2\gamma}}{N^2}F_{\bA}^2(e_0 u_{\bA})^2, \\
\label{lapse}& \Delta N-e^{-2\gamma}N(|H|^2+\f 12 e^{4\gamma}\tau^2)-\frac{2e^{2\gamma}}{N}(e_0 \phi)^2-\f{e^{2\gamma}}{N}\sum_{\bf A} F_{\bf A}^2(e_0 u_{\bA})^2=0,\\
\label{shift}&(L \beta)_{ij}=2Ne^{-2\gamma}H_{ij},\\
\label{hij}&(\partial_t-\beta^k\partial_k) H_{ij}=
-2e^{-2\gamma} N H_i{ }^{\ell} H_{j\ell} +\left(\partial_j \beta^k H_{ki}+\partial_i \beta^k H_{kj}\right)\\
\nonumber&\qquad- \partial_i \partial_j N +\frac 12\delta_{ij}\Delta N  +\left(\delta_i^k\partial_j \gamma +\delta_j^k\partial_i \gamma-\delta_{ij}\de^{\ell k} \partial_{\ell} \gamma\right)\partial_k N \\
\nonumber &\qquad-2N\partial_i \phi \partial_j \phi -N\sum_{\bA} F_{\bA}^2 \rd_i u_{\bA} \rd_j u_{\bA} + N\delta_{ij}|\nabla \phi|^2 +\frac{e^{2\gamma}}{2N}\delta_{ij} \sum_{\bA}F_{\bA}^2(e_0 u_{\bA})^2,\\
\label{phi}&\Box_g \phi = 0,\\
\label{geodesic.reduced} & L_{\bA}^\rho \rd_\rho L_{\bA}^\alpha + \Gamma^\alpha_{\mu \nu} L_{\bA}^\mu L_{\bA}^\nu = 0,\quad\forall \bA,\\
\label{u} & L_{\bA}^\rho  \rd_\rho u_{\bA}=0,\quad\forall \bA,\\
\label{F} &2L_{\bA}^\rho \partial_\rho F_{\bA}+\chi_{\bA} F_{\bA}= 0,\quad\forall \bA,\\
\label{chi} &L_{\bA}^\rho \partial_\rho \chi_{\bA}+\chi_{\bA}^2=-2(L_{\bA}^{\rho} \partial_\rho \phi)^2
-\sum_{\bf B} F_{\bf B}^2(g_{\mu\nu} L_{\bA}^{\mu} L_{\bB}^{\nu})^2,\quad\forall \bA.
\end{align}
In deriving the above equations, we have used the computations in Sections~\ref{sec.Ricci} and \ref{sec.T}. We note that \eqref{tau} is the definition of $\tau$ to be the mean curvature. \eqref{gamma} is derived by setting\footnote{Recall that $R_{\mu\nu}-\f 12 g_{\mu\nu} R=T_{\mu\nu}$. Therefore,  $-\f 12 R=(g^{-1})^{\mu\nu}T_{\mu\nu}$ and hence $R_{\mu\nu}=T_{\mu\nu}+\f 12 g_{\mu\nu}R=T_{\mu\nu}-g_{\mu\nu} \mbox{tr}_g T$.} $\delta^{ij}R_{ij}=\delta^{ij}(T_{ij}-g_{ij}\mbox{tr}_g T)$, where we have used \eqref{spatial.R.tr}; \eqref{lapse} is derived by setting $R_{00}=T_{00}- g_{00} \mbox{tr}_g T$ in the case $e_0\tau=0$; \eqref{shift} follows from \eqref{beta}; \eqref{hij} is derived by setting $R_{ij}-\f 12 \delta_{ij}\delta^{k\ell}R_{k\ell}=T_{ij}-g_{ij}\mbox{tr}_g T-\f 12\delta^{ij}(T_{ij}-g_{ij}\mbox{tr}_g T)$. The equations \eqref{phi}--\eqref{chi} are chosen according to the propagation equations for the matter fields in \eqref{back}, except for issues to be discussed in Remarks~\ref{rmk:chi} and \ref{rmk:u}.

\begin{rk}[Only $N$ and $\beta$ are solved by elliptic equations]
While in the full system, $N$, $\beta$, $\gamma$ and $H$ all satisfy elliptic equations, in the reduced system, only $N$ and $\bt$ are solved through elliptic equations. $\gamma$ and $H$ are defined to be solutions to wave and transport equations respectively. 
We have to adopt this procedure, because if we wanted to {solve the} elliptic equations for $\gamma$ and $H$, given by the constraints \eqref{mom} and {\eqref{ham}}, we would need the conservation law \eqref{cons1} to hold a priori and for each iterate of our {iteration} scheme.
\end{rk}


\begin{rk}[Introduction of $\chi_{\bA}$]\label{rmk:chi}
Notice that in \eqref{F}, we are not directly solving the transport equation for $F_{\bA}$ in \eqref{back}, but we have replaced $\Box_g u_{\bA}$ by $\chi_{\bA}$, which is an auxiliary function that we introduce and is required to satisfy \eqref{chi} according to \eqref{Raychaudhuri} and \eqref{boxu=chi}. The reason is that otherwise we would need to be very careful in the iteration procedure to make use of the special structure in order to not lose derivatives. Instead, by introducing $\chi_{\bA}$ and treating it as a separate variable, we are exploiting that fact that by the Raychaudhuri equation, $\Box_g u_{\bA}$ is more regular than generic second derivatives of $u_{\bA}$ in the full nonlinear system. This allows us to more easily close the iteration scheme, and it is only a posteriori that we show $\chi_{\bA}=\Box_g u_{\bA}$ (see Proposition~\ref{prop:final}).
\end{rk}

\begin{rk}[Solving for $u_{\bA}$]\label{rmk:u}
In order to solve the eikonal equation $\gi^{\mu\nu} \rd_\mu u_{\bA} \rd_\nu u_{\bA}=0$, it is convenient to solve the geodesic equation \eqref{geodesic.reduced} for the geodesic null vector field $L_{\bA}$ and then define $u_{\bA}$ by \eqref{u}. It is a standard fact in Lorentzian geometry that (given appropriate initial conditions) in fact $L_{\bA}^{\alp} = -\gi^{\alp\bt} \rd_\bt u_{\bA}$ and that $\gi^{\mu\nu} \rd_\mu u_{\bA} \rd_\nu u_{\bA}=0$.
\end{rk}

\section{Initial data and the constraint equations}\label{sec.data.constraints}

In this section, we discuss the initial data for the reduced system. The most important task is to solve the constraint equations. In particular, we will show (as claimed in Section~\ref{sec.id}) that an admissible free initial data set gives rise to a unique admissible initial data set satisfying the constraint equations. Unlike in the later section of the paper, in this section, we will consider general $k\geq 3$ as it does not complicate the notations. In this step, we largely follow the ideas in \cite{Huneau.constraints}.

After we solve the constraint equations (which can be viewed as PDEs for $\gamma$ and $H$, in the remainder of this section, we will derive the initial data for $N$, $\bt$ (Lemma~\ref{lem:data.N.beta}), $e_0\gamma$ (Lemma~\ref{lem:e0gamma}), $L_{\bA}$ (Lemma~\ref{lem:data.L}) and $\chi_{\bA}$ (Lemma~\ref{lem:data.chi}) and prove their regularity properties. Note that since $\nab\phi$, $\dot{\phi}$ and $\breve{F}_\bA$ are prescribed (cf. \eqref{data.rescaled}), after we derive the initial data for $N$, $\gamma$ and $\bt$, we obtain the initial data for $\nab\phi$, $e_0\phi$ and $F_{\bA}$.

Let us first set up the notation of this section: We will use $C$ to denote a constant depending only on $C_{eik}$, $k$, $\de$ and $R$ (and independent of $C_i$ and $\ep$). We will also use the notation $\ls$ where the implicit constant has the same dependence as $C$.

Before we proceed to solving the constraints, it will be convenient to rewrite \eqref{mom} and \eqref{ham} in term of $\dot{\phi}$ and $\breve{F}$ as follows:
\begin{align}
\label{mom2} &\de^{ik} \partial_k H_{ij}=-2\dot{\phi}\partial_j \phi - \sum_{\bA} \breve{F}_{\bA}^2|\nabla u_{\bA}|\partial_j u_{\bA},\\
\label{ham2}&\Delta \gamma + e^{-2\gamma}\left(\dot{\phi}^2+\frac{1}{2}|H|^2\right)+|\nabla \phi|^2+ \sum_{\bA} e^{-\gamma} \breve{F}_{\bA}^2|\nabla u_{\bA}|^2=0.
\end{align}
The following is the main result on solving the constraint equations:
\begin{lm}\label{lm.constraint}
	Let $-\frac{1}{2}<\delta<0$, $k\geq 3$, $R>0$ and $\mathcal A$ be a finite set.
	Given an admissible free initial data set as in Definition~\ref{def.free.data} such that the smallness assumption \eqref{smallness.fd} holds. Then, for $\ep$ sufficiently small (depending on $C_{eik}$ (cf. \eqref{init.u.bd}), $k$, $\delta$ and $R$), there exists a unique admissible initial data set as in Definition~\ref{def.data} corresponding to the given the admissible free initial data set. In particular, there exist a unique solution $(H,\gamma)$ solving the constraint equations \eqref{mom2}, \eqref{ham2} with
$H \in H^{k+1}_{\delta+1}$ being a symmetric traceless covariant $2$-tensor and 
$$\gamma = -\alpha \chi(|x|){\log}(|x|)+ \wht \gamma$$
with $\alpha\geq 0$ being a constant, $\chi$ being as in {Definition~\ref{def.cutoff}} and $\wht \gamma \in H^{k+2}_\delta$ being a function.
Moreover,
\begin{equation}\label{est.constraints.0}
|\alpha|+ \|H\|_{H^{1}_{\delta+1}} + \|\wht \gamma \|_{H^{2}_\delta} \lesssim \ep^2 ,
\end{equation}
\begin{equation}\label{est.constraints}
\|H\|_{W^1_{\de+\f 32,4}}+ \|\wht\gamma\|_{W^2_{\de+\f 12,4}} + \|H\|_{C^0_{\delta+2}} + \|\wht \gamma \|_{C^1_{\delta+1}} \lesssim \ep^2.
\end{equation}
\end{lm}

\begin{proof}
	{\bf Solving \eqref{mom2}.} We solve for $H_{ij}$ which takes the form 
	$H_{ij}=(L Y)_{ij}= \de_{j\ell} \partial_i Y^\ell + \de_{i\ell} \partial_j Y^\ell-\delta_{ij}\partial_k Y^k$
	for some $1$-form $Y_i$.
	Then \eqref{mom2} is equivalent to 
	\begin{equation}\label{mom2.Y}
	\Delta Y_j = -2\dot{\phi}\partial_j \phi - \sum_{\bA}\breve{F}_{\bA}^2|\nabla u_{\bA}|\partial_j u_{\bA}.
	\end{equation}
	Since $\dot{\phi}, \phi, F_{\bA}$ are compactly supported, by \eqref{smallness.fd}, we have the following bound on the RHS of \eqref{mom2.Y}
	$$\| 2\dot{\phi}\partial_j \phi + \sum_{\bA} \breve{F}_{\bA}^2|\nabla u_{\bA}|\partial_j u_{\bA}\|_{H^0_{\delta+2}} \lesssim \ep^2.$$
	Moreover, by the regularity assumptions and support properties of $\dot\phi$, $\nab\phi$, $\breve{F}_\bA$ and $u_\bA$, the RHS of \eqref{mom2.Y} is also in $H^k_{\de+2}$. Therefore, by the condition \eqref{main.data.cond} and Theorem~\ref{laplacien}, there exists a unique $Y_j \in H^{k+2}_{\delta}$ with
	$$\|Y_j\|_{H^{2}_{\delta}}\lesssim \ep^2.$$
	Consequently, there exists a symmetric traceless (with respect to $\de$) covariant $2$-tensor $H \in H^{k+1}_{\delta+1}$ solving \eqref{mom2} with
	\begin{equation}\label{H.constraint.1st.est}
	\|H\|_{H^{1}_{\delta+1}}\lesssim \ep^2.
	\end{equation}
	Moreover, since every symmetric  traceless divergence-free (with respect to $\de_{ij}$) covariant $2$-tensor on $\mathbb R^2$ vanishes\footnote{This is standard and can be seen by noting that for a symmetric traceless covariant $2$-tensor $\eta_{ij} \in H^2_{\de'+1}${,} we have {the componentwise identity} $\Delta \eta_{ij}=0$, so {that} Theorem~\ref{laplacien} implies the conclusion.\label{footnote.divfree.tracefree}}
	
	{\bf Solving \eqref{ham2}.}
	We now turn to \eqref{ham2}. First of all, we note that thanks to Proposition~\ref{produit} we have $|H|^2 \in H^{k+1}_{2\delta+3}\subset H^{k+1}_{\delta+2},$
	and
	$$\||H|^2\|_{H^0_{2\delta+3}} \lesssim \ep^4.$$
	We solve \eqref{ham2} with the contraction mapping theorem. We consider the map\footnote{Here, and below, we use the notation that $B_{H^{2}_\delta} (0,\ep)$ denotes the open ball centered at $0$ with radius $\ep$ in Banach space $B_{H^{2}_\delta}$.}
	$$\Phi : [0,\ep]\times B_{H^{2}_\delta} (0,\ep)\rightarrow [0,\ep]\times B_{H^{2}_\delta} (0,\ep),$$
	which maps $(\alpha^{(1)},\wht\gamma^{(1)})\mapsto (\alpha^{(2)},\wht\gamma^{(2)})$ such that $\gamma^{(1)}= -\alpha^{(1)} \chi(|x|){\log}(|x|)+\wht \gamma^{(1)}$, $ \gamma^{(2)} = -\alpha^{(2)}\chi(|x|){\log}(|x|)+\wht \gamma^{(2)}$ and the latter is defined as the solution to
	$$\Delta \gamma^{(2)}=- e^{-2\gamma^{(1)}}\left(\dot{\phi}^2+\frac{1}{2}|H|^2\right)-|\nabla \phi|^2-\sum_{\bA} e^{-\gamma^{(1)}} \breve{F}_{\bA}^2|\nabla u_{\bA}|^2.$$
	All the terms involving $\phi$ or $\breve{F}_{\bA}$ are compactly supported and of size $O(\ep^2)$ in $H^{0}_{\delta+2}$. For the $e^{-2\gamma^{(1)}}\frac{1}{2}|H|^2$ term, we check that by \eqref{H.constraint.1st.est}, $e^{2 \alpha^{(1)} \chi(|x|){\log}(|x|)}|H|^2 \in H^{0}_{\delta+2}$ (provided that $\ep$ is small enough) and $e^{-2\wht{\gamma}^{(1)}}|H|^2 \in H^{0}_{\delta+2}$, both with norms $O(\ep^2)$. Notice also that
$$ -C \ep^2 \leq \f{1}{2\pi} \int_{\mathbb R^2} (\mbox{RHS of \eqref{ham2}}) \leq 0. $$	
Therefore, by Corollary~\ref{coro}, the range of $\Phi$ indeed lies in the set $[0,\ep]\times B_{H^{2}_\delta} (0,\ep)$. Moreover, a similar argument for the differences shows that $\Phi$ is a contraction (for $\ep$ sufficiently small).
	
	Therefore, by the contraction mapping theorem, we obtain a unique fixed point $(\alpha,\wht\gamma)$ of $\Phi$ with $\wht \gamma \in H^2_\delta$ and
	$$|\alpha|+ \|\wht \gamma \|_{H^2_\delta}\lesssim \ep^2.$$
Moreover, $\gamma= -\alpha \chi(|x|){\log}(|x|)+\wht \gamma$ solves \eqref{ham2}.

  Since every term on the RHS of \eqref{ham2} is nonnegative, by Theorem \ref{laplacien}, $\alpha\geq 0$.
	Using the estimates for $H$ and $\gamma$ that we just proved, and also the assumptions on $\dot{\phi}$, $\nab\phi$, $F_{\bA}$ and $u_{\bA}$, we see that the right-hand side of \eqref{ham2} is in $H^2_{\delta+2}$. Hence, by Theorem~\ref{laplacien}, $\wht \gamma \in H^4_{\delta+2}$.
	Continuing to iterate this, we conclude that $\wht \gamma \in H^{k+2}_{\delta}$.

{\bf Proof of \eqref{est.constraints}.}
We first prove the bounds for $H$. Since the right-hand side of \eqref{mom2.Y} is compactly supported and bounded in $L^\infty$ by $\ep^2$, we can use Theorem \ref{laplacien} with $p=4$ to infer that $Y\in W^2_{\nu,4}$ and hence 
$$\|H\|_{W^1_{\nu+1,4}}\ls \ep^2 \mbox{ for all }-\f 12 <\nu <\f 12.$$  In particular, thanks to the Sobolev embedding in Proposition \ref{holder} we have $H \in C^0_{\nu+\f 32 }$ with the bound
	$$\|H \|_{ C^0_{\nu+\f 32}} \lesssim \ep^2.$$
In the same manner, we have, for $-\f 12<\nu<\f 12$,
	$$	\|\wht \gamma\|_{W^2_{\nu,4}}+ \|\wht \gamma \|_{ C^1_{\nu+ \frac 12}} \lesssim \ep^2.$$
	Choosing $\nu= \delta+\f 12$ (recall that $\delta\in (-\f 12, 0)$), we obtain \eqref{est.constraints}.
	\end{proof} 
	We now turn to the initial data for the lapse $N$ and the shift $\bt^i$.
	\begin{lm}\label{lem:data.N.beta}
		Let $\delta'=\delta-\ep$.
		For $\ep$ sufficiently small, there exists unique $(N,\beta)$ such that $N=1+N_{asymp}\chi(|x|){\log}(|x|)+ \wht N$, with $N_{asymp}\in \mathbb R$, $\wht N \in H^{k+2}_\delta$, and $\beta \in H^{k+2}_{\delta'}$ such that
		\begin{equation}
		\label{N.data}
			\Delta N-e^{-2\gamma}N\left(|H|^2+\dot{\phi}^2+e^{\gamma}\sum_{\bf A} \breve{F}_{\bf A}^2|\nabla u_{\bA}|^2\right)=0,
			\end{equation}
		\begin{equation}\label{beta.data}
		(L \beta)_{ij}=2Ne^{-2\gamma}H_{ij}.
		\end{equation}
		Moreover, $N_{asymp}\geq 0$ and
		\begin{equation}\label{estN.data}
		|N_{asymp}|+\|\wht N\|_{H^2_\delta} + \|\wht N \|_{W^2_{\de+\f 12,4}}  + \|\wht N\|_{C^1_{\delta+1}}\lesssim \ep^2.
		\end{equation}
		\begin{equation}\label{estbeta.data}
		\|\beta \|_{H^2_{\delta'}} + \|\beta \|_{W^2_{\de'+\f 12,4}} + \|\beta  \|_{C^1_{\delta'+1}} \lesssim \ep^2.
		\end{equation}
		
	\end{lm}
	\begin{proof}
		{\bf Solving \eqref{N.data}.} We solve \eqref{N.data} with a fixed point argument. Consider the map
		$$\Phi : [0,\ep]\times B_{H^{2}_\delta} (0,\ep)\rightarrow [0,\ep]\times B_{H^{2}_\delta} (0,\ep),$$
		which maps $(N_{asymp}^{(1)}, \wht N^{(1)})\mapsto (N_{asymp}^{(2)}, \wht N^{(2)})$,  where given $N^{(1)}=1+N_{asymp}^{(1)}\chi(|x|){\log}(|x|) + \wht N^{(1)}$ with $N_{asymp}^{(1)}\in [0,\ep]$ and $\wht N^{(1)}\in B_{H^{2}_\delta} (0,\ep)$, we define $N^{(2)}=1+N_{asymp}^{(2)}\chi(|x|){\log}(|x|)+\wht N^{(2)}$ to be the solution to
		\begin{equation}\label{Np.eqn}
		\Delta N^{(2)}= e^{-2\gamma} N^{(1)}\left(|H|^2+\dot{\phi}^2+e^{\gamma}\sum_{\bf A} \breve{F}_{\bf A}^2|\nabla u_{\bA}|^2\right).
		\end{equation}
		We now show that this map has the range as claimed. Since $N_{asymp}^{(1)}\geq 0$, it follows from Proposition~\ref{holder} that $N^{(1)}\geq 1+\wht N^{(1)} \geq 1-C\ep \geq \frac{1}{2}$ for $\ep$ small enough, where $C$ is a universal constant (depending only on the constants in Proposition~\ref{holder}). As a consequence\footnote{Note that if we solve \eqref{Np.eqn} using Corollary~\ref{coro}, then $N_{asymp}^{(2)}$ is given by this expression.},
		$$\f 1{2\pi}\int_{\mathbb R^2} e^{-2\gamma} N^{(1)}\left(|H|^2+\dot{\phi}^2+e^{\gamma}\sum_{\bf A} \breve{F}_{\bf A}^2|\nabla u_{\bA}|^2\right) \geq 0.$$
		Since $(N_{asymp}^{(1)}, \wht N^{(1)}) \in [0,\ep] \times B_{H^2_\de}(0,\ep)$, by \eqref{est.constraints.0}, Proposition~\ref{produit} and Lemma~\ref{produit2}, we have 
		$$\| e^{-2\gamma} N^{(1)} |H|^2\|_{H^0_{\de+2}}\ls \ep^4$$ 
		for $\ep$ sufficiently small (necessary to handle the log weights in $\gamma$ and $N^{(1)}$). Also, using the compact support of $\dot{\phi}$ and $\breve{F}_{\bA}$, we have 
		$$\|e^{-2\gamma}\dot{\phi}^2+e^{-\gamma}\sum_{\bf A} \breve{F}_{\bf A}^2|\nabla u_{\bA}|^2\|_{H^2_{\de+2}} \ls \ep^2.$$
	  Therefore, by Theorem~\ref{laplacien}, for $\ep$ sufficiently small, there indeed exists $N^{(2)}$ with $(N_{asymp}^{(2)},\wht N^{(2)})\in [0,\ep] \times B_{H^2_\de}(0,\ep)$ solving \eqref{Np.eqn}.
		
		Moreover, since RHS of \eqref{Np.eqn} is linear in $N^{(1)}$, it is easy to apply to above argument to show that $\Phi$ is in fact a contraction (for $\ep$ sufficiently small). Hence, by the contraction mapping theorem, there exists a unique fixed point $N = 1+ N_{asymp} \chi(|x|){\log} (|x|) + \wht N$ that solves \eqref{N.data} with $(N_{asymp}, \wht N) \in [0,\ep] \times B_{H^2_\de}(0,\ep)$.
				
			Finally, using the bounds in Proposition~\ref{lm.constraint}, we can iteratively improve the estimate of $\wht N$ by applying Theorem~\ref{laplacien} and show that $\wht N \in H^{k+2}_\de$.
		
		{\bf Solving \eqref{beta.data}.} We now turn to the equation for $\beta$. Taking the divergence of \eqref{beta.data}, we obtain
		\begin{equation}\label{beta.data.2}
		\Delta \beta^j = \de^{ik} \de^{j\ell} \partial_k(2Ne^{-2\gamma}H_{i\ell}),
		\end{equation}
		which is a linear equation in $\bt^j$. We first note that by the estimates in Lemma~\ref{lm.constraint} for $\gamma$ and $H$, the estimates for $N$ that we just proved, Proposition~\ref{produit} and Lemma~\ref{produit2}, for $\delta' = \delta-\ep$, $Ne^{-2\gamma}H_{ij} \in H^{k+1}_{\delta'+1}$. Note that we have in particular used $e^{2\alp\chi(|x|){\log}(|x|)}\ls (1+|x|)^{C\ep^2}$ and $N_{asymp}\chi(|x|){\log}(|x|)\ls \ep^2{\log}(|x|)\ls \ep (1+|x|)^{\f{\ep}{10}}$.
		Hence, by Lemma~\ref{der}, $\partial_{k} (2Ne^{-2\gamma}H_{i\ell})\in H^{k}_{\delta'+2}$. Moreover,
		$$\int_{\mathbb R^2} \de^{ik} \de^{j\ell} \partial_k (2Ne^{-2\gamma}H_{ij})=0,$$
		Hence, by Theorem~\ref{laplacien} there exists a unique solution $\beta \in H^{k+2}_{\delta'}$ to \eqref{beta.data.2}.
		Since every $H^4_{\de'+1}$ symmetric traceless divergence-free (with respect to $\de_{ij}$) covariant $2$-tensor on $\mathbb R^2$ must vanish (cf. Footnote~\ref{footnote.divfree.tracefree} on p.\pageref{footnote.divfree.tracefree}), it then follows that $\beta$ is a solution to \eqref{beta.data}. 
		Moreover, using \eqref{est.constraints.0}, one sees that the above argument gives
		$$\|\beta\|_{H^2_{\delta'}}\lesssim \|\de^{ik}\partial_k(2Ne^{-2\gamma}H_{ij})\|_{H^0_{\delta'+2}}\lesssim \ep^2.$$
		
		\textbf{Smallness in weighted $L^4$-based Sobolev spaces.} It remains to show that $\|\wht N \|_{W^2_{\de+\f 12,4}}$ and $\|\beta \|_{W^2_{\de'+\f 12,4}}$ are $O(\ep^2)$ small, since the weighted $C^1$ estimates will then follow from Proposition~\ref{holder}.
		
		Now notice that these bounds can be proven by essentially the same arguments as above, except to use the estimates for $H$ in \eqref{est.constraints} instead of \eqref{est.constraints.0}. We omit the details.\qedhere
			\end{proof}
	We choose the initial data for $e_0 \gamma$ 
according to \eqref{tau} and the initial condition $\tau =0$. (Note that $\gamma$ satisfies a wave equation (cf. \eqref{tau}, \eqref{gamma}) and therefore we need the initial condition for $e_0\gamma$.)

\begin{lm}\label{lem:e0gamma}
In order that $\tau=0$, we set $e_0 \gamma = \frac{1}{2}div(\beta)$. Then, we have $e_0\gamma \in H^{k+1}_{\delta'+1}$ and 
$$\|e_0 \gamma \|_{H^1_{\delta'+1}} + \|e_0 \gamma\|_{W^1_{\de'+\f 32, 4}} + \|e_0 \gamma \|_{C^0_{\delta'+2}}\lesssim \ep^2.$$ 
\end{lm}
\begin{proof}
The desired estimates follow directly from Lemma~\ref{lem:data.N.beta} and Lemma~\ref{der}.\qedhere
\end{proof}

\begin{lm}\label{lem:data.L}
If $L_{\bA}^\mu = -\gi^{\mu\nu}\rd_\nu u_{\bA}$, $\gi^{\mu\nu}\rd_\mu u_{\bA} \rd_\nu u_{\bA} = 0$, $e_0 u_{\bA} >0$ with $g$ as in \eqref{g.form}, then
\begin{equation}\label{L.initial.computation}
L_{\bA}^t = \f{e^{-\gamma}}{N}|\nab u_{\bA}|, \quad  L_{\bA}^i = - e^{-2\gamma} \de^{ij} \rd_j u_{\bA} - \f{e^{-\gamma}\bt^i}{N}|\nab u_{\bA}|.
\end{equation}
Therefore, setting the initial data $L_{\bA}\restriction_{\Sigma_0}$ as in \eqref{L.initial.computation}, and for $\ep$ sufficiently small, $e^{2\gamma} L_\bA^i + \overrightarrow{c_{\bA}}^i,\, N e^{\gamma} L_\bA^t + |\overrightarrow{c_{\bA}}| \in H^k_{\de''}$, where $\de''= \de' - \ep = \de - 2\ep$.

Moreover, for $C_{eik}$ as in \eqref{init.u.bd}, {the following bounds hold for the initial data:}
\begin{align}
\label{L.init.bd}
\sup_{\bA} \left(\|e^{2\gamma} L_\bA^i + \overrightarrow{c_{\bA}}^i\|_{H^2_{\de''}} + \|N e^{\gamma} L_\bA^t + |\overrightarrow{c_{\bA}}|\|_{H^2_{\de''}}\right)\leq & 4 C_{eik},\\
\label{L.init.bd.lower} \min_{\bA} \inf_{x\in \mathbb R^2} \left|N e^{\gamma} (L_{\bA})^t \right|(x) \geq & {C_{eik}^{-1}}.
\end{align}

\end{lm}
\begin{proof}
$\gi^{\mu\nu}\rd_\mu u_{\bA} \rd_\nu u_{\bA} = 0$ implies 
\begin{equation}\label{eikonal.frame}
\f 1{N^2}(e_0 u_{\bA})^2 = e^{-2\gamma} |\nab u_\bA|^2.
\end{equation}
Hence, a direct computation gives
$$L_{\bA}^t = \f 1{N^2} e_0 u_{\bA} = \f{e^{-\gamma}}{N}|\nab u_{\bA}|,\quad L_{\bA}^i = - e^{-2\gamma} \de^{ij} \rd_j u_{\bA} - \f{\bt^i}{N^2}(e_0 u_{\bA}) = - e^{-2\gamma} \de^{ij} \rd_j u_{\bA} - \f{e^{-\gamma}\bt^i}{N}|\nab u_{\bA}|.$$
Now, the desired estimates $e^{2\gamma} L_\bA^i + \overrightarrow{c_{\bA}}^i,\, N e^{\gamma} L_\bA^t + |\overrightarrow{c_{\bA}}| \in H^k_{\de''}$ and \eqref{L.init.bd} follow from the bounds in \eqref{init.u.bd}, Lemmas~\ref{lm.constraint} and \ref{lem:data.N.beta} using Propositions~\ref{holder}, \ref{produit} and \ref{produit2}. Notice that here we need to use the fact that $\bt \in H^{k+2}_{\de'}$ to handle the term $\f{e^{-\gamma}\bt^i}{N}|\nab u_{\bA}|$ without taking difference with a constant vector $\overrightarrow{c_\bA}$. We also need to change the weight $\de' \mapsto \de''$ to handle the growing factors when $|x|$ is large; we omit the details. 

Finally, \eqref{L.init.bd.lower} follows from \eqref{init.u.bd} and \eqref{L.initial.computation}. \qedhere

\end{proof}

	Finally, we choose the initial data for $\chi_{\bA}$. Since we eventually will need $\chi_{\bA}=\Box_g u_{\bA}$ (cf. \eqref{boxu=chi}), we prescribe the initial data accordingly. Note that while $\Box_g u_{\bA}$ depends on the $e_0$ derivative of $u_{\bA}$, by virtue of the eikonal equation, it can in fact be computed from the initial data of $\nabla u_{\bA}\restriction_{\Sigma_0}$ alone. More precisely, we have the following estimates:
	\begin{lm}\label{lem:data.chi}
	Suppose $u_{\bA}$ satisfies $\gi^{\mu\nu}\rd_\mu u_{\bA} \rd_\nu u_{\bA}=0$, then $\Box_g u_{\bA}\restriction_{\Sigma_0}$ is given by the following expression:
	\begin{equation}\label{Box.u.data}
	\begin{split}
	\Box_g u_{\bA} \restriction_{\Sigma_0}=& \frac{1}{N}e^{-\gamma} (e_0\gamma) |\nabla u_{\bA}|\restriction_{\Sigma_0}+\frac{1}{Ne^{2\gamma}} \delta^{ij}\rd_i(N\rd_j u_{\bA})\restriction_{\Sigma_0}\\
	&-\frac{1}{N}e^{-\gamma}\left(\frac{1}{|\nabla u_{\bA}|}\delta^{ij}\rd_i u_{\bA} \rd_j (e^{-\gamma} N |\nabla u_{\bA}|) + \frac{1}{|\nabla u_{\bA}|}\delta^{ij}(\rd_i u_\bA)(\rd_j \beta^k)\partial_k u_{\bA}\right)\restriction_{\Sigma_0}.
	\end{split}
	\end{equation}
	 Therefore, by setting $\chi_{\bA} \restriction_{\Sigma_0}=(\mbox{RHS of \eqref{Box.u.data}})$, we have on $\Sigma_0$ that $\chi_{\bA} \in H^k_{\de}$. Moreover, there exists $C_\chi$ depending only on $C_{eik}$, $\de$, $k$ and $R$ such that
	\begin{equation}\label{chi.init.bd}
	\sup_{\bA}\|\chi_{\bA}\|_{C^0_{\de+1}} \leq C_{\chi}.
	\end{equation}
	\end{lm}
		\begin{proof}
	Using \eqref{tau}, we have
	\begin{equation*}
	\begin{split}
	\Box_g f=&\gi^{\mu\nu}\left(e_\mu e_\nu f-(\nabla_{e_{\mu}} e_\nu)f\right)\\
	=& -\frac{1}{N^2}e_0^2 f+\f {e_0 N}{N^3} e_0 f+\f 1{N} e^{-2\gamma}\delta^{ij} \rd_i N\rd_j f+e^{-2\gamma}\rd_i^2 f-\f{1}{2N^2}\left(4(e_0 \gamma)-2(\rd_i\bt^i)\right)(e_0 f)\\
	=& -\frac{1}{N^2}e_0^2 f+\f {e_0 N}{N^3} e_0 f+\f 1{N} e^{-2\gamma}\delta^{ij} \rd_i N\rd_j f+e^{-2\gamma}\rd_i^2 f+\f{e^{2\gamma}\tau}{N}(e_0 f).
	\end{split}
	\end{equation*}
	Since $\tau\restriction_{\Sigma_0}$ vanishes, we use the fact $e_0 u_{\bA}=e^{-\gamma}N |\nabla u_{\bA}|$ (cf. \eqref{eikonal.frame}) to obtain
	$$\Box_g u_{\bA} \restriction_{\Sigma_0}= -\frac{1}{N}e_0 (e^{-\gamma}|\nabla u_{\bA}|)\restriction_{\Sigma_0}+\frac{1}{Ne^{2\gamma}} \delta^{ij}\rd_i(N\rd_j u_{\bA})\restriction_{\Sigma_0}.$$
	The only term that does not manifestly depend only on $\nabla u_{\bA} \restriction_{\Sigma_0}$ is $e_0 |\nabla u_{\bA}|$. It can be re-expressed using the eikonal equation as follows
	$$e_0 |\nabla u_{\bA}|= \frac{1}{|\nabla u_{\bA}|}\delta^{ij}(\rd_i u_{\bA})( e_0 \rd_j u_{\bA})
	=\frac{1}{|\nabla u_{\bA}|}\delta^{ij}\rd_i u_{\bA} \rd_j (e^{-\gamma} N |\nabla u_{\bA}|) + \frac{1}{|\nabla u_{\bA}|}\delta^{ij}(\rd_i u_{\bA})(\rd_j \beta^k)\partial_k u_{\bA} .$$ 
	Combining the above expressions gives \eqref{Box.u.data}.
	
The fact that $\chi_{\bA} \in H^k_{\de}$ then follow from the bounds in \eqref{init.u.bd}, Lemmas~\ref{lm.constraint} and \ref{lem:data.N.beta} using Propositions~\ref{holder} and \ref{produit}. Moreover, by the estimates in \eqref{init.u.bd}, Lemmas~\ref{lm.constraint} and \ref{lem:data.N.beta}, we have \eqref{chi.init.bd}.
	
	\end{proof}
		
	We conclude this section with the following corollary, which summarizes the estimates in this section:
		\begin{cor}\label{cor.data}
		Given a free initial data set satisfying the assumptions of Theorem \ref{lwp}. Suppose that $\ep$ is sufficiently small, then 
		\begin{itemize}
		\item there exists an initial data set to the reduced system \eqref{tau}--\eqref{chi} such that the constraint equations \eqref{mom2} and \eqref{ham2} are satisfied and $\tau\restriction_{\Sigma_0} = 0$. 
		\item Also, there exists a constant $C$ (depending on $C_{eik}$, $k$, $\de$, $R$) such that all the smallness estimates \eqref{est.constraints.0}, \eqref{est.constraints}, \eqref{estN.data}, \eqref{estbeta.data} hold with implicit constant $C$.
		\item	For the quantities associated to $u_{\bA}$, $L_{\bA}$ and $\chi_{\bA}$, the estimates \eqref{L.init.bd}, \eqref{L.init.bd.lower} and \eqref{chi.init.bd} hold.
		\item Moreover, there exists a constant $C_i$ (depending on $C_{high}$, in addition to $C_{eik}$, $k$, $\de$, $R$) such that the following estimates hold for the initial data to the reduced system \eqref{tau}--\eqref{chi}:
		\begin{align*}
		\|H\|_{H^{k+1}_{\de}} + \|\wht N \|_{H^{k+2}_{\de}} + \|\beta\|_{H^{k+2}_{\de'}} + \|\wht \gamma\|_{H^{k+2}_\de} + \|e_0 \wht \gamma\|_{H^{k+1}_{\de'}} + \|\rd\phi\|_{H^k} & \\
		+ \sup_{\bA} \left(\|F_{\bA}\|_{H^k} + \| e^{2\gamma} L_\bA^i + \overrightarrow{c_{\bA}}^i\|_{H^k_{\de''}} + \| N e^{\gamma} L_\bA^t + |\overrightarrow{c_{\bA}}|\|_{H^k_{\de''}} + \|\chi_{\bA}\|_{H^k_\de}\right) &\leq C_i.
		\end{align*}
		\end{itemize}
				\end{cor}

	\section{Solving the reduced system of equations}\label{sec.solve.reduced}	
	
	In this section, we solve the reduced system of equations that we introduced in Section~\ref{sec.reduced}. This will be done by an iteration method. The iteration scheme will be introduced in Section~\ref{sec.iterative}. In Section~\ref{sec:iteration.bdd}, we show that in appropriate norms, the iterates we define are uniformly bounded. Finally, in Section~\ref{sec:iteration.convergence}, we show the convergence of the iterates in appropriate norms, which imply the existence and uniqueness of solutions to the reduced system of equations.
	
	\subsection{Iteration scheme}\label{sec.iterative}
	
	From now on we only consider the case $k=3$. As we mentioned previously, larger $k$ can be treated in a similar manner, but would unnecessarily complicate the exposition.
	
We construct the sequence\footnote{Note that $\alp$ is a non-negative constant independent of $n$.}
$$(N^{(n)},\beta^{(n)},\tau^{(n)},H^{(n)},\gamma^{(n)}=-\alpha \chi(|x|){\log}(|x|)+\wht \gamma^{(n)}, \phi^{(n)}, L^{(n)}_{\bA}, F^{(n)}_{\bA}, \chi^{(n)}_{\bA})$$
 iteratively as follows: 
For $n=1, 2$, let $N^{(n)},\beta^{(n)},\tau^{(n)},H^{(n)},\gamma^{(n)}=-\alpha \chi(|x|){\log}(|x|)+\wht \gamma^{(n)}, \phi^{(n)}, L^{(n)}_{\bA}, F^{(n)}_{\bA}, \chi^{(n)}_{\bA}$ be time-independent, with initial data as in Section~\ref{sec.data.constraints}.
For $n \geq 2$, given the $n$-th iterate, the $(n+1)$-st iterate is then defined by solving the following system (Latin indices are raised and lowered with respect to $\de$ as before):
	
\begin{align}
\nonumber -\Delta N^{(n+1)}=&-e^{-2\gamma^{(n)}}N^{(n)}(|H^{(n)}|^2+\f 12 e^{4\gamma^{(n)}}(\tau^{(n)})^2)\\
\label{lapsen} &-\frac{2 e^{2\gamma^{(n)}}}{N^{(n)}}(e_0^{(n-1)} \phi^{(n)})^2-\sum_{\bA} e^{4\gamma^{(n)}} N^{(n)} (F^{(n)}_{\bA})^2 \de_{ij} (L^{(n)}_{\bA})^i (L^{(n)}_{\bA})^j,\\
\label{shiftn}(L \beta^{(n+1)})_{ij}= & 2N^{(n)}e^{-2\gamma^{(n)}}(H^{(n)})_{ij},\\
\nonumber -\frac{e_0^{(n)}\left(\frac{e_0^{(n)} \gamma^{(n+1)}}{N^{(n)}} \right)}{N^{(n)}}+\Delta \gamma^{(n+1)} = &-\frac{2 (e_0^{(n-1)} \gamma^{(n)})^2}{N^{(n)}N^{(n-1)}}+\f 12 (\tau^{(n)})^2e^{2\gamma^{(n)}}
-\frac{e^{2\gamma^{(n)}}}{2 N^{(n)}}e_0^{(n-1)} \left( \frac{e^{-2\gamma^{(n)}}}{N^{(n-1)}}div(\beta^{(n)})\right)\\
\label{gamman} & -\frac{\Delta N^{(n)}}{2 N^{(n)}}- \delta^{ij}\rd_i\phi^{(n)}\rd_j\phi^{(n)}-\f 12\sum_{\bA} e^{4\gamma^{(n)}} (F^{(n)}_{\bA})^2 \de_{ij} (L^{(n)}_{\bA})^i (L^{(n)}_{\bA})^j, \\
\label{taun} 
\tau^{(n+1)}= & \frac{e^{-2\gamma^{(n)}}}{N^{(n)}}\left( -2e_0^{(n-1)} \gamma^{(n)} +div(\beta^{(n)})\right),\\
\nonumber (e_0^{(n+1)}) (H^{(n+1)})_{ij}= &
-2e^{2\gamma^{(n)}}N^{(n)}(H^{(n)})_i{ }^\ell (H^{(n)})_{j\ell} +2 \partial_{(j} (\beta^{(n)})^k (H^{(n)})_{i)k}\\
\nonumber & - (\partial_i \breve{\otimes } \partial_j)  N^{(n)} 
+\partial_i \gamma^{(n)} \breve{\otimes } \partial_j N^{(n)}
-2 N^{(n)}\partial_i \phi^{(n)} \breve{\otimes}\partial_j \phi^{(n)} \\
\label{hijn} & -\sum_{\bA} N^{(n)} (F^{(n)}_{\bA})^2  (L^{(n)}_{{\bA}})_i\breve{\otimes}(L^{(n)}_{\bA})_j,\\
\label{phin}\Box_{g^{(n)}} \phi^{(n+1)} = & 0,\\
\label{un} (L^{(n)}_{\bA})^\rho \rd_\rho (L^{(n+1)}_\bA)^{\alp} = & -(\Gamma^{(n)})^{\alp}_{\mu\nu} (L^{(n)}_{\bA})^\mu (L^{(n)}_{\bA})^\nu, \\
\label{Fn} 2(L_{\bA}^{(n)})^\rho\partial_\rho F^{(n+1)}_{\bA}
=& -(\chi^{(n)})_{\bA} F^{(n)}_{\bA},\quad\forall \bA,\\
\label{chin} (L_{\bA}^{(n)})^\rho\partial_\rho \chi^{(n+1)}_{\bA} + (\chi^{(n)}_{\bA})^2 =& -2((L^{(n)}_{\bA})^\rho \partial_\rho \phi^{(n)})^2
-\sum_{\bf B} F_{\bf B}^2(g^{(n)}_{\mu\nu} (L^{(n)}_{\bA})^\mu (L^{(n)}_{\bf B})^\nu)^2,\quad\forall \bA,
\end{align}	
where $g^{(n)}=-(N^{(n)})^2 dt^2+e^{2\gamma^{(n)}} \de_{ij} (dx^i +(\bt^{(n)})^i dt)(dx^j +(\bt^{(n)})^j dt)$; $D^{(n)}$, $(\Gamma^{(n)})^{\alp}_{\mu\nu}$ and $\Box_{g^{(n)}}$ are the Levi--Civita connection, Christoffel symbols and the Laplace--Beltrami operator, respectively, associated to $g^{(n)}$; and $e_0^{(n)}= \rd_t -(\bt^{(n)})^i\rd_i$. We have also used the notation
 $u_i \breve{\otimes} v_j = u_i v_j+u_j v_i-\delta_{ij}(u^k v_k).$

\begin{rk}[Well-posedness of \eqref{lapsen}--\eqref{chin}]
Notice that \eqref{lapsen}--\eqref{chin} is \underline{not} a linear system due to the term $e_0^{(n+1)} H^{(n+1)}$ on the LHS of \eqref{hijn}, which has a nonlinear term $(\bt^{(n+1)})^k\rd_k H_{ij}^{(n+1)}$ in the $(n+1)$-st iterate. (This will be useful in exploiting the nonlinear structure to prove estimates.) Nevertheless, the (local) well-posedness of the system \eqref{lapsen}--\eqref{chin} follows from the estimates we are about to prove.
\end{rk}

\subsection{Uniform boundedness of the sequence}\label{sec:iteration.bdd}

The first order of business is to show inductively that the sequence we just defined is uniformly bounded in appropriate function spaces. To carry out the induction, we assume as induction hypothesis that the following estimates for some $n \geq 2$ and for all $t\in [0,T]$. Here, $A_0 \ll A_1 \ll A_2 $ are all sufficiently large constants (independent of $\ep$) to be chosen later, $\de''' = \de''-\ep$, $\delta''=\de'-\ep$ and $\de'=\de-\ep$. Choosing $\ep$ smaller if necessary, we assume throughout that $-1<\delta'''$.
\begin{itemize}
\item (Estimates for $N^{(n)}$) $N^{(n)}$ admits a decomposition $N^{(n)} = 1+N^{(n)}_{asymp} (t)\chi(|x|){\log}(|x|)+ \wht N^{(n)}$ with $N^{(n)}_{asymp}\geq 0$ and satisfy the estimates
\begin{align}
\label{estN.small} |N_{asymp}^{(n)}| + \|\wht N^{(n)}\|_{H^2_{\de}} + \|\wht N^{(n)}\|_{W^2_{\de+\f 12, 4}} +\|\wht N^{(n)}\|_{C^1_{\de+1}} \leq & \ep,\\
\label{estN.large} |\partial_t N^{(n)}_{asymp}|+ \| \wht N^{(n)}\|_{H^{5}_\delta}+\|\partial_t \wht N^{(n)}\|_{H^{2}_\delta }\leq & 2C_i,\\
\label{estrdtN.large} \|\partial_t \wht N^{(n)}\|_{H^{3}_\delta }\leq & 2 C_i^2.
\end{align}
Assume that the same holds with $(n)$ replaced by $(n-1)$.
\item (Estimates for $\bt^{(n)}$) $\bt^{(n)}$ satisfies the following estimates:
\begin{align}
\label{estbeta.small} & \|\beta^{(n)}\|_{H^2_{\delta'}} + \|\beta^{(n)}\|_{W^2_{\delta'+\f 12, 4}} + \|\beta^{(n)}\|_{C^1_{\delta'+1}}\leq \ep, \\
\label{estbeta.large} &\|\beta^{(n)}\|_{H^{4}_{\delta'}}\leq A_0 C_i,\quad \|e_0^{(n-1)} \beta^{(n)}\|_{H^{2}_{\delta'}}\leq A_0 C_i,\quad \|e_0^{(n-1)} \bt^{(n)}\|_{H^4_{\de'}}\leq A_0 C_i^2.
\end{align}
Assume also that the above estimates for $\beta^{(n)}$, but not necessarily that for $e_0^{(n-1)} \bt^{(n)}$, also hold with $(n)$ replaced by $(n-1)$.
\item (Estimates for $\wht \gamma^{(n)}$) $\gamma^{(n)}$ admits a decomposition $\gamma^{(n)} = -\alp \chi(|x|) {\log}(|x|) + \wht \gamma^{(n)}$, where $\alp$ is as in Lemma~\ref{lm.constraint} (with $0\leq \alp \leq C\ep^2$) and $\wht \gamma^{(n)}$ satisfies
\begin{align}
\label{estgamma} 
\sum_{|\alpha|\leq 3}\left\|\frac{e_0^{(n-1)} \nabla^{\alpha} \wht \gamma^{(n)}}{N^{(n-1)}}\right\|_{L^2_{\delta'+1+|\alpha|}}
+\left\|\nabla \wht \gamma^{(n)}\right\|_{H^3_{\delta'+1}}
\leq & 4 C_i ,\\
\label{estrdtgamma}\left\|\partial_t \frac{(e_0^{(n-1)})\wht \gamma^{(n)}}{N^{(n-1)}}\right\|_{H^1_{\delta'+1}} \leq A_1 C_i, \quad
\left\|\partial_t \frac{(e_0^{(n-1)})\wht \gamma^{(n)}}{N^{(n-1)}}\right\|_{H^2_{\delta'+1}} \leq & A_1 C_i^2.
\end{align}
Assume that the same holds with $(n)$ replaced by $(n-1)$.
\item (Estimates for $\tau^{(n)}$) $\tau^{(n)}$ satisfies the following estimates:
\begin{align}
\label{esttau} & 
\|\tau^{(n)}\|_{H^3_{\delta''+1}}\leq A_1 C_i,\quad \|\partial_t\tau^{(n)}\|_{H^1_{\delta''+1}}\leq A_2 C_i,\quad \|\partial_t\tau^{(n)}\|_{H^2_{\delta''+1}}\leq A_2 C_i^2.
\end{align}
\item (Estimates for $H^{(n)}$) $H^{(n)}$ satisfies the following estimates:
\begin{align}
\label{estH} \|H^{(n)}\|_{H^{3}_{\delta+1}}\leq 2 C_i,\quad\|e_0^{(n)} H^{(n)}\|_{H^3_{\delta+1}} \leq 20 C_i.
\end{align}
\item (Estimates for the vector fields $L_{\bA}^{(n)}$ and auxiliary functions $\chi_\bA^{(n)}$)
Let $L_{\bA}^{(n)}$ be decomposed with respect to $\{\rd_t, \rd_i\}$, i.e., $L_{\bA}^{(n)} = (L_{\bA}^{(n)})^t \rd_t + (L_{\bA}^{(n)})^i \rd_i$. Then $(L_{\bA}^{(n)})^t$ obeys the lower bound
\begin{align}
\label{estu.lower} \min_{\bA}\inf_{x\in \mathbb R^2} \left|N^{(n-1)} e^{\gamma^{(n-1)}} (L_{\bA}^{(n)})^t \right|(x) \geq \f 12 C_{eik}^{-1}.
\end{align}
for $C_{eik}$ as in \eqref{init.u.bd}, and $L_{\bA}^{(n)}$ satisfies the following estimates:
\begin{align}
\label{estu.small} \sup_{\bA}\left(\left\|e^{2\gamma^{(n-1)}} (L^{(n)}_{\bA})^i+\overrightarrow{c_{\bA}}^i \right\|_{H^2_{\delta''}} + \left\|N^{(n-1)}e^{\gamma^{(n-1)}} (L^{(n)}_{\bA})^{t} - |\overrightarrow{c_{\bA}}| \right\|_{H^2_{\delta''}}\right) & \leq A_0 C_{eik},\\
\nonumber \sup_{\bA}\left(\left\|e^{2\gamma^{(n-1)}} (L^{(n)}_{\bA})^i+\overrightarrow{c_{\bA}}^i \right\|_{H^3_{\delta''}} + \left\|\partial_t \left(e^{2\gamma^{(n-1)}} (L^{(n)}_{\bA})^i\right) \right\|_{H^2_{\delta'''}} \right)&\\ 
\label{estu}  + \sup_{\bA}\left(\left\|N^{(n-1)}e^{\gamma^{(n-1)}} (L^{(n)}_{\bA})^{{t}} - |\overrightarrow{c_{\bA}}| \right\|_{H^3_{\delta''}} + \left\|\partial_t \left(N^{(n-1)}e^{\gamma^{(n-1)}} (L^{(n)}_{\bA})^t\right) \right\|_{H^2_{\delta'''}} \right) & \leq A_1 C_i.
\end{align}
Also, for $C_\chi$ as in Lemma~\ref{lem:data.chi}, $\chi^{(n)}_\bA$ satisfies the following estimates:
\begin{equation}
\label{estchi} \sup_{\bA}\left\|\chi^{(n)}_{\bA}\right\|_{C^0_{\de'+1}} \leq 2 C_{\chi}, \quad \sup_{\bA}\left\|\chi^{(n)}_{\bA}\right\|_{H^{3}_{\de}} \leq A_0 C_i.
\end{equation}
\item (Estimates for the matter fields) $\phi^{(n)}$ and $F_{\bA}^{(n)}$ are compactly supported in 
$$\{(t,x)\in [0,T]\times \mathbb R^2: C_s(1+R^\ep) t-|x|\geq -R\}, $$
where $C_s>0$ is to be chosen in Lemmas~\ref{lm:cpt.supp} and \ref{lm:cpt.supp.2}. Choosing $T$ smaller if necessary, we assume the above set $\subset \{ (t,x)\in [0,T]\times \mathbb R^2: |x|\leq 2 R\}$.

Moreover, the following estimates hold:
\begin{align}
\label{estphi}
\|\partial \phi^{(n)}\|_{H^3}+\left\|\partial_t \frac{(e_0^{(n-1)})\phi^{(n)}}{N^{(n-1)}}\right\|_{H^{2}}\leq & A_0 C_i ,\\
\label{estF} \sup_{\bA}\left\|F^{(n)}_{\bA}\right\|_{H^3}\leq A_0 C_i,\quad \sup_{\bA}\left\|\partial_t F^{(n)}_{\bA}\right\|_{H^{2}} \leq & A_1 C_i.
\end{align}
\end{itemize}

\begin{rk}[Choice of constants]
Recalling the statement of Theorem \ref{lwp}, $C_{high}$ is a potentially large constant such that $T$ can depend on $C_{high}$ but $\ep_{low}$ has to be \underline{in}dependent of $C_{high}$. In the previous section, we have proven that there is a $C_i$ \emph{depending on $C_{high}$} so that the bounds in Corollary \ref{cor.data} hold. Therefore, in the following $\ep_{low}$ and $T$ are chosen according to the following rules:
\begin{itemize}
\item $\ep_{low}$ (and therefore $\ep$) can be chosen to be small depending on $\de$, $R$, $A_0$, $A_1$, $A_2$ and $C_{eik}$, but not $C_i$.
\item The time parameter $T$ can be chosen to be small depending on all of $\de$, $R$, $C_i$, $C_{eik}$, $A_0$, $A_1$, $A_2$ and $\ep^{-1}$.
\end{itemize}
\end{rk}

In the remainder of this subsection, \textbf{$C$ will denote numerical constant, independent of $A$ and $C_i$, but can depend on\footnote{Recall that for this subsection, we have fixed $k=3$. Hence, none of the constants depend on $k$.} $C_{eik}$, $\delta$ and $R$. Similarly, we use the convention $\lesssim$ when the implicit constant is independent of $A_0$, $A_1$, $A_2$ and $C_i$. Constants that depend on $A_0$, $A_1$, $A_2$ or $C_i$ (in addition to $C_{eik}$, $\delta$ and $R$) will be written explicitly as $C(A_0)$, $C(A_1)$, $C(A_2)$ or $C(C_i)$.}

With the above conventions in mind, we note that in the rest of this subsection, we will have the following \emph{hierarchy of constants}:
$$C(A_0)\ll A_1,\quad C(A_1)\ll A_2;$$
while $\ep$ is much smaller than $C_{eik}$, $A_0$, $A_1$, $A_2$ so that for any $\eta>\f 1{100}$,
$$\ep^\eta C(C_{eik})\ll 1\quad  \ep^\eta C(A_2)\ll 1.$$
Notice however that since $\ep$ has to be chosen independent of $C_i$, $\ep^\eta C_i$ can\underline{not} be considered as a small constant.

It is easy to check using estimates in Section~\ref{sec.data.constraints} that the estimates \eqref{estN.small}--\eqref{estF} hold for the base case $n=2$. Our goal now is to prove the analogue of the estimates \eqref{estN.small}--\eqref{estF} with $(n)$ replaced by $(n+1)$ (and $(n-1)$ replaced by $(n)$). For most of these, we will in fact show that they hold with better constants on the RHS.

We begin with a propagation of smallness result, which states that for $T$ sufficiently small, the smallness of the data in the low norms can be propagated. Since we need to propagate smallness in $L^4$- and $L^\infty$-type spaces as well as $L^2$-type spaces, it is convenient to achieve this directly using the smallness of initial data, the boundedness of the time derivatives and the the smallness of the time interval.

\begin{prp}[Propagation of smallness]
	\label{prpsmallness}
	{The following estimates hold} for $T$ sufficiently small (depending on $C_i$)
	\begin{align*}
	\|\partial_t \phi^{(n)}\|_{L^\infty}+\|\nabla \phi^{(n)}\|_{L^\infty}+\|F^{(n)}_{\bA}\|_{L^\infty} \leq & C\ep,\\
	\|H^{(n)}\|_{H^1_{\delta+1}}+ \|H^{(n)}\|_{W^1_{\de+\f 32,4}}	+\|H^{(n)}\|_{C^0_{\delta+2}} 
	\leq & C\ep^2, \\
	\|\wht \gamma^{(n)} \|_{H^2_{\delta'}} + \|\wht\gamma^{(n)}\|_{W^2_{\de'+\f 12,4}} + \|\wht \gamma^{(n)} \|_{C^1_{\delta'+1}} \leq & C\ep^2, \\
	\left\|\f{e_0^{(n-1)}\wht \gamma^{(n)}}{N^{(n-1)}}\right\|_{H^1_{\delta'+1}}  + \left\|\f{e_0^{(n-1)} \wht\gamma^{(n)}}{N^{(n-1)}}\right\|_{W^1_{\de'+\f 32,4}} + \left\|\f{e_0^{(n-1)}\wht \gamma^{(n)}}{N^{(n-1)}}\right\|_{C^0_{\delta'+2}} \leq & C\ep^2{,}\\
	\|\tau^{(n)} \|_{H^{1}_{\delta''+1}} + \left\|\tau^{(n)}\right\|_{W^1_{\de''+\f 32,4}} + \left\|\tau^{(n)} \right\|_{C^0_{\delta''+2}} \leq & C\ep^2.
	\end{align*}
	
\end{prp}

\begin{proof}
	By \eqref{est.constraints.0}, \eqref{est.constraints}, Lemma~\ref{lem:e0gamma} and the fact $\tau^{(n)}\restriction_{\Sigma_0} = 0$ (Lemma~\ref{lem:e0gamma}), all these quantities initially satisfy the desired smallness estimates. The conclusion thus follows from the fact that the $\rd_t$ derivatives of all these terms in the relevant norms are bounded by a constant depending on $A_0$, $A_1$, $A_2$ and $C_i$, which is a consequence of the weighted-$L^2$ estimates in \eqref{estN.large}, \eqref{estbeta.large}, \eqref{estgamma}, \eqref{estrdtgamma}, \eqref{esttau}, \eqref{estH}, \eqref{estphi} and \eqref{estF}, together with the Sobolev embedding results in Proposition~\ref{holder}. (Notice that in applying the above estimates to obtain bounds for the $\rd_t$ derivatives, we often need to write $\rd_t = e_0^{(n)} + (\bt^{(n)})^i\rd_i$ or $\rd_t = e_0^{(n-1)} + (\bt^{(n-1)})^i\rd_i$ and estimate $\bt^{(n)}$ and $\bt^{(n-1)}$ using \eqref{estbeta.small}.) 
	
Therefore, the result follows from using calculus inequality of the type
$$\sup_{t\in [0,T]} \|f\|_{W^s_{\eta, p}}(t) \leq C\left(\|f\|_{W^s_{\eta, p}}(0) + \int_0^T \|\rd_t f\|_{W^s_{\eta, p}}(t')\, dt'\right)$$	
and choosing $T$ to be sufficiently small.\qedhere
\end{proof}

\begin{prp}[Estimate for $e_0^{(n-1)}\wht \gamma^{(n)}$]\label{prpeogamma}
	{The following estimate holds:}
\begin{equation} \label{estgammabis}\left\| \frac{e_0^{(n-1)} \wht \gamma^{(n)}}{N^{(n-1)}}\right\|_{H^3_{\delta'+1}} \leq 5C_i.
	\end{equation}
\end{prp}
\begin{proof}
{In view of \eqref{estgamma}, the proof of this proposition amounts to commuting $e_0^{(n-1)}$ and $\nab$. Using $[e_0^{(n-1)},\rd_i] = (\rd_i\bt^j)\rd_j$, w}e have
\begin{equation}\label{h.bdry}
\begin{split}
&\left\|\f{e_0^{(n-1)} \wht \gamma^{(n)}}{N^{(n-1)}}\right\|_{H^3_{\de'+1}} \\
\leq & \sum_{|\alp|\leq 3} \left\|\f{e_0^{(n-1)} \nab^{\alp} \wht \gamma^{(n)}}{N^{(n-1)}}\right\|_{L^2_{\de'+1+|\alp|}} + C\sum_{\substack{|\alp_1|+|\alp_2|+i \leq 2}} \left\|\nab^{\alp_1} (\nab \log N^{(n-1)})^{i+1} \nab^{\alp_2} \f{e_0^{(n-1)} \wht \gamma^{(n)}}{N^{(n-1)}}\right\|_{L^2_{\de'+1+|\alp_1|+|\alp_2|}} \\
& + C\sum_{|\alp_1|+|\alp_2|+|\alp_3|+i \leq 2} \left\| \f{\nab^{\alp_1}\nab \beta^{(n-1)} \nab^{\alp_2}(\nab\log N^{(n-1)})^i \nab^{\alp_3}\nab \wht \gamma^{(n)}}{N^{(n-1)}}\right\|_{L^2_{\de'+1+|\alp_1|+|\alp_2|}}.
\end{split}
\end{equation}
{Here, we have used the convention that $\nab^{\alp} (\nab \log N^{(n)})^i$ denotes a product of $i$ factors, each of which is some spatial derivatives of $\nab\log N^{(n)}$ and the total number of derivatives is $|\alp|$.}
Now using H\"older's inequality, Lemma~\ref{der}, Proposition~\ref{holder}, \eqref{estN.small}, \eqref{estN.large} and Proposition \ref{prpsmallness},
\begin{equation}\label{h.bdry.lo.1}
\begin{split}
& \sum_{|\alp_1|+|\alp_2| +i \leq 2} \left\|\nab^{\alp_1} (\nab \log N^{(n-1)})^{i+1} \nab^{\alp_2} \f{e_0^{(n-1)} {\wht \gamma^{(n)}}}{N^{(n)}}\right\|_{L^2_{\de'+1+|\alp_1|+|\alp_2|}} \\
\ls & \left\|\nab \log N^{(n-1)}\right\|_{C^0_{\de+1}} \left\|\f{e_0^{(n-1)} \wht \gamma^{(n)}}{N^{(n-1)}}\right\|_{H^{2}_{\de'+1}} + \left\|\nab \log N^{(n-1)}\right\|_{W^1_{\de+\f 32,4}} \left\|\f{e_0^{(n-1)} \wht \gamma^{(n)}}{N^{(n-1)}}\right\|_{W^{2}_{\de'+\f 32,4}}\\ 
	&+\left\|\nab \log N^{(n-1)}\right\|_{H^2_{\de+1}}\left\|\f{e_0^{(n-1)} \wht \gamma^{(n)}}{N^{(n-1)}}\right\|_{C^0_{\de'+2}}\\
	\ls & \ep \left\|\f{e_0^{(n)} \wht \gamma^{(n)}}{N^{(n-1)}}\right\|_{H^{2}_{\de'+1}}+ \ep C_i.
\end{split}
\end{equation}
Similarly, but using \eqref{estbeta.small} and \eqref{estbeta.large} in addition to {\eqref{estN.small}}, and also the fact that $\f{1}{N^{(n)}}\ls 1$, we have
\begin{equation}\label{h.bdry.lo.2}
\sum_{|\alp_1|+|\alp_2|+|\alp_3|+i \leq 2} \left\| \f{\nab^{\alp_1}\nab \beta^{(n-1)} \nab^{\alp_2}(\nab\log N^{(n-1)})^i \nab^{\alp_3}\nab \wht \gamma^{(n)}}{N^{(n-1)}}\right\|_{L^2_{\de'+1+|\alp_1|+|\alp_2|}} \ls \ep \left\|\nab \wht \gamma^{(n)} \right\|_{H^{3}_{\de'+1}}+\ep A_0C_i.
\end{equation}
Plugging \eqref{h.bdry.lo.1} and \eqref{h.bdry.lo.2} into \eqref{h.bdry}, we obtain 
\begin{equation}\label{h.bdry.comp.1}
\begin{split}
&\left\|\f{e_0^{(n-1)} \wht \gamma^{(n)}}{N^{(n-1)}}\right\|_{H^3_{\de'+1}} \\
\leq &\sum_{|\alp|\leq 3} \left\|\f{e_0^{(n-1)} \nab^{\alp} \wht \gamma^{(n)}}{N^{(n-1)}}\right\|_{L^2_{\de'+1+|\alp|}} + C\ep \left(\left\|\f{e_0^{(n-1)} \wht \gamma^{(n)}}{N^{(n-1)}}\right\|_{H^{k}_{\de'+1}} + \left\|\nab \wht \gamma^{(n)} \right\|_{H^{3}_{\de'+1}}\right)(t) + \ep C A_0C_i.
\end{split}
\end{equation}
Consequently, choosing $\ep$ sufficiently small, {and using \eqref{estgamma},} we conclude the proof of the Proposition.
\end{proof}

\begin{prp}[Estimates for $N^{(n+1)}$]\label{prpN}
	For $n \geq 2$, $N^{(n+1)}$ admits a decomposition 
	$$N^{(n+1)}=1+N_{asymp}^{(n+1)}\chi(|x|){\log}(|x|)+ \wht N^{(n+1)},$$ 
	with $N_{asymp}^{(n+1)}\geq 0$, and such that the following bounds are satisfied:
	\begin{align}
	|N_{asymp}^{(n+1)}| + \| \wht N^{(n+1)}\|_{H^{2}_\delta} + \| \wht N^{(n+1)}\|_{W^{2}_{\delta+\f 12, 4}} +\|\wht N^{(n+1)}\|_{C^1_{\delta+1}} \ls & \ep^2,\label{N.lower}\\
	|\rd_t N^{(n+1)}_{asymp}| + \|\wht N^{(n+1)}\|_{H^5_\delta} + \|\partial_t \wht N^{(n+1)}\|_{H^{2}_\delta }\ls & \ep C(A_2) C_i, \label{N.higher}\\
	\|\partial_t \wht N^{(n+1)}\|_{H^{3}_\delta} \ls & \ep C(A_1) C_i^2. \label{N.top}
	\end{align}
	\end{prp}
\begin{proof}
\textbf{Existence of decomposition and proof of \eqref{N.lower}.}	We claim that
	\begin{equation}\label{prpN.1}
	\left\|\mbox{RHS of \eqref{lapsen}}\right\|_{L^2_{\delta+2}} + \left\|\mbox{RHS of \eqref{lapsen}}\right\|_{L^4_{\delta+\f 52}} \leq C \ep^2.
	\end{equation}
	Except for\footnote{Note that in \eqref{N.data}, $\tau$ does not appear on the RHS.} the term $e^{2\gamma^{(n)}}N^{(n)}(\tau^{(n)})^2$, all the other terms can be estimated in an identical manner as in Lemma~\ref{lem:data.N.beta}, except that we estimate the terms using Proposition~\ref{prpsmallness} instead of using the assumptions on the reduced data and the estimates in Lemma~\ref{lm.constraint}.
	
	It therefore remains to control $e^{2\gamma^{(n)}}N^{(n)}(\tau^{(n)})^2$. For this we note that, for $\ep$ sufficiently small, by Proposition~\ref{prpsmallness}, 
	$$\|(\tau^{(n)})^2\|_{L^2_{\de+2+\ep}} \ls \|\tau^{(n)}\|_{L^2_{\de''+1}}\|\tau^{(n)}\|_{C^0_{\de''+2}}\ls \ep^4.$$
	Now by Lemma~\ref{lm.constraint} (for $\alp$), Proposition~\ref{prpsmallness} (for $\wht \gamma^{(n)}$) and \eqref{estN.small} (for $N^{(n)}_{asymp}$ and $\wht N^{(n)}$), for $\ep$ sufficiently small, $e^{2\gamma^{(n)}}N^{(n)}$ grows at worst as $|x|^{C\ep^2}$ for large $|x|$ and $\|e^{2\gamma^{(n)}}N^{(n)}\|_{C^0_\ep} \ls 1$. This proves that
	$$\|e^{2\gamma^{(n)}}N^{(n)}(\tau^{(n)})^2\|_{L^2_{\de+2}} \ls \ep^4.$$
 An essentially identical argument also shows
  $$\|e^{2\gamma^{(n)}}N^{(n)}(\tau^{(n)})^2\|_{L^4_{\de+\f 52}} \ls \ep^4.$$
	This proves the claim. Applying Theorem~\ref{laplacien} and Corollary~\ref{coro} (to $\Delta (N^{(n+1)}-1)$) yields the existence of the decomposition of $N^{(n+1)}$, as well as the estimate \eqref{N.lower}.

\textbf{Proof of first part of \eqref{N.higher}.} To obtain the $H^5_\delta$ bound for $\wht {N}^{(n)}$ (first part of \eqref{N.higher}), we need to control the RHS of \eqref{lapsen} in $H^3_{\delta+2}$. We note that it is easy to obtain \emph{some} bound in $H^3_{\delta+2}$. The key point here, however, is that the bound must be at worst linear in $C_i$, with an $\ep$ smallness constant.

We first bound the term $e^{-2\gamma^{(n)}}N^{(n)}|H^{(n)}|^2$ in $H^3_{\de+2}$. There are various cases: in order to shorten the exposition, let us use the notation $(a,b,c,d)$ (with $c\leq d$) to denote the case with at most $a$ derivatives on $\gamma^{(n)}$, at most $b$ derivatives on $N^{(n-1)}$, at most $c$ and $d$ derivatives on the two factors of $H^{(n)}$. The following cases, though not mutually exclusive, exhaust all possibilities:
\begin{itemize}
\item $(1,1,0,3)$. By H\"older's inequality, \eqref{est.constraints.0}, \eqref{estN.small}, \eqref{estH} and Proposition~\ref{prpsmallness},
$$ \ls \| H^{(n)}\|_{C^0_{\de+2}} \|H^{(n)}\|_{H^3_{\de+1}} \leq C_i \ep^2.$$
\item $(1,3,0,0)$. By H\"older's inequality, \eqref{est.constraints.0}, \eqref{estN.small}, \eqref{estN.large} and Proposition~\ref{prpsmallness},
$$ \ls  \left(1+|N_{asymp}^{(n)}|+\|\wht N^{(n)} \|_{H^3_{\de}}\right) \|H^{(n)}\|_{C^0_{\de+2}}^2\leq C_i \ep^4.$$
\item $(3,1,0,0)$. By H\"older's inequality, \eqref{est.constraints.0}, \eqref{estN.small}, \eqref{estgamma}, Proposition~\ref{prpsmallness} and Lemma~\ref{der},
$$ \ls \left(1 + |\alp| + \|\nab \wht \gamma^{(n)}\|_{H^2_{\de'+1}}\right) \|H^{(n)}\|_{C^0_{\de+2}}^2 \ls C_i \ep^4.$$
\item $(1,1,1,2)$. By H\"older's inequality, \eqref{est.constraints.0}, \eqref{estN.small}, \eqref{estH}, Propositions~\ref{prpsmallness} and \ref{holder},
$$ 
\ls \|H^{(n)}\|_{W^1_{\de+\f 32,4}}\|H^{(n)}\|_{W^2_{\de+\f 32,4}}\ls \|H^{(n)}\|_{W^1_{\de+\f 32,4}}\|H^{(n)}\|_{H^3_{\de+2}} \ls C_i \ep^2.$$
\item $(2,0,0,1)$. By H\"older's inequality, \eqref{est.constraints.0}, \eqref{estN.small} and Proposition~\ref{prpsmallness},
$$\ls \left(1 + |\alp| + \|\wht\gamma^{(n)} \|_{W^2_{\de'+\f 12, 4}}\right) \|H^{(n)}\|_{C^0_{\de+2}} \|H^{(n)}\|_{W^1_{\de+\f 32, 4}}\ls \ep^4.$$
\item $(0,2,0,1)$. By H\"older's inequality, \eqref{est.constraints.0}, \eqref{estN.small} and Proposition~\ref{prpsmallness},
$$\ls \left(1 + |N^{(n)}_{asymp}| + \|\wht N^{(n)}\|_{W^2_{\de+\f 12, 4}}\right) \|H^{(n)}\|_{C^0_{\de+2}} \|H^{(n)}\|_{W^1_{\de+\f 32, 4}}\ls \ep^4.$$
\end{itemize}

The term $e^{2\gamma^{(n)}}N^{(n)}(\tau^{(n)})^2$ can be treated in a similar fashion, since $\tau^{(n)}$ and $H^{(n)}$ (according to \eqref{esttau}, \eqref{estH} and Proposition~\ref{prpsmallness}) obey similar estimates\footnote{Notice that this comparison is only true for $\tau^{(n)}$ and $H^{(n)}$ without $\rd_t$ derivatives, which is what we are concerned about for this estimate.} except for a slight difference of weights ($\de'$ compared to $\de$) and constants ($A_1$ compared to $2$). Since this term is at least quadratic in $\tau^{(n)}$ (and its derivatives), there is plenty of room to handle the weights. We give the estimate here and omit the straightforward proof:
$$\|e^{2\gamma^{(n)}}N^{(n)}(\tau^{(n)})^2\|_{H^3_{\de+2}} \ls C(A_1) C_i \ep^2 \ls C_i \ep.$$

We next discuss the scalar field term, $\frac{e^{2\gamma^{(n)}}}{N^{(n)}}(e^{(n-1)}_0 \phi^{(n)})^2$. Note that this term poses a different challenge in the sense that the smallness is at a much lower level (i.e., taking any derivative of $e_0^{(n-1)}\phi^{(n)}$ destroys the $\ep$-smallness). Nevertheless, it has the advantage that the term is \emph{compactly supported}, and we can use the product estimate in unweighted Sobolev spaces in Proposition~\ref{product} to obtain\footnote{To see that using Proposition~\ref{product} and the fact that $supp(e^{(n-1)}_0\phi^{(n)})\subset B(0,2R)$ indeed imply such an estimate where we only require the bounds for $\frac{e^{2\gamma^{(n)}}}{N^{(n-1)}}$ in $B(0,3R)$, we argue as follows: Let $\eta$ be a smooth cutoff function compactly supported in $B(0,3R)$ which is $\equiv 1$ in $B(0, 2R)$. Then
\begin{equation*}
\begin{split}
\left\|\frac{e^{2\gamma^{(n)}}}{N^{(n)}}(e^{(n-1)}_0 \phi^{(n)})^2\right\|_{H^{3}_{\delta+2}} \ls & \left\|(\eta\frac{e^{2\gamma^{(n)}}}{N^{(n-1)}})(e^{(n-1)}_0 \phi^{(n)})^2\right\|_{H^{3}} \\
\ls & \|(\eta\frac{e^{2\gamma^{(n)}}}{N^{(n)}})\|_{L^\infty} \|(e^{(n-1)}_0 \phi^{(n)})^2\|_{H^3} + \|(\eta\frac{e^{2\gamma^{(n)}}}{N^{(n-1)}})\|_{H^3} \|(e^{(n-1)}_0 \phi^{(n)})^2\|_{L^\infty}.
\end{split}
\end{equation*}
The support properties of $\eta$ thus imply the desired estimate.\label{cpt.supp.product}}
\begin{equation}\label{Nest.scalar}
\begin{split}
&\left\|\frac{e^{2\gamma^{(n)}}}{N^{(n)}}(e^{(n-1)}_0 \phi^{(n)})^2\right\|_{H^{3}_{\delta+2}} \\
\ls & \left\|\frac{e^{2\gamma^{(n)}}}{N^{(n)}}\right\|_{L^\infty(B(0,3R))} \|e^{(n-1)}_0 \phi^{(n)} \|_{H^{3}} \|e^{(n-1)}_0 \phi^{(n)} \|_{L^\infty} + \left\|\frac{e^{2\gamma^{(n)}}}{N^{(n)}}\right\|_{H^3(B(0,3R))} \|e^{(n-1)}_0 \phi^{(n)}\|_{L^\infty}^2\\
\ls & \ep C(A_0) C_i.
\end{split}
\end{equation}
Here we have used \eqref{estN.small}, \eqref{estN.large}, \eqref{estgamma}, \eqref{estphi}, Proposition~\ref{prpsmallness} and also \eqref{estbeta.small} and \eqref{estbeta.large} (to control the difference between $e_0^{(n-1)}$ and $\rd_t$).

A similar argument as \eqref{Nest.scalar} can be used to bound the term involving $(F_{\bA}^{(n)})^2$, using \eqref{estF}, \eqref{estu.small} and \eqref{estu} instead of \eqref{estphi}, since $F_{\bA}^{(n)}$ is also compactly supported, to get\footnote{We note again that the implicit constant in $\ls$ may depend on $C_{eik}$.}
\begin{equation*}
\begin{split}
\left\|\sum_{\bA} e^{4\gamma^{(n)}} N^{(n)} (F^{(n)}_{\bA})^2 \de_{ij} (L^{(n)}_{\bA})^i (L^{(n)}_{\bA})^j\right\|_{H^{3}_{\delta+2}} 
\ls & \ep C(A_0) C_i + \ep^2 C(A_1) C_i \ls \ep C(A_0) C_i.
\end{split}
\end{equation*}

Combining all the estimates above, we have $\|\mbox{(RHS of \eqref{lapsen})}\|_{H^3_{\de+2}}\ls \ep C(A_0) C_i$. By Theorem~\ref{laplacien}, we obtain $\|\wht N^{(n+1)}\|_{H^5_\delta}\leq \ep C_i$, which is the first part of \eqref{N.higher}.

\textbf{Proof of second part of \eqref{N.higher}.}
	We now turn to the estimate for $\partial_t N^{(n+1)}$, including both for $\rd_t N^{(n+1)}_{asymp}$ and $\rd_t \wht N^{(n+1)}$, in \eqref{N.higher}. Since RHS of \eqref{lapsen} is differentiable in $t$, it is easy to see that $\rd_t N^{(n+1)} =(\rd_t N^{(n+1)}_{asymp}) \chi(|x|){\log} (|x|) + \rd_t\wht N^{(n+1)}$ is the solution given by Corollary~\ref{coro} to the equation
	$$\Delta (\rd_t N^{(n+1)}) = \rd_t(\mbox{RHS of \eqref{lapsen}}).$$
	Therefore, to prove the second part of \eqref{N.higher}, it suffices (1) to bound the integral of $\rd_t(\mbox{RHS of \eqref{lapsen}})$ with respect to $dx$, and (2) to bound $\rd_t(\mbox{RHS of \eqref{lapsen}})$ in $L^2_{\de+2} = H^0_{\de+2}$. Noticing moreover that (by H\"older's inequality) $L^2_{\de+2}\subset L^1$ continuously, it therefore suffices to bound $\rd_t(\mbox{RHS of \eqref{lapsen}})$ in $H^0_{\de+2}$

Since the estimates for $\rd_t\tau^{(n)}$ are worse than those for $\rd_t H^{(n)}$, and those for $\tau^{(n)}$ and $H^{(n)}$ are similar (compare \eqref{esttau} and \eqref{estH}), we will treat the term $\rd_t\left(e^{2\gamma^{(n)}}N^{(n)}(\tau^{(n)})^2\right)$ and leave the (easier) term $\rd_t\left(e^{-2\gamma^{(n)}}N^{(n)}|H^{(n)}|^2\right)$ to the reader. For the term $\rd_t\left(e^{2\gamma^{(n)}}N^{(n)}(\tau^{(n)})^2\right)$, we can in fact bound it in the norm $H^1_{\de+2}$ (which is stronger than $H^0_{\de+2}$) as follows, using \eqref{estN.small}, \eqref{estN.large}, \eqref{esttau} and Proposition~\ref{prpsmallness}:
\begin{equation}\label{Nest.tau.top}
\begin{split}
&\|\rd_t\left(e^{2\gamma^{(n)}}N^{(n)}(\tau^{(n)})^2\right)\|_{H^1_{\de+2}} \\
\ls &\|e^{2\gamma^{(n)}}N^{(n)}\|_{C^1_{-\f{\ep}{10}}}\left(\|\tau^{(n)}\|_{C^0_{\de''+2}}\|\rd_t\tau^{(n)}\|_{H^1_{\de''+1}} + \|\tau^{(n)}\|_{W^1_{\de''+\f 32, 4}}\|\rd_t \tau^{(n)}\|_{W^0_{\de''+\f 32, 4}}\right)\\
& + \left\|\rd_t \left(e^{2\gamma^{(n)}}N^{(n)}\right)\right\|_{W^1_{-\f{\ep}{10}-\f 12, 4}} \|\tau^{(n)}\|_{W^1_{\de''+\f 32, 4}} \|\tau^{(n)}\|_{C^0_{\de''+2}} \\
\ls & \ep^2 A_2 C_i+ C_i A_1\ep^4\ls \ep^2 C(A_2) C_i\ls \ep C_i.
\end{split}
\end{equation}

We now turn to the compactly supported terms involving $(e_0^{(n-1)}\phi^{(n)})^2$ and $(F_{\bA}^{(n)})^2$. First, for $(e_0^{(n-1)}\phi^{(n)})^2$, by H\"older's inequality, the support properties of $e_0^{(n-1)}\phi^{(n)}$, \eqref{estN.small}, \eqref{estgamma}, \eqref{estphi} and Proposition{s}~\ref{prpsmallness}{, \ref{prpeogamma} and \ref{holder}}, we have
\begin{equation}
\begin{split}
&\left\|\rd_t\left(\frac{e^{2\gamma^{(n)}}}{N^{(n)}}(e_0^{(n-1)} \phi^{(n)})^2\right)\right\|_{H^0_{\de+2}} \\
\ls & \left\|\rd_t \frac{e^{2\gamma^{(n)}} (N^{(n-1)})^2}{N^{(n)}}\right\|_{L^\infty(B(0,2R))} \left\| \f{e_0^{(n-1)} \phi^{(n)}}{N^{(n-1)}} \right\|_{L^\infty}^2 \\
& + \left\|\frac{e^{2\gamma^{(n)}}(N^{(n-1)})^2}{N^{(n)}} \right\|_{L^\infty(B(0,2R))} \left\| \rd_t \f{e_0^{(n-1)} \phi^{(n)}}{N^{(n-1)}} \right\|_{L^\infty} \left\| e_0^{(n-1)} \phi^{(n)} \right\|_{L^\infty} \ls \ep C(A_0) C_i.
\end{split}
\end{equation}
The $(F_{\bA}^{(n)})^2$ term can be treated similarly using \eqref{estu.small}, \eqref{estu} and \eqref{estF} instead of \eqref{estphi}:
\begin{equation*}
\begin{split}
&\left\|\rd_t\left(\sum_{\bA} e^{4\gamma^{(n)}} N^{(n)} (F^{(n)}_{\bA})^2 \de_{ij} (L^{(n)}_{\bA})^i (L^{(n)}_{\bA})^j\right)\right\|_{H^0_{\de+2}} 
\ls \ep C(A_1) C_i.
\end{split}
\end{equation*}
Combining all these gives the estimates for $\rd_t N^{(n+1)}_{asymp}$ and $\rd_t \wht N^{(n+1)}$ in \eqref{N.higher}.

\textbf{Proof of \eqref{N.top}.} Finally, in order to prove \eqref{N.top}, we estimate $\rd_t(\mbox{RHS of \eqref{lapsen}})$ in $H^1_{\de+2}$. Now, in contrast to the second part of \eqref{N.higher}, we allow the estimates to be quadratic in $C_i$. First we note that the $\rd_t\left(e^{2\gamma^{(n)}}N^{(n)}(\tau^{(n)})^2\right)$ term has been estimated above in \eqref{Nest.tau.top}. The $\rd_t\left(e^{-2\gamma^{(n)}}N^{(n)}|H^{(n)}|^2\right)$ term, as we argued above, is similar.

It therefore remains to estimate the $(e_0^{(n-1)}\phi^{(n)})^2$ term and the $(F_\bA^{(n)})^2$ term. For the scalar field term, we have, using the support properties\footnote{We refer the reader to Footnote~\ref{cpt.supp.product} on p.\pageref{cpt.supp.product} regrading the use of Proposition~\ref{product} when one of the factors is compactly supported.} of $e_0^{(n-1)}\phi^{(n)}$, Proposition~\ref{product}, \eqref{estN.small}, \eqref{estN.large}, \eqref{estbeta.small}, \eqref{estgamma}, \eqref{estphi}{,} Proposition{s}~\ref{prpsmallness} {and \ref{prpeogamma}},
\begin{equation*}
\begin{split}
&\left\|\rd_t\left(\frac{e^{2\gamma^{(n)}}}{N^{(n)}}(e_0^{(n-1)} \phi^{(n)})^2\right)\right\|_{H^1_{\de+2}} \\
\ls & \left\|\rd_t \frac{e^{2\gamma^{(n)}} (N^{(n-1)})^2}{N^{(n)}}\right\|_{L^\infty(B(0,3R))} \left\| \f{e_0^{(n-1)} \phi^{(n)}}{N^{(n-1)}} \right\|_{L^\infty} \left\| \f{e_0^{(n-1)} \phi^{(n)}}{N^{(n-1)}} \right\|_{H^1} \\
& + \left\|\rd_t \frac{e^{2\gamma^{(n)}} (N^{(n-1)})^2}{N^{(n)}}\right\|_{H^1(B(0,3R))} \left\| \f{e_0^{(n-1)} \phi^{(n)}}{N^{(n-1)}} \right\|_{L^\infty} \left\| \f{e_0^{(n-1)} \phi^{(n)}}{N^{(n-1)}} \right\|_{L^\infty} \\
& + \left\|\frac{e^{2\gamma^{(n)}}(N^{(n-1)})^2}{N^{(n)}} \right\|_{L^\infty(B(0,3R))} \left\| \rd_t \f{e_0^{(n-1)} \phi^{(n)}}{N^{(n-1)}} \right\|_{H^1} \left\| \f{e_0^{(n-1)} \phi^{(n)}}{N^{(n-1)}} \right\|_{L^\infty} \\
& + \left\|\frac{e^{2\gamma^{(n)}}(N^{(n-1)})^2}{N^{(n)}} \right\|_{H^1(B(0,3R))}\left\| \rd_t \f{e_0^{(n-1)} \phi^{(n)}}{N^{(n-1)}} \right\|_{L^\infty} \left\| \f{e_0^{(n-1)} \phi^{(n)}}{N^{(n-1)}} \right\|_{L^\infty}\ls  \ep C(A_0) C_i^2.
\end{split}
\end{equation*}
Using \eqref{estu.small}, \eqref{estu} and \eqref{estF} instead of \eqref{estphi}, the $F_{\bA}^2$ term is similar, for which we have
\begin{equation*}
\begin{split}
&\left\|\rd_t\left(\sum_{\bA} e^{4\gamma^{(n)}} N^{(n)} (F^{(n)}_{\bA})^2 \de_{ij} (L^{(n)}_{\bA})^i (L^{(n)}_{\bA})^j\right)\right\|_{H^1_{\de+2}} 
\ls \ep C(A_1) C_i^2.
\end{split}
\end{equation*}
\end{proof}

\begin{prp}[Estimates for $\bt^{(n+1)}$]\label{prpbeta}
	For $n\geq 2$, the following estimates hold:
	\begin{align}
	\|\beta^{(n+1)}\|_{H^2_{\delta'}}+\|\beta^{(n+1)}\|_{W^2_{\delta'+\f 12,4}}+\|\beta^{(n+1)}\|_{C^1_{\delta'+1}} \ls & \ep^2,\label{betaest.low}\\
	\|\beta^{(n+1)}\|_{H^4_{\delta'}}\ls & C_i, \label{betaest.high}\\
	\|e_0^{(n)} \beta^{(n+1)}\|_{H^2_{\delta'}}\ls & C_i, \label{betaest.higher}\\
	\|e_0^{(n)} \beta^{(n+1)}\|_{H^4_{\delta'}}\ls & C_i^2.\label{betaest.top}
	\end{align}
\end{prp}

\begin{proof}
In view of Proposition~\ref{prpsmallness}, the existence of $\beta^{(n+1)}$ and the estimates \eqref{betaest.low} can be proven in exactly the same manner as Lemma~\ref{lem:data.N.beta}; we omit the details. We only focus on the proofs of \eqref{betaest.high}, \eqref{betaest.higher} and \eqref{betaest.top}.

\textbf{Proof of \eqref{betaest.high}.} To prove \eqref{betaest.high}, we take the divergence of \eqref{shiftn} (in a similar manner as in the proof of Lemma~\ref{lem:data.N.beta}) to get 
\begin{equation}\label{betan.poisson}
\Delta (\beta^{(n+1)})^i = 2\de^{i\ell} \de^{jk} \rd_k \left(N^{(n)}e^{-2\gamma^{(n)}}(H^{(n)})_{j\ell}\right).
\end{equation}
Note that the RHS obviously has $0$ mean and therefore by Theorem~\ref{laplacien}, in order to prove \eqref{betaest.high}, it suffices to bound the RHS of \eqref{betan.poisson} in $H^2_{\de'+2}$ by $C C_i$.

Let us note explicitly that in this estimate, the need to have a small loss in the weight (with $\de'=\de-\ep$ instead of $\de$) is due to\footnote{This is for instance in contrast to the proof of Proposition~\ref{prpN}, where because the corresponding RHS is more nonlinear, one can put $\wht N^{(n+1)}$ in a better weighted space.} the factor $N^{(n)}e^{-2\gamma^{(n)}}$, which grows at infinity. On the other hand, since $\alp$ and $|N_{asymp}^{(n+1)}|$ are small (by \eqref{est.constraints.0} and \eqref{estN.small}), $N^{(n)}e^{-2\gamma^{(n)}}$ grows at worst as $|x|^{\f{\ep}{10}}$ for large $|x|$, and this can indeed be handled by putting in $\de'$ in place of $\de$ in the estimates. More precisely, by Lemma~\ref{der}, H\"older's inequality, \eqref{est.constraints.0}, \eqref{estN.small}, \eqref{estN.large}, \eqref{estgamma}, \eqref{estH} and Proposition~\ref{prpsmallness},
\begin{equation}\label{betaest.RHS}
\begin{split}
& \left\| \rd_k \left(N^{(n)}e^{-2\gamma^{(n)}}(H^{(n)})_{j\ell}\right) \right\|_{H^2_{\de'+2}} \\
\ls & \| N^{(n)}e^{-2\gamma^{(n)}}(H^{(n)})_{j\ell}\|_{H^3_{\de'+1}}\\
\ls & (1 + \|\wht N^{(n)}\|_{W^2_{\de+\f 12,4}} + \|\wht \gamma^{(n)}\|_{W^2_{\de'+\f 12,4}})\|H^{(n)}\|_{H^3_{\delta+1}} + (\|\wht N^{(n)}\|_{H^3_{\de}} + \|\nab \wht \gamma^{(n)}\|_{H^2_{\de'+1}})\|H^{(n)}\|_{C^0_{\de+2}}\\
\ls & C_i + \ep^2 C_i \ls C_i.
\end{split}
\end{equation}

\textbf{Proof of \eqref{betaest.higher}.} For the estimate of $e_0^{(n)}\beta^{(n+1)}$, we take the divergence of \eqref{shiftn} and commute the resulting equation\footnote{To obtain this one needs to justify that $e_0^{(n)}(\beta^{(n+1)})^i$ is well-defined, but this follows from the fact that the RHS is differentiable by $e_0^{(n)}$; we omit the details.} with $e_0^{(n)}$ to obtain
\begin{equation}\label{e0.beta.eqn}
\Delta \left(e_0^{(n)}(\beta^{(n+1)})^i\right) = 2\de^{k\ell}\de^{ij}e_0^{(n)}\partial_k (e^{-2\gamma^{(n)}}N^{(n)} H^{(n)}_{\ell j}) + [\Delta, e_0^{(n)}] (\beta^{(n+1)})^i =: I + II.
\end{equation}
It is easy to check that the RHS of \eqref{e0.beta.eqn} in fact has mean zero (as is expected). As a consequence, we can apply Theorem~\ref{laplacien} so that in order to prove the estimate for $e_0^{(n)} \bt^{(n+1)}$ in \eqref{betaest.higher}, it suffices to bound the RHS of \eqref{e0.beta.eqn} in $L^2_{\de'+2} = H^0_{\de'+2}$ by $C_i$.

For $I$ in \eqref{e0.beta.eqn}, after commuting $[e_0^{(n)}, \rd_k]$ and using Lemma~\ref{der}, it easy to see that we only need to estimate the following terms:
\begin{equation}\label{rdtbetaest.RHS}
\begin{split}
\| I \|_{H^0_{\de'+2}} \ls & \|e^{-2\gamma^{(n)}} N^{(n)} (e_0^{(n)} H^{(n)})\|_{H^1_{\de'+1}} + \|e^{-2\gamma} (e_0^{(n)} \gamma^{(n)}) N^{(n)} H^{(n)}\|_{H^1_{\de'+1}} \\
& + \|e^{-2\gamma} (e_0^{(n)} N^{(n)}) H^{(n)}\|_{H^1_{\de'+1}} + \| \bt^{(n)}\|_{C^0_{\de'+1}} \|\rd_k (e^{-2\gamma^{(n)}} N^{(n)} H^{(n)})\|_{H^0_{\de'+1}}.
\end{split}
\end{equation}
Using \eqref{estbeta.small} and \eqref{betaest.RHS}, the last term is clearly $\ls \ep C_i$.

To proceed, notice that according to Proposition~\ref{prpsmallness}, $\wht N^{(n)}$, $\wht \gamma^{(n)}$ and $H^{(n)}$ are all $O(\ep^2)$ small in $L^4$-based norms up to $1$ derivatives, while according to \eqref{estN.large}, \eqref{estgamma}{,} \eqref{estH} {and Proposition~\ref{prpeogamma}} (and \eqref{estbeta.small}), $e_0^{(n)} \wht N^{(n)}$, $e_0^{(n)} \wht \gamma^{(n)}$ and $e_0^{(n)} H^{(n)}$ are $O(C_i)$ in appropriate weighted $H^2$ spaces. Hence, the first three terms in \eqref{rdtbetaest.RHS} can be treated in a similarly manner. We consider the first term as an example. By \eqref{est.constraints.0}, \eqref{estN.small}, \eqref{estN.large}, \eqref{estgamma}, \eqref{estH}, Proposition~\ref{prpsmallness} and Proposition~\ref{holder},
\begin{equation*}
\begin{split}
& \|e^{-2\gamma^{(n)}} N^{(n)} (e_0^{(n)} H^{(n)})\|_{H^1_{\de'+1}}\\
\ls & \|e^{-2\gamma^{(n)}} N^{(n)} \|_{C^0_{-\ep}} \|e_0^{(n)} H^{(n)}\|_{H^1_{\de+1}} + \|e^{-2\gamma^{(n)}} N^{(n)} \|_{W^{1}_{-\ep-\f 12,4}} \|e_0^{(n)} H^{(n)}\|_{L^4_{\de+\f 32}}\\
\ls & \|e^{-2\gamma^{(n)}} N^{(n)} \|_{W^{1}_{-\ep-\f 12,4}}\|e_0^{(n)} H^{(n)}\|_{H^1_{\de+1}}\ls C_i.
\end{split}
\end{equation*}
The other terms in \eqref{rdtbetaest.RHS} can be similarly shown to be $\ls C_i$.

For term $II$ in \eqref{e0.beta.eqn}, we first compute
$$\left|[\Delta, e_0^{(n-1)}] (\beta^{(n)})^i \right|\ls |\nab\nab \bt^{(n-1)}||\nab \bt^{(n)}| + |\nab \bt^{(n-1)}||\nab\nab \bt^{(n)}|. $$
We then estimate this term. In fact, it will be convenient later to bound it in a stronger norm, namely the $H^2_{\de'+2}$ norm (instead of the $H^0_{\de'+2}$ norm), using \eqref{estbeta.small}, \eqref{estbeta.large} and Proposition~\ref{holder}:
\begin{equation}\label{rdtbetaest.II}
\begin{split}
\|II\|_{H^2_{\de'+2}} \ls & \| \bt^{(n-1)}\|_{H^4_{\de'}} \|\bt^{(n)}\|_{C^1_{\de'+1}} + \| \bt^{(n-1)}\|_{W^3_{\de'+\f 12,4}} \|\bt^{(n)}\|_{W^2_{\de'+\f 12,4}} \\
& + \| \bt^{(n)}\|_{H^4_{\de'}} \|\bt^{(n-1)}\|_{C^1_{\de'+1}} + \| \bt^{(n)}\|_{W^3_{\de'+\f 12,4}} \|\bt^{(n-1)}\|_{W^2_{\de'+\f 12,4}} \ls \ep C(A_0) C_i. 
\end{split}
\end{equation}

\textbf{Proof of \eqref{betaest.top}.} Finally, to prove \eqref{betaest.top}, we apply Theorem~\ref{laplacien} and estimate the RHS of \eqref{e0.beta.eqn} in $H^2_{\de'+2}$. Let us note at this point that it is for the purpose of \eqref{betaest.top}, we need to commute the equation with $e_0^{(n)}$ instead of, say, $\rd_t$. This is so that we can make use of the top order estimate in \eqref{estH}. Indeed, using the estimates for $H^{(n)}$ in \eqref{estH}, we cannot bound general fifth derivatives of $H^{(n)}$, but can only bound the combination of one $e_0^{(n)}$ and four spatial derivatives.

Arguing as \eqref{rdtbetaest.RHS} and using Proposition~\ref{produit}, for $I$ in \eqref{e0.beta.eqn}, it suffices to estimate the following terms:
\begin{equation}\label{rdtbetaest.RHS.top}
\begin{split}
\| I \|_{H^2_{\de'+2}} \ls & \|e^{-2\gamma^{(n)}} N^{(n)} (e_0^{(n)} H^{(n)})\|_{H^3_{\de'+1}} + \|e^{-2\gamma} (e_0^{(n)} \gamma^{(n)}) N^{(n)} H^{(n)}\|_{H^3_{\de'+1}} \\
& + \|e^{-2\gamma} (e_0^{(n)} N^{(n)}) H^{(n)}\|_{H^3_{\de'+1}} + \| \bt^{(n)}\|_{H^2_{\de'}} \|\rd_k (e^{-2\gamma^{(n)}} N^{(n)} H^{(n)})\|_{H^2_{\de'+1}}.
\end{split}
\end{equation}
As in the proof of \eqref{betaest.higher}, the last term is somewhat easier, and can be bounded by $\ls \ep C_i$ by \eqref{estbeta.small} and \eqref{betaest.RHS}. It remains to estimate the first three terms. Again, as in the proof of \eqref{betaest.higher}, they are rather similar and we will only carry out the estimate for the first term in detail. More precisely, by Lemma~\ref{lm.constraint}, \eqref{estN.small}, \eqref{estN.large}, \eqref{estgamma}, \eqref{estH} and Proposition~\ref{prpsmallness},
\begin{equation*}
\begin{split}
& \|e^{-2\gamma^{(n)}} N^{(n)} (e_0^{(n)} H^{(n)})\|_{H^3_{\de'+1}} \\
\ls & \left(1 + |N^{(n)}_{asymp}| + \| \wht N^{(n)}\|_{W^2_{\de{+}\f 12, 4}} \right)\left(1 + \|\wht \gamma^{(n)} \|_{W^2_{\de'+\f 12,4}} \right)\|e_0^{(n)} H^{(n)}\|_{H^3_{\de'+1}} \\
& + \left( \| \wht N^{(n)}\|_{H^3_{\de}} +  \|\wht \gamma^{(n)} \|_{H^3_{\de'}} \right) \|e_0^{(n)} H^{(n)} \|_{C^0_{\de'+2}}\\
\ls & C_i + C_i^2 \ls C_i^2.
\end{split}
\end{equation*}
The other terms in \eqref{rdtbetaest.RHS.top} can be treated in a similar manner. Finally, the term $II$ in \eqref{e0.beta.eqn} has already been estimated in $H^2_{\de'+2}$ in \eqref{rdtbetaest.II}. We therefore conclude the proof of \eqref{betaest.top}.
\end{proof}

We have now completed all the elliptic estimates. In the remaining estimates, we can exploit the smallness time parameter $T$. However, one still needs to take caution when estimating the time derivatives, as these are typically controlled by estimating the RHS of the evolution equations, and one still needs to track precisely the dependence of the constants.

The next quantity we bound is $\gamma^{(n+1)}$. In the following lemma, we prove an energy estimate for general solutions to {inhomogenous} wave equations of the type satisfied by $\gamma^{(n+1)}$ in \eqref{gamman}.
\begin{lm}[Energy estimate for the wave equation satisfied by $\gamma^{(n+1)}$]\label{gamma.wave.EE}
Suppose $h$ satisfies the following inhomogeneous wave equation
\begin{equation}\label{gamma.wave.general}
\frac{1}{N^{(n)}}e_0^{(n)} \left(\frac{1}{N^{(n)}} e_0^{(n)} h\right) -\Delta h =f.
\end{equation}
Then, for a weight function $w(|x|) = (1+|x|^2)^\sigma$, it holds that
\begin{equation}\label{wave.gamma.EE}
\begin{split}
& \int_{\mathbb R^2}  w\left(\frac{1}{(N^{(n)})^2}(e_0^{(n)} h)^2+|\nabla h|^2\right)(t,x) \, dx \\
\leq & 2 \int_{\mathbb R^2}  w\left(\frac{1}{(N^{(n)})^2}(e_0^{(n)} h)^2+|\nabla h|^2\right)(0,x) \, dx + C T \sup_{t'\in [0,T]}\int_{\mathbb R^2} w (N^{(n)})^2 f^2 (t',x)\, dx.
\end{split}
\end{equation}
In particular, this implies that
\begin{equation}\label{h.wave.higher}\begin{split}
\sum_{|\alpha|\leq 3}\left\|\f{e_0^{(n)} \nabla^{\alpha}h}{N^{(n)}} \right\|_{L^2_{\de'+1+|\alpha|}}(t) + \|\nab h\|_{H^3_{\de'+1}}(t) \leq &2 \left(\sum_{|\alpha|\leq 3}\left\|\f{e_0^{(n)} \nabla^{\alpha}h}{N^{(n)}} \right\|_{L^2_{\de'+1+|\alpha|}}(0) + \|\nab h\|_{H^3_{\de'+1}}(0)\right) \\
&+ C(C_i) T \sup_{t'\in [0,T]} \|N^{(n)} f\|_{H^3_{\de'+1}}(t').
\end{split}
\end{equation}
\end{lm}
\begin{proof}
\textbf{Proof of \eqref{wave.gamma.EE}.} 
Let $w(|x|)$ be as in the statement of the lemma. We multiply \eqref{gamma.wave.general} by $w e_0^{(n)} h$ and integrate over $\m R^2$ with respect to $dx$.
	After integration by parts, we obtain
	$$\int_{\mathbb R^2} \frac{1}{2}we_0^{(n)} \left( \frac{e_0^{(n)} h}{N^{(n)}}\right)^2 \,dx +\int_{\mathbb R^2} \nabla h \cdot \nabla\left(we_0^{(n)} h\right)\,dx=\int_{\mathbb R^2} wfe_0^{(n)} h \,dx.$$
	Hence,
	\begin{align*}
	&\frac{d}{dt}\int_{\mathbb R^2} w\left(\frac{(e_0^{(n)} h)^2}{2(N^{(n)})^2}+\frac{1}{2}|\nabla h|^2\right)\,dx+\int_{\mathbb R^2} \frac{1}{2}\partial_i\left((\beta^{(n)})^i w\right) \left(\frac{1}{(N^{(n)})^2}(e_0^{(n)} h)^2+|\nabla h|^2\right)\,dx\\
	&\qquad+\int_{\mathbb R^2} w'\f{x^i}{|x|}\partial_i h e_0^{(n)} h \,dx
	-\int_{\mathbb R^2} w \de^{ik} \partial_i (\beta^{(n)})^j\partial_k h\partial_j h \,dx=\int_{\mathbb R^2} w fe_0^{(n)} h \,dx.
	\end{align*}
	Since $w(|x|)=(1+|x|^2)^\sigma$ and $N^{(n)}\lesssim (1+\ep)(1+\chi(|x|){\log}(|x|))$,
	$$w'(|x|)\lesssim \frac{w(|x|)}{(1+|x|^2)^\frac{1}{2}}\ls \frac{w(|x|)}{N^{(n)}}.$$
	Moreover, by \eqref{estbeta.small} and Proposition~\ref{holder}, $\|\beta^{(n)}\|_{L^\infty} + \|\nab \beta^{(n)}\|_{L^\infty}\lesssim \|\beta^{(n)}\|_{W^2_{\delta'+\f 12,4}}\lesssim 1$. Hence,
	$$|\nab (\beta^{(n)} w)|\lesssim  w, \quad |w \nab \beta^{(n)}| \lesssim w.$$
	Using the above estimates and Cauchy--Schwarz, we therefore have
	\begin{equation*}
	\begin{split}
	& \frac{d}{dt}\int_{\mathbb R^2}  w\left(\frac{(e_0^{(n)} h)^2}{(N^{(n)})^2}+|\nabla h|^2\right)\, dx\\
	\lesssim & \int_{\mathbb R^2}  w\left(\frac{(e_0^{(n)} h)^2}{(N^{(n)})^2}+|\nabla h|^2\right) \, dx
	+\left( \int_{\mathbb R^2}  w\frac{(e_0^{(n)} h)^2}{(N^{(n)})^2} \, dx\right)^\frac{1}{2}\left(\int_{\mathbb R^2} w(N^{(n)})^2f^2 \, dx\right)^\frac{1}{2}.
	\end{split}
	\end{equation*}
	Therefore, using Gr\"onwall's inequality and choosing $T$ sufficiently small, we obtain \eqref{wave.gamma.EE}.

\textbf{Proof of \eqref{h.wave.higher}.} We differentiate \eqref{gamma.wave.general} by up to three spatial derivatives to obtain that for multi-index $\alp$ with $|\alp|\leq 3$,
\begin{equation}\label{h.wave.commute.0}
\frac{1}{N^{(n)}}e_0^{(n)} \left(\frac{1}{N^{(n)}} e_0^{(n)} (\nab^\alp h)\right) -\Delta (\nab^\alp h) =\nab^\alp f + \left[\f{1}{N^{(n)}} e_0^{(n)} (\f{1}{N^{(n)}} e_0^{(n)} ), \nab^{\alp}\right] h .
\end{equation}
By \eqref{wave.gamma.EE} with $\sigma = \f{\de'+1+|\alp|}{2}$, in order to prove \eqref{h.wave.higher}, we need to multiply the RHS of \eqref{h.wave.commute.0} by $N^{(n)}$ and bound it in $L^2_{\de'+1+|\alp|}$. We first control the $\nab^\alp f$ term. Notice that in order to obtain the term on RHS of \eqref{h.wave.higher}, we in particular need to commute $[N^{(n)}, \nab^{\alp}]$. By H\"older's inequality, Proposition~\ref{holder} and \eqref{estN.large} (and noting that $\| \nab (\chi(|x|){\log}(|x|) N_{asymp}^{(n)}) \|_{C^k_{-1}}\ls_k |N_{asymp}|$, $\forall k\in \mathbb N\cup \{0\}$ and $\|\f{1}{N^{(n)}}\|_{C^0} \ls 1$),
\begin{equation}\label{N.nab.commute}
\begin{split}
&\sum_{|\alp|\leq 3} \|N^{(n)} \nab^\alp f \|_{L^2_{\de'+1+|\alp|}}\\
\ls & \sum_{|\alp|\leq 3} \| \nab^\alp (N^{(n)} f) \|_{L^2_{\de'+1+|\alp|}} + \sum_{|\alp_1|+|\alp_2| \leq 2} \|\nab^{\alp_1} \nab \log N^{(n)} \nab^{\alp_2} (N^{(n)}f) \|_{L^2_{\de'+2+|\alp_1|+|\alp_2|}} \\
& + \sum_{|\alp_1|+|\alp_2| \leq 1} \|(\nab \log N^{(n)} )(\nab^{\alp_1} \nab \log N^{(n)}) \nab^{\alp_2} (N^{(n)}f) \|_{L^2_{\de' + 3 +|\alp_1|+|\alp_2|}}\\
& + \|(\nab \log N^{(n)})^3 (N^{(n)}f) \|_{L^2_{\de' + 4}}\\
\ls & \|N^{(n)} f \|_{H^3_{\de'+1}} \left(1 + |N_{asymp}^{(n)}| + \| \nab \log \wht N^{(n)}\|_{H^2_{\de+1}}^3\right) \ls C(C_i) \|N^{(n)} f \|_{H^3_{\de'+1}}.
\end{split}
\end{equation}
To control the commutator term in \eqref{h.wave.commute.0}, we compute
\begin{equation*}
\begin{split}
& \left|\left[\f{1}{N^{(n)}} e_0^{(n)} \left(\f{1}{N^{(n)}} e_0^{(n)} \right), \nab \right] h \right| \\
 \ls &  \left| \f{\nab \bt^{(n)}}{N^{(n)}}\nab \left(\f{e_0^{(n)}h}{N^{(n)}}  \right)\right|+ \left|f \nab \log N^{(n)} \right| +
\left|\nabla^2 h \nab \log N^{(n)} \right|+
  \left| \f{1}{N^{(n)}} \left(e_0^{(n)} \f{\nab \bt^{(n)}}{N^{(n)}} \right) \nab h\right| \\
	&+ \left|\f{\nab\beta^{(n)}}{N^{(n)}} \nab\log N^{(n)} \f{e_0^{(n)}h}{N^{(n)}}\right|
+ \left| \f{\nab \bt^{(n)}}{N^{(n)}} \f{\nab \bt^{(n)}}{N^{(n)}} \nab h\right| + \left|\f{\left(e_0^{(n)}  \nab \log N^{(n)}\right)}{N^{(n)}} \left( \f{e_0^{(n)}h}{N^{(n)}} \right)\right|.
\end{split}
\end{equation*}
Here, we have silently used $[e_0^{(n)}, \rd_i] = \rd_i (\bt^{(n)})^j \rd_j${, and have also used the equation \eqref{gamma.wave.general} to rewrite $\f{1}{N^{(n)}} e_0^{(n)} \left(\f{1}{N^{(n)}} e_0^{(n)}h\right)$}. In a similar manner, one can compute the commutator with higher derivatives. We have, for $|\alp|\leq 3$,
\begin{equation}\label{N.nab.commute.1}
\begin{split}
& \left|\left[\f{1}{N^{(n)}} e_0^{(n)} (\f{1}{N^{(n)}} e_0^{(n)} ), \nab^{\alp} \right] h \right| \\
\ls &  \sum_{|\alp_1| + |\alp_2| + i= |\alp|}\left(\left|\nab^{\alp_1} f \nab^{\alp_2} (\nab \log N^{(n)})^i \right|
+\left|\nab^{\alp_1} \nabla ^2 h \nab^{\alp_2} (\nab \log N^{(n)})^i \right|\right)\\
& + \sum_{|\alp_1| + |\alp_2| + |\alp_3| + i = |\alp| - 1} \left| \nab^{\alp_1}\left(\f{1}{N^{(n)}} e_0^{(n)} \left(\f{\nab \bt^{(n)}}{N^{(n)}} \right)\right) \nab^{\alp_2}(\log N^{(n)})^i \nab^{\alp_3} \nab h\right|\\
& + \sum_{|\alp_1| + |\alp_2| + |\alp_3| + i = |\alp|}\left|\nab^{\alp_1} \left(\f{\nab \bt^{(n)}}{N^{(n)}}\right) \nab^{\alp_2} (\nab \log N^{(n)})^i \nab^{\alp_3}\f{e_0^{(n)} h}{N^{(n)}} \right|\\
& + \sum_{|\alp_1| + |\alp_2| + |\alp_3| + i= |\alp|} \left|\nab^{\alp_1} \left(\f{\nab \bt^{(n)}}{N^{(n)}}\right)^{2} \nab^{\alp_2} (\nab \log N^{(n)})^i \nab^{\alp_3}\nab h \right|\\
& + \sum_{|\alp_1| + |\alp_2| + |\alp_3| + i= |\alp| - 1} \left|\nab^{\alp_1} \f{\left( e_0^{(n)}\log N^{(n)}\right)}{N^{(n)}} \nab^{\alp_2} (\nab\log N^{(n)})^i \nab^{\alp_3}\f{e_0^{(n)} h}{N^{(n)}} \right|\\
& + \sum_{|\alp_1| + |\alp_2| + |\alp_3| +|\alpha_4|+ i= |\alp| - 1} \left|\nab^{\alp_1} \f{\left(  e_0^{(n)} \log N^{(n)}\right)}{N^{(n)}} \nab^{\alp_2} \f{\nab \bt^{(n)}}{N^{(n)}} \nab^{\alp_3} (\nab\log N^{(n)})^i \nab^{\alp_4} h \right|.
\end{split}
\end{equation}
Here (and below), we use the notation {as in the proof of Proposition~\ref{prpeogamma}} that, say, $\nab^{\alp_2} (\nab \log N^{(n)})^i$ denotes a product of $i$ factors, each of which is some spatial derivatives of $\nab\log N^{(n)}$ and the total number of derivatives is $|\alp_2|$.

We claim that upon multiplying by $N^{(n)}$, each of the terms in \eqref{N.nab.commute.1} can be bounded in $L^2_{\de'+1+|\alp|}$ by $C(A_1,C_i)(\|h\|_{H^3_{\de'+1}}+ \|N^{(n)} f \|_{H^3_{\de'+1}})$. Since the constant can depend on $A_1$ and $C_i$ in an arbitrary manner, all terms linear in $h$ can be treated in essentially the same way. We bound one representative term. Using H\"older's inequality, Proposition~\ref{holder}, \eqref{estN.large} (and noting that $\| \nab (\chi(|x|){\log}(|x|) N_{asymp}^{(n)}) \|_{C^k_{-1}}\ls_k |N_{asymp}|$, $\forall k\in \mathbb N\cup \{0\}$ and $\|\f{1}{N^{(n)}}\|_{C^0} \ls 1$) and \eqref{estbeta.large}, we have
\begin{equation}\label{N.nab.commute.2}
\begin{split}
& \sum_{|\alp_1| + |\alp_2| +|\alp_3| + i = |\alp|-1} \left\|N^{(n)} \nab^{\alp_1}\left(\f{1}{N^{(n)}} e_0^{(n)} \left(\f{\nab \bt^{(n)}}{N^{(n)}} \right)\right) \nab^{\alp_2}(\nab \log N^{(n)})^i \nab^{\alp_3} \nab h\right\|_{L^2_{\de'+1+|\alp|}}\\
\ls &  \sum_{|\alp_1| + |\alp_2| + |\alp_3| + i = |\alp| - 1} \left\|\nab^{\alp_1} \rd_t \f{\nab \bt^{(n)}}{N^{(n)}} \nab^{\alp_2}(\nab \log N^{(n)})^i \nab^{\alp_3} \nab h\right\|_{L^2_{\de'+1+|\alp|}} \\
& + \sum_{|\alp_1| + |\alp_2| + |\alp_3| + |\alp_4| + i = |\alp|-1} \left\| N^{(n)} \nab^{\alp_1} \f{\beta^{(n)}}{N^{(n)}} \nab^{\alp_2} \nab \f{\nab \bt^{(n)}}{N^{(n)}} \nab^{\alp_3}(\nab \log N^{(n)})^i \nab^{\alp_4} \nab h\right\|_{L^2_{\de'+1+|\alp|}}\\
\ls & \left(\left\|\rd_t \f{\nab \bt^{(n)}}{N^{(n)}}\right\|_{H^2_{\de'}}+\left\|\f{\bt^{(n)}}{N^{(n)}}\right\|_{H^2_{\de'}} \left\|\nab \f{\nab \bt^{(n)}}{N^{(n)}}\right\|_{H^2_{\de'+1}}\right) \left(1 + |N^{(n)}_{asymp}|^2 + \|\nab \wht N^{(n)}\|_{H^1_{\de+1}}^2\right) \|\nab h\|_{H^2_{\de'+1}} \\
\ls & A_1 C_i^2\|\nab h\|_{H^2_{\de'+1}}.
\end{split}
\end{equation}
The term which is linear in $f$ can be treated as in \eqref{N.nab.commute} to obtain
\begin{equation}\label{N.nab.commute.3}
\begin{split}
&\sum_{|\alp_1| + |\alp_2| + i= |\alp|} \left\| N^{(n)} \nab^{\alp_1} f \nab^{\alp_2} (\nab \log N^{(n)})^i \right\|_{L^2_{\de'+1+|\alp|}}\ls C(C_i) \|N^{(n)} f \|_{H^3_{\de'+1}}.
\end{split}
\end{equation}
Combining \eqref{h.wave.commute.0}, \eqref{N.nab.commute}, \eqref{N.nab.commute.1}, \eqref{N.nab.commute.2} (and analogous estimates for terms in \eqref{N.nab.commute.1}) and \eqref{N.nab.commute.3}, and using \eqref{wave.gamma.EE}, yield
\begin{equation}\label{wave.gamma.almost}
\begin{split}
& \sum_{|\alp|\leq 3} \left\|\f{e_0^{(n)} \nab^{\alp} h}{N^{(n)}}\right\|_{L^2_{\de'+1+|\alp|}}(t) + \|\nab h\|_{H^k_{\de'+1}}(t) \\
\leq & 2\sum_{|\alp|\leq 3} \left(\left\|\f{e_0^{(n)} \nab^{\alp} h}{N^{(n)}} \right\|_{L^2_{\de'+1+|\alp|}}(0) + \| \nab \nab^{\alp} h\|_{L^2_{\de'+1+|\alpha|}}(0)\right) + C(C_i) T \sup_{t'\in [0,T]} \|N^{(n)} f\|_{H^3_{\de'+1}}(t').
\end{split}
\end{equation}
which concludes the proof of the Lemma.
\end{proof}

\begin{prp}[Estimates for $\wht{\gamma}$]\label{prpgamma}
	For $n\geq 2$, the following estimates hold:
	\begin{align}
	 \label{gamma.est.lower} \sum_{|\alpha|\leq 3} \left\|\frac{e_0^{(n)} \nabla^\alpha \wht \gamma^{(n+1)}}{N^{(n)}}\right\|_{L^2_{\delta'+1+|\alpha|}}
	+\left\|\nabla \wht \gamma^{(n+1)}\right\|_{H^3_{\delta'+1}}
	\leq & 4C_i,\\
	\label{gamma.est.higher.2}	\left\|\partial_t \frac{e_0^{(n)} \wht \gamma^{(n+1)}}{N^{(n)}}\right\|_{H^2_{\delta'+1}}\lesssim & C(A_0)C_i^{2},\\
	\label{gamma.est.higher.3} \left\|\partial_t \frac{e_0^{(n)} \wht \gamma^{(n+1)}}{N^{(n)}}\right\|_{H^1_{\delta'+1}}\lesssim & C(A_0) C_i.	
		\end{align}
\end{prp}

\begin{proof}
	The strategy is to use Lemma~\ref{gamma.wave.EE} to estimate $\wht \gamma^{(n+1)}$. Notice that \eqref{gamman} is an equation for $\gamma^{(n+1)}$. Nevertheless, from \eqref{gamman} one can easily derive the following equation for $\wht \gamma^{(n+1)}$:
	
	\begin{equation}\label{wave.whtgamma}
	\begin{split}
	& -\f 2{N^{(n)}}e_0^{(n)} \left(\f{e_0^{(n)} \wht \gamma^{(n+1)}}{N^{(n)}}\right)+2\Delta \wht \gamma^{(n+1)} \\
	= & (\mbox{RHS of \eqref{gamman}})+2 \alpha \Delta(\chi(|x|){\log}(|x|)) + \frac{2 \alp}{N^{(n)}} e_0^{(n)} \left(\f{(\beta^{(n)})^i}{N^{(n)}} \rd_i (\chi(|x|){\log}(|x|) ) \right). 
	\end{split}
	\end{equation}
	
	Moreover, if $\wht \gamma^{(n+1)}$ satisfies the estimates as indicated in the statement of the proposition and $\alp$ is a fixed constant as in the initial data (and in particular time-independent), then $\gamma^{(n+1)}= -\alp \chi(|x|){\log} (|x|) + \wht \gamma^{(n+1)}$ is indeed the solution to \eqref{gamman}.

\textbf{Proof of \eqref{gamma.est.lower}: Estimates for RHS of \eqref{wave.whtgamma} in $H^3_{\de'+1}$.} 	By Lemma~\ref{gamma.wave.EE}, to prove \eqref{gamma.est.lower}, it suffices to show that $N^{(n)} \times (\mbox{RHS of \eqref{wave.whtgamma}})$ is bounded in the $H^3_{\de'+1}$ norm by $C(A_0, C_i)$. 
	We first estimate the RHS of \eqref{gamman}.
	For the term $\frac{2 (e_0^{(n-1)} \gamma^{(n)})^2}{N^{(n)} N^{(n-1)}}$, we write
	$$(e_0^{(n-1)} \gamma^{(n)})^2= \left(e_0^{(n-1)} \wht \gamma^{(n)} + \alpha\beta^{(n-1)}_i\partial_i(\chi(|x|){\log}(|x|))\right)^2.$$
	By Cauchy--Schwarz, it clearly suffices to bound $\frac{(e_0^{(n-1)} \wht \gamma^{(n)})^2}{(N^{(n-1)})}$ and $\f{\left(\alp (\beta^{(n-1)})^i\partial_i(\chi(|x|){\log}(|x|))\right)^2}{N^{(n-1)}}$ in $H^3_{\delta'+1}$.
	
	By H\"older's inequality \eqref{estN.small}, \eqref{estN.large}, \eqref{estgamma}, Proposition~\ref{prpsmallness}, \eqref{estgammabis} and Proposition~\ref{holder},
	\begin{equation*}
	\begin{split}
	&\left\|\frac{(e_0^{(n-1)} \wht \gamma^{(n)})^2}{(N^{(n-1)})}\right\|_{H^3_{\delta'+1}} = \left\|N^{(n-1)} \cdot \frac{e_0^{(n-1)} \wht \gamma^{(n)}}{(N^{(n-1)})}\cdot \frac{e_0^{(n-1)} \wht \gamma^{(n)}}{(N^{(n-1)})}\right\|_{H^3_{\delta'+1}} \\
	\lesssim & \left(1+ |N^{(n-1)}_{asymp}| + \|\wht N^{(n-1)}\|_{H^2_\de}\right) \left(\left\|\f{e_0^{(n-1)} \wht \gamma^{(n)}}{N^{(n-1)}} \right\|_{H^3_{\delta'+1}} \left\|\f{e_0^{(n-1)} \wht \gamma^{(n)}}{N^{(n-1)}} \right\|_{C^0_{\delta'+2}} \right.\\
	&\qquad \left.+ \left\|\f{e_0^{(n-1)} \wht \gamma^{(n)}}{N^{(n-1)}} \right\|_{W^2_{\delta'+\f 32, 4}} \left\|\f{e_0^{(n-1)} \wht \gamma^{(n)}}{N^{(n-1)}} \right\|_{W^1_{\de'+\f 32, 4}} \right) + \|\wht N^{(n-1)} \|_{H^3_\de} \left\|\f{e_0^{(n-1)} \wht \gamma^{(n)}}{N^{(n-1)}} \right\|_{C^0_{\delta'+2}}^2 \ls \ep C_i.
	\end{split}
	\end{equation*}
Next, we estimate using H\"older's inequality, Proposition~\ref{holder}, \eqref{est.constraints.0}, \eqref{estN.small}, \eqref{estN.large}, \eqref{estbeta.small} and \eqref{estbeta.large} (and the fact $|\f{1}{N^{(n)}}|\ls 1$) that
	\begin{equation}
	\begin{split}
	& \left\|\f{\left(\alp (\beta^{(n-1)})^i\partial_i(\chi(|x|){\log}(|x|))\right)^2}{N^{(n)}}\right\|_{H^3_{\delta'+1}} \\
	\lesssim &|\alp|^2 \left(1 + |N^{(n)}_{asymp}|^3 + \|\wht N^{(n)}\|_{H^2_{\de}}^2 \right) \left(\|\beta^{(n-1)}\|_{H^3_{\delta'}} \|\beta^{(n-1)}\|_{C^0_{\delta'+1}} + \|\beta^{(n-1)}\|_{W^2_{\delta'+\f 12, 4}} \|\beta^{(n-1)}\|_{W^1_{\delta'+\f 12, 4}} \right)\\
	& + |\alp|^2 \|\wht N^{(n)}\|_{H^3_{\de}} \|\bt^{(n-1)}\|_{C^0_{\de'+1}}^2 \lesssim \ep^5 C(A_0) C_i.
	\end{split}
	\end{equation}
For the second term on the RHS of \eqref{gamman}, by H\"older's inequality, Proposition~\ref{holder}, Lemma~\ref{produit2}, \eqref{est.constraints.0}, \eqref{estN.small}, \eqref{estN.large}, \eqref{estgamma}, \eqref{esttau} and Proposition~\ref{prpsmallness}, we have
	\begin{equation}
	\begin{split}
	&\| N^{(n)}(\tau^{(n)})^2 e^{2\gamma^{(n)}} \|_{H^3_{\de'+1}}\\
	\ls & \left(1+ |N^{(n)}_{asymp}| + \|\wht N^{(n)}\|_{W^2_{\de+\f 12,4}}\right) \left(1 + |\alp| + \|\wht \gamma^{(n)} \|_{W^2_{\de'+\f 12, 4}} \right)\left(\|\tau^{(n)}\|_{H^3_{\de''+1}} \|\tau^{(n)} \|_{C^0_{\de''+2}} \right.\\
	&\left.\qquad + \|\tau^{(n)}\|_{W^2_{\de''+\f 32, 4}} \|\tau^{(n)} \|_{W^1_{\de''+\f 32,4}} \right) + \left( \|\wht N^{(n)} \|_{H^3_{\de}} + \|\wht \gamma^{(n)} \|_{H^3_{\de'}} \right)\|\tau^{(n)} \|_{C^0_{\de''+2}}^2\\
	\ls & \ep C(A_1) C_i.
	\end{split}
	\end{equation}
	
	We next bound the most difficult term, $e^{2\gamma^{(n)}}e^{(n-1)}_0 \left( \frac{e^{-2\gamma^{(n)}}}{N^{(n-1)}}div(\beta^{(n)})\right)$, which is also the term that limits the weight allowable in the estimate. Distributing the $e_0^{(n-1)}$ derivative and commuting $e^{(n-1)}_0$ with $\bt^{(n)}$, it is clear that it suffices to control the following terms:
	\begin{equation}\label{gamma.est.lower.main}
	\begin{split}
	& \left\|e^{2\gamma^{(n)}}e^{(n-1)}_0 \left( \frac{e^{-2\gamma^{(n)}}}{N^{(n-1)}}div(\beta^{(n)})\right)\right\|_{H^3_{\delta'+1}}\\
	\leq & \left\|  \frac{2(e^{(n-1)}_0\gamma^{(n)})}{N^{(n-1)}}div\,\beta^{(n)}\right\|_{H^3_{\delta'+1}} + \left\| \frac{e^{(n-1)}_0 \log N^{(n-1)}}{N^{(n-1)}} div\,\beta^{(n)}\right\|_{H^3_{\delta'+1}} \\
	& + \left\|  \frac{\nab\beta^{(n)} \nab\beta^{(n-1)}}{N^{(n-1)}}\right\|_{H^3_{\delta'+1}} + \left\|  \frac{div (e^{(n-1)}_0  \beta^{(n)})}{N^{(n-1)}}\right\|_{H^3_{\delta'+1}}
	=:  I + II + III + IV. 
	\end{split}
	\end{equation}
	We begin with term I. By \eqref{est.constraints.0}, \eqref{estN.small}, \eqref{estN.large}, \eqref{estbeta.small}, \eqref{estbeta.large}, \eqref{estgamma}, Proposition{s}~\ref{prpsmallness}{, \ref{prpeogamma}}, Lemma~\ref{der} and Proposition~\ref{holder},
	\begin{equation}\label{gamma.est.lower.I}
	\begin{split}
	I  \ls & \left(|\alp| \left\| \beta^{(n-1)}  \f{\nab(\chi(|x|){\log} (|x|))}{N^{(n-1)}}\right\|_{W^1_{\de'+\f 32, 4}} +  \left\|\f{e^{(n-1)}_0 \wht \gamma^{(n)}}{N^{(n-1)}}\right\|_{W^1_{\de'+\f 32, 4}} \right)\|div \,\beta^{(n)}\|_{H^3_{\delta'+1}} \\
	& + \left(|\alp| \left\| \beta^{(n-1)} \f{\nab(\chi(|x|){\log} (|x|))}{N^{(n-1)}}\right\|_{H^3_{\de'+1}} +  \left\|\f{e^{(n-1)}_0 \wht \gamma^{(n)}}{N^{(n-1)}}\right\|_{H^3_{\de'+1}} \right)\|div \,\beta^{(n)}\|_{W^1_{\delta'+\f 32, 4}} \\
	\ls & \ep C(A_0) C_i.
	\end{split}
	\end{equation}
	For term II, we write $e_0^{(n-1)} \log N^{(n-1)} = \f{(\rd_t - (\bt^{(n-1)})^i \rd_i) (N_{asymp}^{(n-1)}\chi(|x|){\log}(|x|) + \wht N^{(n-1)})}{N^{(n-1)}}$. Hence, using Lemma~\ref{der}, Proposition~\ref{produit}, \eqref{estN.small}, \eqref{estN.large}, \eqref{estbeta.small}, \eqref{estbeta.large} and dropping the good $\f{1}{N^{(n-1)}}$ factor, we have
	\begin{equation}\label{gamma.est.lower.II}
	\begin{split}
	II \ls & \left(|\rd_t N_{asymp}^{(n-1)}| + |N_{asymp}^{(n-1)}| \|\bt^{(n-1)} (1+|x|)^{-1} \|_{H^3_{\de}} + \left\|\rd_t \wht N^{(n-1)}\right\|_{H^2_{\de}} + \left\|\beta^{(n-1)} \nab \wht N^{(n-1)}\right\|_{H^3_{\de}} \right) \\
	&\qquad \times\left(\left(1 + |N_{asymp}^{(n-1)}|^3 + \|\wht N^{(n-1)} \|_{H^2_{\de}}^3\right) \|div \,\beta^{(n)}\|_{H^3_{\delta'+1}}+ \|\wht N^{(n-1)} \|_{H^3_{\de}} \|div \,\beta^{(n)}\|_{C^0_{\delta'+2}}\right)\\
	& + \|\rd_t \wht N^{(n-1)} \|_{H^3_\de} \|div\, \bt^{(n)}\|_{C^0_{\de'+1}}\\
	\ls & C(A_0)C_i^2.
	\end{split}
	\end{equation}
	For term $III$, after expanding in terms of derivatives of $\bt^{(n-1)}$, $\beta^{(n)}$ and $\log N^{(n-1)}$ (and dropping the good $\f{1}{N^{(n)}}$ factor), there are the following possibilities: (1) Any factors of $\bt^{(n-1)}$, $\beta^{(n)}$ and $N^{(n-1)}$ have at most $2$ derivatives; (2) one factor of $\bt^{(n-1)}$ or $\beta^{(n)}$ has at least three derivatives; (3) there is a factor of three derivatives of $\log N^{(n-1)}$. In case (1), by \eqref{estN.small}, \eqref{estbeta.small} and Lemma~\ref{der}, we estimate
	\begin{equation}
	\begin{split}
	\ls &\left(1+ |N_{asymp}^{(n-1)}|^3+ \|\nab\log \wht N^{(n-1)}\|_{C^0_{\de+2}}^3 \right) \|\beta^{(n-1)}\|_{W^2_{\de'+\f 12, 4}}\|\beta^{(n)}\|_{W^2_{\de'+\f 12, 4}} \\
	&+ \left(1+ |N_{asymp}^{(n-1)}|+ \|\nab\log \wht N^{(n-1)}\|_{W^1_{\de+\f 32, 4}}\right) \left(1+ |N_{asymp}^{(n-1)}|+ \|\nab\log \wht N^{(n-1)}\|_{C^0_{\de+2}}\right) \\
	&\qquad\times\left(\|\beta^{(n-1)}\|_{C^1_{\de'+1}} \|\beta^{(n)}\|_{W^2_{\de'+\f 12, 4}}+ \|\beta^{(n)}\|_{C^1_{\de'+1}} \|\beta^{(n-1)}\|_{W^2_{\de'+\f 12, 4}}\right)\\
	\ls & \ep^2.
	\end{split}
	\end{equation}
	In case (2), by \eqref{estN.small}, \eqref{estbeta.small}, \eqref{estbeta.large}, Lemma~\ref{der} and Proposition~\ref{holder}, we have
	\begin{equation}
	\begin{split}
	\ls &\|\bt^{(n-1)}\|_{W^3_{\de'+\f 12, 4}}(\|\bt^{(n)}\|_{W^2_{\de'+\f 12, 4}} + \|\bt^{(n)}\|_{C^1_{\de'+1}} \|\nab\log \wht N^{(n-1)}\|_{C^0_{\de+2}}) + \|\bt^{(n-1)}\|_{H^4_{\de'}} \|\bt^{(n)}\|_{C^1_{\de'+1}}\\
	&+\|\bt^{(n)}\|_{W^3_{\de'+\f 12, 4}}(\|\bt^{(n-1)}\|_{W^2_{\de'+\f 12, 4}} + \|\bt^{(n-1)}\|_{C^1_{\de'+1}} \|\nab\log \wht N^{(n-1)}\|_{C^0_{\de+2}}) + \|\bt^{(n)}\|_{H^4_{\de'}} \|\bt^{(n-1)}\|_{C^1_{\de'+1}}\\
	\ls & \ep C(A_0) C_i.
	\end{split}
	\end{equation}
	Finally, in case (3), there must be only one derivative on $\bt^{(n-1)}$ and $\bt^{(n)}$ and hence using \eqref{estN.small}, \eqref{estN.large}, \eqref{estbeta.small} and Lemma~\ref{der}, the term can be bounded by
	\begin{equation}
	\ls \left(1+ |N_{asymp}^{(n-1)}|+ \|\nab\log \wht N^{(n-1)}\|_{H^2_{\de+1}} \right) \|\bt^{(n-1)}\|_{C^1_{\de'+1}}\|\bt^{(n)}\|_{C^1_{\de'+1}} \ls (1 + \ep + C_i)  \ep^2\ls \ep^2 C_i.
	\end{equation}
	Combining these we have
	\begin{equation}\label{gamma.est.lower.III}
	III\ls \ep C(A_0) C_i.
	\end{equation}
	It remains to bound the term $IV$, which is the hardest: it is the term that determines the weight we can put in. By \eqref{estN.small}, \eqref{estN.large}, \eqref{estbeta.large}, Lemma~\ref{der} and Proposition~\ref{holder},
	\begin{equation}\label{gamma.est.lower.IV}
	\begin{split}
	IV \ls & \left(1 + |N^{(n)}_{asymp}|^3 + \|\wht N^{(n)}\|_{W^2_{\de+\f 12, 4}}^3  \right)\|div (e_0^{(n-1)}\beta^{(n)})\|_{H^3_{\de'+1}}  + \|\wht N^{(n)}\|_{W^3_{\de+\f 12, 4}} \|div (e_0^{(n-1)}\beta^{(n)})\|_{L^4_{\de'+\f 32}}\\
	\ls & C(A_0) C_i^2.
	\end{split}
	\end{equation}
This concludes the estimates for \eqref{gamma.est.lower.main}, we summarize it as follows:
	\begin{equation}\label{gamma.est.lower.main.main}
	\begin{split}
	\left\|e^{2\gamma^{(n)}} e^{(n-1)}_0 \left( \frac{e^{-2\gamma^{(n)}}}{N^{(n-1)}}div(\beta^{(n)})\right)\right\|_{H^3_{\delta'+1}}\ls C(A_0) C_i^2.
	\end{split}
	\end{equation}
The fourth term on the RHS of \eqref{gamman}, after multiplying by $\wht N^{(n)}$, can be estimated in a trivial manner using Lemma~\ref{der}, \eqref{estN.small} and \eqref{estN.large}:
	$$\|\Delta N^{(n)}\|_{H^3_{\de'+1}} \ls |N_{asymp}^{(n)}| + \|\wht N^{(n)} \|_{H^5_{\de'-1}}\ls |N_{asymp}^{(n)}| + \|\wht N^{(n)} \|_{H^5_{\de}}\ls  C_i.$$
	Finally, the last two terms on the RHS of \eqref{gamman}, i.e., the terms involving $\rd\phi^{(n)}$ and $F_{\bA}^{(n)}$, are compactly supported, and can be controlled exactly as in the proof of Proposition~\ref{prpN} by
	\begin{equation}
	\begin{split}
	\left\|\f {\delta^{ij}}2\rd_i\phi^{(n)}\rd_j\phi^{(n)} + \f 12\sum_{\bA} e^{4\gamma^{(n)}} (F^{(n)}_{\bA})^2 \de_{ij} (L^{(n)}_{\bA})^i (L^{(n)}_{\bA})^j\right\|_{H^3_{\de'+1}} \ls \ep C(A_0) C_i.
	\end{split}
	\end{equation}

	We now bound the remaining terms on RHS of \eqref{wave.whtgamma} (i.e., those that are not on RHS of \eqref{gamman}), after multiplying by $N^{(n)}$. First, an easy explicit computation, together with \eqref{est.constraints.0}, \eqref{estN.small} and \eqref{estN.large}, show
	\begin{equation}
	\begin{split}
	\|\alpha N^{(n)} \Delta(\chi(|x|){\log}(|x|))\|_{H^3_{\de'+1}}\ls |\alp| (1+|N^{(n)}_{asymp}| + \|\wht N^{(n)}\|_{H^3_{\de}}) \ls \ep^2 C_i.
	\end{split}
	\end{equation}
	For the final term,  we have
	\begin{equation}\label{gamma.cutoff.term}
	\begin{split}
	 \left\|\alp e_0^{(n)} \left(\f{(\beta^{(n)})^i}{N^{(n)}} \rd_i (\chi(|x|){\log}(|x|) ) \right)\right\|_{H^3_{\de'+1}}
	\ls &|\alp| \left(\left\|e_0^{(n)} \f{\beta^{(n)}}{N^{(n)}} \right\|_{H^3_{\de'}} + \left\|\f{\beta^{(n)}\beta^{(n)}}{N^{(n)}}\right\|_{H^3_{\de'}} \right) \ls \ep^2 C(A_0) C_i^2,
	\end{split}
	\end{equation}
	where the estimate is obtained by writing $e_0^{(n)}(\bt^{(n)})^i = e_0^{(n-1)}(\bt^{(n)})^i - (\bt^{(n)} - \bt^{(n-1)})^j \rd_j (\beta^{(n)})^i$ and $e_0^{(n)}N^{(n)} = \rd_t N^{(n)} - (\bt^{(n)})^i \rd_i N^{(n)}$ and using \eqref{est.constraints.0}, \eqref{estN.small}, \eqref{estN.large}, \eqref{estrdtN.large}, \eqref{estbeta.small} and \eqref{estbeta.large}.
		
	We now apply the energy estimate in \eqref{h.wave.higher}. Combining all the estimates above, we have shown that the $N^{(n)}\times \mbox{(RHS of \eqref{wave.whtgamma})}$ is bounded in $H^3_{\de'+1}$ by $C(A_0, A_1, C_i)$. Since $T$ can depend on $C_i$, $A_0$ and $A_1$, by choosing $T$ sufficiently small, \eqref{h.wave.higher} implies \eqref{gamma.est.lower}.
		
	\textbf{Proof of \eqref{gamma.est.higher.2}{.}} First note that by equation \eqref{wave.whtgamma}, we need to control {(1)} $N^{(n)} \times \mbox{(RHS of \eqref{wave.whtgamma})}$ in $H^2_{\de'+1}$ by $C(A_0) C_i^2${, (2) $N^{(n)}\Delta\gamma^{(n+1)}$ in $H^2_{\de'+1}$ by $C(A_0) C_i^2$, and (3) $N^{(n)}\delta^{ik}\beta_k^{(n)} \partial_i \frac{e^{(n)}_0 \wht \gamma^{(n+1)}}{N^{(n)}}$ in $H^2_{\de'+1}$ by $C(A_0) C_i^2$}. {For (1)}, note that in the estimates we proved in the course of {obtaining} \eqref{gamma.est.lower}, indeed all terms {on RHS of \eqref{wave.whtgamma}} satisfy th{e desired} bound. {For $(2)$, we have the desired bound thanks to \eqref{estN.small} and the estimate \eqref{gamma.est.lower} that we just established above.} {Finally, for (3), we follow} the proof of \eqref{estgammabis} {and use additionally \eqref{estN.large} and \eqref{estbeta.large} to} obtain
		\begin{align*}
	&\left\|N^{(n)}\delta^{ik}\beta_k^{(n)} \partial_i\frac{e^{(n)}_0 \wht \gamma^{(n+1)}}{N^{(n)}}\right\|_{H^2_{\delta'+1}}\\
		\lesssim &(1+\|\nabla \log N^{(n)}\|_{H^2_{\delta}}+\|\beta^{(n)}\|_{H^3_{\delta'}})\left(\|\nabla \wht \gamma^{(n+1)}\|_{H^2_{\delta'+1}}+\sum_{|\alpha|\leq 2} \left\|\frac{e_0^{(n)}\nabla^\alpha\wht \gamma^{(n+1)}}{N^{(n)}}\right\|_{L^2_{\delta'+1+|\alpha|}}\right)\\ 
		\lesssim & C_i^2.
	\end{align*}
	
	\textbf{Proof of \eqref{gamma.est.higher.3} {step 1: Estimates for RHS of \eqref{wave.whtgamma} in $H^1_{\de'+1}$.}} Similar to in the proof of \eqref{gamma.est.higher.2}, {we first} bound $N^{(n)} \times \mbox{(RHS of \eqref{wave.whtgamma})}$ {- except that this time we bound it} in $H^1_{\de'+1}$ by $C(A_0) C_i$. Now the bounds in the proof of \eqref{gamma.est.lower} show that it remains to improve the estimates for the term $II$ in \eqref{gamma.est.lower.II}, the term $IV$ in \eqref{gamma.est.lower.IV} and the term \eqref{gamma.cutoff.term} when we replace the $H^3_{\de'+1}$ norm by the $H^1_{\de'+1}$ norm. We first estimate the term analogous to \eqref{gamma.est.lower.II}. By \eqref{estN.small}, \eqref{estN.large}, \eqref{estbeta.small}, Lemma~\ref{der} and Proposition~\ref{holder},
	\begin{equation*}
	\begin{split}
	&\left\| \frac{e^{(n-1)}_0 \log N^{(n-1)}}{N^{(n-1)}} div\,(\beta^{(n)})\right\|_{H^1_{\delta'+1}} \\
	\ls & \left(|\rd_t N^{(n-1)}_{asymp}| + |N_{asymp}^{(n-1)}| \left\|\bt^{(n-1)} (1+|x|)^{-1} \right\|_{C^1_{\de+1}} + \left\|\rd_t \wht N^{(n-1)}\right\|_{C^0_{\de+1}} + \left\|\beta^{(n-1)} \nab \wht N^{(n-1)}\right\|_{C^1_{\de+1}} \right)\\
	&\qquad \times \left( 1+ |N_{asymp}^{(n-1)}| + \left\|\nab\log \wht N^{(n-1)} \right\|_{C^0_{\de+2}}\right) \left\|div\,(\beta^{(n)})\right\|_{H^1_{\delta'+1}} + \left\|\rd_t \wht N^{(n-1)}\right\|_{W^1_{\de+\f 12,4}} \left\|div\,(\beta^{(n)})\right\|_{W^1_{\delta'+\f 32, 4}} \\
	\ls & \ep C_i.
	\end{split}
	\end{equation*}
	Next, we consider the term analogous to \eqref{gamma.est.lower.IV}. By Lemma~\ref{der}, Proposition~\ref{produit}, \eqref{estN.small}, \eqref{estbeta.small} and \eqref{estbeta.large},
	\begin{equation*}
	\begin{split}
	\left\|  \frac{div \,(e^{(n-1)}_0  \beta^{(n)})}{N^{(n-1)}}\right\|_{H^1_{\delta'+1}}
	\ls \left(1 + |N^{(n-1)}_{asymp}| + \|\wht N^{(n-1)}\|_{H^2_\de}\right)\| div \,(e^{(n-1)}_0  \beta^{(n)})\|_{H^1_{\de'+1}} \ls C(A_0) C_i.
	\end{split}
	\end{equation*}
	Finally, for the term analogous to \eqref{gamma.cutoff.term}, we have
	\begin{equation}
	\begin{split}
	 & \left\|\alp e_0^{(n)} \left(\f{(\beta^{(n)})^i}{N^{(n)}} \rd_i (\chi(|x|){\log}(|x|) ) \right)\right\|_{H^1_{\de'+1}} \\
	\ls &|\alp| \left( \left\|e_0^{(n)} \f{\beta^{(n)}}{N^{(n)}} \right\|_{H^1_{\de'}} + \left\|\f{\bt^{(n)}\bt^{(n)}}{N^{(n)}} \right\|_{H^1_{\de'}}\right) \ls \ep^2 C(A_0) C_i,
	\end{split}
	\end{equation}
	where the last estimate is obtained by writing $e_0^{(n)}(\bt^{(n)})^i = e_0^{(n-1)}(\bt^{(n)})^i - (\bt^{(n)} - \bt^{(n-1)})^j \rd_j (\beta^{(n)})^i$ and $e_0^{(n)}N^{(n)} = \rd_t N^{(n)} - (\bt^{(n)})^i \rd_i N^{(n)}$ and using \eqref{est.constraints.0}, \eqref{estN.small}, \eqref{estN.large} and \eqref{estbeta.small}.
	
	\textbf{Proof of \eqref{gamma.est.higher.3} {step 2: Completion of the proof.}}
	{As in the proof of \eqref{gamma.est.higher.2}, it remains to control $N^{(n)}\Delta\gamma^{(n+1)}$ and $N^{(n)}\delta^{ik}\beta_k^{(n)} \partial_i \frac{e^{(n)}_0 \wht \gamma^{(n+1)}}{N^{(n)}}$. The former term can be controlled using \eqref{estN.small} and \eqref{gamma.est.lower} (that we proved above) as follows:
	\begin{equation*}
	\begin{split}
	&\|N^{(n)}\Delta\gamma^{(n+1)}\|_{H^1_{\delta'+1}}\ls (1+|N_{asymp}| + \|\wht N^{(n)} \|_{C^1_{\de+1}})(\ep^2 + \|\nab\gamma^{(n+1)}\|_{H^2_{\de'+1}})\ls \ep^2 + C_i\ls C_i.
	\end{split}
	\end{equation*}}
	{Finally, the remaining term can be estimated using \eqref{estN.small}, \eqref{estbeta.small} and the argument leading to \eqref{estgammabis}:}
		\begin{align*}
			&\left\|N^{(n)}\delta^{ik}\beta_k^{(n)} \partial_i\frac{e^{(n)}_0 \wht \gamma^{(n+1)}}{N^{(n)}}\right\|_{H^1_{\delta'+1}}\\
			\lesssim &(1+\|\nabla \log N^{(n)}\|_{H^1_{\delta}}+\|\beta^{(n)}\|_{H^2_{\delta'}})\left(\|\nabla \wht \gamma^{(n+1)}\|_{H^1_{\delta'+1}}+\sum_{|\alpha|\leq 1} \left\|\frac{e_0^{(n)}\nabla^\alpha\wht \gamma^{(n+1)}}{N^{(n)}}\right\|_{L^2_{\delta'+1+|\alpha|}}\right)\\ 
			\lesssim &(1+C\ep)C_i.
			\end{align*}
This concludes the proof of the proposition.

\end{proof}

\begin{prp} [Estimates for $\tau^{(n+1}$]\label{prptau}
For $n\geq 2$, the following estimates hold:
	\begin{align}
	\label{tau.higher}\|\tau^{(n+1)}\|_{H^3_{\delta''+1}} \leq & C(A_0) C_i, \\
	\label{rdttau.lower}\|\partial_t\tau^{(n+1)}\|_{H^1_{\delta''+1}}\leq & C(A_1) C_i,\\
	\label{rdttau.higher}\|\partial_t\tau^{(n+1)}\|_{H^2_{\delta''+1}} \leq & C(A_1) C_i^2.
	\end{align}
\end{prp}
\begin{proof}
	In view of \eqref{taun}, the estimates for $\tau^{(n+1)}$ can be obtained by directly controlling 
	$$\frac{e^{-2\gamma^{(n)}}}{N^{(n)}}\left( -2e_0^{(n-1)} \gamma^{(n)} +div\, (\beta^{(n)})\right).$$
	Similarly, to bound $\rd_t \tau^{(n+1)}$, it suffices to estimate the $\rd_t$ derivative of the above quantity.
To this end, we use the estimates in \eqref{estN.small}, \eqref{estN.large}, \eqref{estbeta.small}, \eqref{estbeta.large}, \eqref{estgamma} and \eqref{estrdtgamma}. Let us first control the factor $\frac{e^{-2\gamma^{(n)}}}{N^{(n)}}$. Notice that it has growing factors in $|x|^{2\alpha}$ or $N_{asymp}^{(n-1)}{\log}(|x|)$, which ultimately contributes to the fact that we need to worsen the weight in our estimates -- using $\de''$ instead of $\de'$. Nevertheless, for $\ep$ sufficiently small, for $|x|$ large, these growing factors can be controlled by $|x|^{\f{\ep}{10}}$. Hence, by Propositions~\ref{holder}, \ref{produit} and \ref{prpsmallness}, we have
	\begin{equation}\label{tau.metric.lower}
	\left\|\frac{e^{-2\gamma^{(n)}}}{N^{(n)}}\right\|_{H^2_{-\f{\ep}{10}-1}} +  \left\|\frac{e^{-2\gamma^{(n)}} N^{(n-1)}}{N^{(n)}}\right\|_{H^2_{-\f{\ep}{10}-1}} \ls 1.
	\end{equation}
	Also, by \eqref{estN.large}, \eqref{estgamma}, \eqref{estrdtgamma}{,} Proposition{s}~\ref{prpsmallness} {and \ref{prpeogamma}}, as well as Propositions~\ref{holder} and \ref{produit}, we also have
	\begin{equation}\label{tau.metric.higher}
	\left\|\frac{e^{-2\gamma^{(n)}}}{N^{(n)}}\right\|_{H^3_{-\f{\ep}{10}-1}} +  \left\|\frac{e^{-2\gamma^{(n)}} N^{(n-1)}}{N^{(n)}}\right\|_{H^3_{-\f{\ep}{10}-1}} + \left\|\rd_t \frac{e^{-2\gamma^{(n)}}}{N^{(n)}}\right\|_{H^2_{-\f{\ep}{10}-1}} +  \left\|\rd_t \frac{e^{-2\gamma^{(n)}} N^{(n-1)}}{N^{(n)}}\right\|_{H^2_{-\f{\ep}{10}-1}} \ls C_i.
	\end{equation}
Using \eqref{tau.metric.lower} and \eqref{tau.metric.higher} together with Proposition~\ref{produit}, \eqref{estbeta.small}, \eqref{estbeta.large}, \eqref{estgamma}{,} Proposition{s}~\ref{prpsmallness} {and \ref{prpeogamma}}, we prove \eqref{tau.higher}:
\begin{align*}
&	\left\|\frac{e^{-2\gamma^{(n)}}}{N^{(n)}}\left( -2e_0^{(n-1)} \gamma^{(n)} +div(\beta^{(n)})\right)\right\|_{H^3_{\delta''+1}}\\
	\ls & \left\|\frac{e^{-2\gamma^{(n)}} N^{(n-1)}}{N^{(n)}}\right\|_{H^2_{-\f{\ep}{10}-1}} \left\|\f{e_0^{(n-1)} \gamma^{(n)}}{N^{(n-1)}}\right\|_{H^3_{\de'+1}} +  \left\|\frac{e^{-2\gamma^{(n)}}}{N^{(n)}}\right\|_{H^2_{-\f{\ep}{10}-1}} \|div\, \beta^{(n)}\|_{H^3_{\de'+1}}\\
	 & + \left\|\frac{e^{-2\gamma^{(n)}} N^{(n-1)}}{N^{(n)}}\right\|_{H^3_{-\f{\ep}{10}-1}} \left\|\f{e_0^{(n-1)} \gamma^{(n)}}{N^{(n-1)}}\right\|_{C^0_{\de'+2}} +  \left\|\frac{e^{-2\gamma^{(n)}}}{N^{(n)}}\right\|_{H^3_{-\f{\ep}{10}-1}} \|div\, \beta^{(n)}\|_{C^0_{\de'+2}}
	\ls C(A_0) C_i.
	\end{align*}
Next, we prove \eqref{rdttau.lower}. By Proposition~\ref{produit}, \eqref{estbeta.small}, \eqref{estbeta.large}, \eqref{estgamma}, \eqref{estrdtgamma}, \eqref{tau.metric.lower}{,} \eqref{tau.metric.higher} {and Proposition~\ref{prpeogamma}},
\begin{align*} 
	&\left\|\partial_t\left(\frac{e^{-2\gamma^{(n)}}}{N^{(n)}}\left( -2e_0^{(n-1)} \gamma^{(n)} +div(\beta^{(n)})\right)\right)\right\|_{H^1_{\delta''+1}}\\
	\lesssim & \left\|\frac{e^{-2\gamma^{(n)}} N^{(n-1)}}{N^{(n)}}\right\|_{H^2_{-\f{\ep}{10}-1}} \left(\left\|\rd_t \f{e_0^{(n-1)} \wht \gamma^{(n)}}{N^{(n-1)}}\right\|_{H^1_{\de'+1}} + |\alp| \left\|\rd_t \f{\bt^{(n-1)} \nab (\chi(|x|) {\log} (|x|))}{N^{(n-1)}}\right\|_{H^1_{\de'+1}} \right)\\
	& + \left\|\rd_t \frac{e^{-2\gamma^{(n)}} N^{(n-1)}}{N^{(n)}}\right\|_{H^1_{-\f{\ep}{10}-1}} \left(\left\|\f{e_0^{(n-1)} \wht \gamma^{(n)}}{N^{(n-1)}}\right\|_{H^2_{\de'+1}} + |\alp| \left\|\f{\bt^{(n-1)} \nab (\chi(|x|) {\log} (|x|))}{N^{(n-1)}}\right\|_{H^2_{\de'+1}} \right)\\
	& + \left\|\frac{e^{-2\gamma^{(n)}}}{N^{(n)}}\right\|_{H^2_{-\f{\ep}{10}-1}} \left\|\rd_t (div\,\beta^{(n)})\right\|_{H^1_{\de'+1}} + \left\|\rd_t \frac{e^{-2\gamma^{(n)}}}{N^{(n)}}\right\|_{H^2_{-\f{\ep}{10}-1}} \left\|div\,\beta^{(n)}\right\|_{H^1_{\de'+1}}	\ls  C(A_1) C_i.
	\end{align*}

	Finally, we prove \eqref{rdttau.higher}. Again, using Proposition~\ref{produit}, \eqref{estbeta.small}, \eqref{estbeta.large}, \eqref{estgamma}, \eqref{estrdtgamma}, \eqref{tau.metric.lower}{,} \eqref{tau.metric.higher} {and Proposition~\ref{prpeogamma}},
we obtain
\begin{align*} 
	&\left\|\partial_t\left(\frac{e^{-2\gamma^{(n)}}}{N^{(n)}}\left( -2e_0^{(n-1)} \gamma^{(n)} +div(\beta^{(n)})\right)\right)\right\|_{H^2_{\delta''+1}}\\
	\lesssim & \left\|\frac{e^{-2\gamma^{(n)}} N^{(n-1)}}{N^{(n)}}\right\|_{H^2_{-\f{\ep}{10}-1}} \left(\left\|\rd_t \f{e_0^{(n-1)} \wht \gamma^{(n)}}{N^{(n-1)}}\right\|_{H^2_{\de'+1}} + |\alp| \left\|\rd_t \f{\bt^{(n-1)} \nab (\chi(|x|) {\log} (|x|))}{N^{(n-1)}}\right\|_{H^2_{\de'+1}} \right)\\
	& + \left\|\rd_t \frac{e^{-2\gamma^{(n)}} N^{(n-1)}}{N^{(n)}}\right\|_{H^2_{-\f{\ep}{10}-1}} \left(\left\|\f{e_0^{(n-1)} \wht \gamma^{(n)}}{N^{(n-1)}}\right\|_{H^2_{\de'+1}} + |\alp| \left\|\f{\bt^{(n-1)} \nab (\chi(|x|) {\log} (|x|))}{N^{(n-1)}}\right\|_{H^2_{\de'+1}} \right)\\
	& + \left\|\frac{e^{-2\gamma^{(n)}}}{N^{(n)}}\right\|_{H^2_{-\f{\ep}{10}-1}} \left\|\rd_t (div\,\beta^{(n)})\right\|_{H^2_{\de'+1}} + \left\|\rd_t \frac{e^{-2\gamma^{(n)}}}{N^{(n)}}\right\|_{H^2_{-\f{\ep}{10}-1}} \left\|div\,\beta^{(n)}\right\|_{H^2_{\de'+1}}	\ls  C(A_1) C_i^2.
	\end{align*}	
\end{proof}

\begin{prp}[Estimates for $H^{(n+1)}$]\label{prpH}
For $n\geq 2$, the following estimates hold:
\begin{align}
\label{Hnlower} \|e_0^{(n+1)} H^{(n+1)}\|_{H^{3}_{\de+1}}\leq & 20 C_i,\\
\label{Hnfinal} \|H^{(n+1)}\|_{H^3_{\delta+1}}\leq & 2 C_i.
\end{align}
\end{prp}

\begin{proof}
\textbf{Proof of \eqref{Hnlower}.} In order to estimate $(e_0^{(n+1)}) (H^{(n+1)})_{ij}$ in $H^3_{\de+1}$, it obviously suffices to bound the RHS of \eqref{hijn} in the $H^3_{\de+1}$ norm by $20C_i$. We consider each term on RHS of \eqref{hijn}: in fact, all but two terms can be controlled in the stronger $H^3_{\de+2}$ space.

First notice that the terms $2e^{2\gamma^{(n)}}N^{(n)}(H^{(n)})_i{ }^\ell (H^{(n)})_{j\ell}$, $N^{(n)}\partial_i \phi^{(n)} \breve{\otimes}\partial_j \phi^{(n)}$ and $N^{(n)} (F^{(n)}_{\bA})^2  (L^{(n)}_{{\bA}})_i\breve{\otimes}(L^{(n)}_{\bA})_j$ are analogous to terms in \eqref{lapsen} and can be treated as in Proposition~\ref{prpN} so that they are bounded as $\|\cdot \|_{H^3_{\de+2}}\leq \ep C(A_{0}) C_i$; we omit the details.

The remaining terms can be treated as follows. For $\partial_{(j} (\beta^{(n)})^k (H^{(n)})_{i)k}$, we use \eqref{estbeta.small}, \eqref{estbeta.large}, \eqref{estH}, Proposition~\ref{prpsmallness} and Lemma~\ref{der} to get
\begin{equation*}
\begin{split}
&\|\partial_{(j} (\beta^{(n)})^k (H^{(n)})_{i)k} \|_{H^3_{\de+2}}\\
\ls & \|\nab \beta^{(n)}\|_{C^0_{\de'+2}} \|H^{(n)}\|_{H^3_{\de+1}} + \|\nab \beta^{(n)}\|_{H^3_{\de'+1}} \|H^{(n)}\|_{C^0_{\de+2}} \\
& + \|\nab \beta^{(n)}\|_{W^2_{\de'+\f 32, 4}} \|H^{(n)}\|_{W^1_{\de+\f 32, 4}} + \|\nab \beta^{(n)}\|_{W^1_{\de'+\f 32, 4}} \|H^{(n)}\|_{W^2_{\de+\f 32, 4}}\ls \ep C(A_0) C_i.
\end{split}
\end{equation*}
The term $\partial_i \gamma^{(n)} \breve{\otimes } \partial_j N^{(n)}$ can be treated similarly, except for extra care regarding the logarithmically growing terms. More precisely, by H\"older's inequality, \eqref{estN.small}, \eqref{estN.large}, \eqref{estgamma}, Proposition~\ref{prpsmallness} and Lemma~\ref{der}, we have
\begin{equation*}
\begin{split}
&\|\partial_i \gamma^{(n)} \breve{\otimes } \partial_j N^{(n)} \|_{H^3_{\de+1}}\\
\ls & |\alp| |N_{asymp}^{(n)}| \|(1+|x|)^{-2}\|_{H^3_{\de+1}} + |\alp| \|\nab \wht N^{(n)}\|_{H^3_{\de+1}} + |N_{asymp}^{(n)}|\|\nab \wht \gamma^{(n)}\|_{H^3_{\de'+1}} +\|\nab \wht \gamma^{(n)}\|_{C^0_{\de'+2}} \|\nab \wht N^{(n)}\|_{H^3_{\de+1}}\\
& + \|\nab \wht \gamma^{(n)}\|_{H^3_{\de'+1}} \|\nab \wht N^{(n)}\|_{C^0_{\de+2}} + \|\nab \wht\gamma^{(n)}\|_{W^1_{\de'+\f 32, 4}} \|\nab \wht N^{(n)}\|_{W^2_{\de+\f 32, 4}} + \|\nab \wht\gamma^{(n)}\|_{W^2_{\de'+\f 32, 4}} \|\nab \wht N^{(n)}\|_{W^1_{\de+\f 32, 4}}\\
\ls &\ep C_i.
\end{split}
\end{equation*}
Finally, by Lemma~\ref{der}, \eqref{estN.small} and \eqref{estN.large},
$$\| \partial_i \breve{\otimes } \partial_j N^{(n)}\|_{H^3_{\de+1}} \leq 6\| \wht N^{(n)} \|_{H^5_{\de}} + C|N_{asymp}| \leq (12+ \ep C)C_i, $$
where the norm on the LHS is to be understood as the $H^3_{\de+2}$ norm for a $2$-tensor.

Combining all the above estimates and choosing $\ep$ sufficiently small give \eqref{Hnlower}.

\textbf{Proof of \eqref{Hnfinal}.} First note that  $|[e_0^{(n+1)},\nab^\alp] H^{(n+1)}| \ls \sum_{|\alp_1| + |\alp_2| = |\alp|} |\nab^{\alp_1}\nab \bt^{(n+1)}| |\nab^{\alp_2} H^{(n+1)}|$. Hence, for $|\alp|\leq 3$, using Proposition~\ref{holder} and Proposition~\ref{prpbeta}, we have
\begin{equation}\label{Hn.commute}
\|e_0^{(n+1)} \nab^{\alp} H_{ij}^{(n+1)} \|_{L^2_{\de+1+|\alp|}} \ls \|e_0^{(n+1)} H_{ij}^{(n+1)} \|_{H^3_{\de+1}} + C_i \|H^{(n+1)}\|_{H^3_{\de+1}}.
\end{equation}

Next, writing $e_0^{(n+1)} \nab^\alp H_{ij}^{(n+1)} = (\rd_t -(\bt^{(n+1)})^k\rd_k) \nab^{\alp} H_{ij}^{(n+1)}$, squaring the expression, multiplying by $(1+|x|^2)^{\de+2+|\alp|}$, integrating with respect to $dx\,dt$, integrating by parts and using the estimates for $\beta^{(n+1)}$ in Proposition~\ref{prpbeta}, we have
\begin{equation*}
\begin{split}
\sup_{t\in [0, T]} \|\nab^{\alp} H^{(n+1)}\|_{L^2_{\de+2+|\alp|}}(t)\leq  &\|\nab^{\alp} H^{(n+1)}\|_{L^2_{\de+2+|\alp|}}(0) + C T \sup_{t\in [0,T]} \| e_0^{(n+1)} \nab^\alp H_{ij}^{(n+1)}\|_{L^2_{\de+2+|\alp|}}(t)\\
\leq & \|\nab^{\alp} H^{(n+1)}\|_{L^2_{\de+2+|\alp|}}(0) + C C_i T ( 1 + \|H^{(n+1)}\|_{H^3_{\de+1}}),
\end{split}
\end{equation*}
where in the last inequality we have used \eqref{Hn.commute} and \eqref{Hnlower}. 

Now, summing over all $|\alp|\leq 3$, choosing $T$ sufficiently small, and absorbing the term $C C_i T \|H^{(n+1)}\|_{H^3_{\de+1}}$ to the LHS, we obtain \eqref{Hnfinal}.
\end{proof}

\begin{lm}[Suppport of $\phi^{(n+1)}$]\label{lm:cpt.supp}
There exists a constant $C_s > 0$ such that for $\ep$, $T$ sufficiently small (depending on $R$) and $n\geq 2$, $\phi^{(n+1)}$ is supported in the set $\{(t,x)\in [0,T]\times \mathbb R^2: C_s(1+R^\ep) t - |x|\geq -R\} $. In particular, choosing $T$ smaller if necessary, $supp(\phi^{(n+1)}) \subset \{(t,x)\in [0,T]\times \mathbb R^2: |x| \leq 2R\}$.
\end{lm}
\begin{proof}
Since the initial data for $\phi^{(n+1)}$ and $\rd_t\phi^{(n+1)}$ is compactly supported in $|x| \leq R$, it suffices to show that $\{(t,x)\in [0,T]\times \mathbb R^2: C_s(1+R^\ep)t - |x| = -R\}$ is a spacelike hypersurface with respect to $g^{(n)}$. We compute using \eqref{g.inverse}
\begin{equation}\label{spacelike.calculation}
\begin{split}
&(g^{(n)})^{-1}\left(d\left(C_s(1+R^\ep)t-|x|\right), d\left(C_s(1+R^\ep)t-|x|\right)\right)\\
=& -\f {C_s^2(1+R^\ep)^2}{(N^{(n)})^2}-\f{2C_s(1+R^\ep)(x\cdot \beta^{(n)})}{|x| (N^{(n)})^2}+e^{-2\gamma^{(n)}}-\f{(x\cdot \beta^{(n)})^2}{|x|^2 (N^{(n)})^2}. 
\end{split}
\end{equation}
For $|x| \geq 2$,
$e^{-2\gamma^{(n)}}=e^{2\alp \chi(|x|){\log}(|x|)}e^{-2\wht \gamma^{(n)}}{\ls |x|^{C\ep^2}}$, $\f {1}{(N^{(n)})^2}{\gtrsim \min\{1, \f{1}{\ep^2\log |x|}\}}$, and $\f{|x\cdot \beta^{(n)}|}{|x| (N^{(n)})^2},\, \f{(x\cdot \beta^{(n)})^2}{|x|^2 (N^{(n)})^2} \ls \ep$. Hence, {after} choosing the parameters appropriately, one easily sees that \eqref{spacelike.calculation} is non-positive.
\end{proof}

\begin{prp}[Estimates for $\phi^{(n+1)}$]\label{prpphi}
	For $n\geq 2$, the following estimate holds:
	$$\|\partial \phi^{(n+1)}\|_{H^3}+\left\|\partial_t \left(\frac{e_0^{(n)} \phi^{(n+1)}}{N^{(n)}}\right) \right\|_{H^{2}}\ls C_i.$$
\end{prp}

\begin{proof}
	We perform the energy estimate for the wave equation.
First, note that since $\phi^{(n+1)}$ is compactly supported in $B(0,2R)$ for all time by Lemma~\ref{lm:cpt.supp}, we do not need to worry about the spatial decay. 
	
	Given a function $f$, define a $2$-tensor $Q^{(n)}$ as follows:
$$Q_{\alpha \beta}^{(n)}[f] :=\partial_\alpha  f \partial_\beta  f -\frac{1}{2}g^{(n)}_{\alpha \beta}((g^{(n)})^{-1})^{\rho\sigma}\partial_\sigma 
 f \partial_\rho  f .$$
An easy computation shows that 
$$((g^{(n)})^{-1})^{\mu\alpha} D^{(n)}_{\mu} Q^{(n)}_{\alpha \beta}[f]= (\rd_\beta f)(\Box_{g^{(n)}} f),$$
where $D^{(n)}$ is the Levi-Civita connection associated to $g$. Defining
$$^{(\rd_t)}\pi^{(n)}_{\alpha\beta}=D^{(n)}_\alpha (\rd_t)_\beta+D^{(n)}_\beta (\rd_t)_\alpha,$$
we have by Stoke's theorem that for every $t\in (0,T]$,
\begin{equation}\label{EE.n}
\begin{split}
&\int_{\Sigma_t} Q^{(n)}[f](\rd_t, \f{1}{N^{(n)}} e_0^{(n)})(t,x)\, \sqrt{|\det \bar{g}^{(n)}|}\, dx\\
=&\int_{\Sigma_0} Q^{(n)}[f](\rd_t, \f{1}{N^{(n)}} e_0^{(n)})(0,x)\, \sqrt{|\det \bar{g}^{(n)}|}\, dx\\
&-\int_0^t \int_{\Sigma_{t'}} \left((\rd_t f) (\Box_{g^{(n)}} f) +\f 12 Q^{(n)}_{\alpha\beta}[f] { }^{(\rd_t)}(\pi^{(n)})^{\alpha\beta}\right)(t',x)\, \sqrt{|\det g^{(n)}|}\, dx\, dt',
\end{split}
\end{equation}
where $\bar{g}^{(n)}$ is as in \eqref{g.form.0}.
	We now apply \eqref{EE.n} to $\phi^{(n+1)}$ and its derivatives. The key point here is to note that by \eqref{estN.small}, \eqref{estbeta.small} and Proposition~\ref{prpsmallness}, the metric components have appropriate \emph{smallness} in the $C^0$ norm on $B(0,2R)$, and therefore on the compact set $B(0,2R)$, for $\ep$ sufficiently small,
	$$(1-C\ep) \leq \sqrt{|\det \bar{g}^{(n)}|}\leq (1+ C\ep),\quad (1-C\ep) \leq \sqrt{|\det {g}^{(n)}|}\leq (1+ C\ep).$$
	On the other hand, since $|{ }^{(\rd_t)}(\pi^{(n)})^{\alpha\beta}|$ is controlled by the $C^1$ norm of the metric, by the estimates in \eqref{estN.small}, \eqref{estN.large}, \eqref{estbeta.small}, \eqref{estbeta.large} and \eqref{estgamma}, $|{ }^{(\rd_t)}(\pi^{(n)})^{\alpha\beta}|\ls C(A_0, C_i)$ on $B(0,2R)$. Therefore, 
		$$\sup_{t\in [0,T]} \|\rd\phi^{(n+1)}\|_{L^2}(t) \leq 2\|\rd\phi^{(n+1)}\|_{L^2}(0) + C(A_0, C_i) T \sup_{t\in [0,T]} \|\rd\phi^{(n+1)}\|_{L^2}(t)\leq 3C_i,$$
		after choosing $T$ to be sufficiently small.
		
	To obtain up to the $H^3$ estimates for $\rd \phi^{(n+1)}$, however, we need to differentiate the equation with respect to spatial derivatives and this leads to higher derivatives of the metric components. Nevertheless, even though these higher derivative terms are no longer small (and are in general only bounded by constants depending on $A_0$ and $C_i$), these terms only appear as \emph{inhomogeneous} terms in the wave equation. Hence, by choosing $T$ sufficiently small, we obtain
	$$ \sup_{t\in [0,T]} \|\rd \phi^{(n+1)}\|_{H^3}(t)  \leq 2 \|\rd \phi^{(n+1)}\|_{H^3}(0) + C(A_0, C_i) T \sup_{t\in [0,T]} \|\rd\phi^{(n+1)}\|_{H^3}(t) \leq 3C_i.$$
	
	Finally, in order to control $\partial_t \left(\frac{e_0^{(n)} \phi^{(n+1)}}{N^{(n)}}\right)$, notice that 
	the equation for $\phi^{(n+1)}$ is given in coordinates as follows:
	\begin{equation}\label{waven}
	-\frac{e^{-2\gamma^{(n)}}}{N^{(n)}}e_0^{(n)}\left(\frac{e^{2\gamma^{(n)}}}{N^{(n)}} e_0^{(n)} \phi^{(n+1)}\right)+\frac{1}{N^{(n)} e^{2\gamma^{(n)}}}div(N^{(n)}\nabla \phi^{(n+1)})+\f{e^{2\gamma^{(n)}}(div\, \bt^{(n)})}{N^{(n)}}(e_0^{(n)} \phi^{(n+1)})=0.
	\end{equation}
	Therefore,
	\begin{equation}\label{waven.2}
	\begin{split}
	\rd_t \left(\frac{e_0^{(n)} \phi^{(n+1)}}{N^{(n)}}\right)= & (\bt^{(n)})^i\rd_i\left(\frac{e_0^{(n)} \phi^{(n+1)}}{N^{(n)}}\right) + \frac{div(N^{(n)}\nabla \phi^{(n+1)})}{N^{(n)} e^{2\gamma^{(n)}}} +\f{e^{2\gamma^{(n)}}(-2e_0^{(n)}\gamma^{(n)}+ div\, \bt^{(n)})}{N^{(n)}}(e_0^{(n)} \phi^{(n+1)}).
	\end{split}
	\end{equation}
We directly bound each term on the RHS of \eqref{waven.2} in $H^2$. The key point in handling these terms is to notice that upon expanding the derivatives, the only way that $\wht N^{(n)}$, $\bt^{(n)}$ or $\wht \gamma^{(n)}$ has three spatial derivatives (or $e_0^{(n)}\gamma^{(n)}$ has two spatial derivatives) is when $\phi^{(n+1)}$ has at most one derivative. In that case, we can bound the first derivative of $\phi^{(n+1)}$ in $L^\infty$ by Proposition~\ref{prpsmallness} independent of $C_i$. In the case where we do not have the highest derivative on the metric components, we can use \eqref{estN.small}, \eqref{estbeta.small} and Proposition~\ref{prpsmallness} to control the metric components independent of $C_i$ and use \eqref{estphi} to estimate the scalar field. Let us consider a typical term. By \eqref{estN.small}, \eqref{estN.large}, \eqref{estbeta.small}, \eqref{estbeta.large}, \eqref{estphi} and Proposition~\ref{prpsmallness}, 
\begin{equation*}
\begin{split}
& \left\|\bt^{(n)} \nab \left(\frac{1}{N^{(n)}}e_0^{(n)} \phi^{(n+1)}\right)\right\|_{H^2}\\
\ls & \|\bt^{(n)}\|_{W^2_{0,4}(B(0,R))}(1 + |N^{(n)}_{asymp}| + \|\wht N^{(n)}\|_{W^2_{0,4}(B(0,R))})(1+\|\bt^{(n)}\|_{W^2_{0,4}(B(0,R))}) \|\rd \phi^{(n+1)}\|_{H^3} \\
& + (\|\wht N^{(n)}\|_{H^3(B(0,R))} + \|\bt^{(n)}\|_{H^3(B(0,R))} )(1 + \|\bt^{(n)}\|_{L^\infty(B(0,R))})\|\rd\phi\|_{L^\infty}\\
\ls & \ep A_0 C_i + \ep^2 C_i \ls \ep C(A_0) C_i\ls C_i.
\end{split}
\end{equation*}
The other terms can be estimated in a similar manner. \qedhere
\end{proof}

\begin{lm}\label{trans.est}
Let $h$ satisfy the following transport equation with an inhomogeneous term $f$ for some $\bA$:
\begin{equation}\label{gen.transport}
L_{\bA}^{(n)} h= f.
\end{equation}
Then, $h$ obeys the estimate
$$\sup_{t\in [0,T]}\int_{\Sigma_t} (1+|x|^2)^{\sigma} h^2 \, dx\leq C(\sigma)\left( \int_{\Sigma_0} (1+|x|^2)^{\sigma} h^2 \, dx +\int_0^T \int_{\Sigma_t} (1+|x|^2)^{\sigma+\f{\ep}{10}} f^2 \, dx\, dt\right),$$
where $C(\sigma)$ is a constant depending on $\sigma$, in addition to $C_{eik}$, $\de$ and $R$.
\end{lm}
\begin{proof}
Decompose $L_{\bA}^{(n)}$ with respect to $\{\rd_t, \rd_i\}$, i.e.,
\begin{equation}\label{L.decomp}
L_{\bA}^{(n)}=(L_{\bA}^{(n)})^t \rd_t + (L_{\bA}^{(n)})^i \rd_i.
\end{equation}
\eqref{gen.transport} can be written as
\begin{equation}\label{gen.transportcoord}
(L_{\bA}^{(n)})^t\rd_t h+(L_{\bA}^{(n)})^i\rd_i h = f.
\end{equation}
Let $w(|x|)=(1+|x|^2)^{\sigma}$. Multiplying \eqref{gen.transportcoord} by $(e^{\gamma^{(n-1)}} N^{(n-1)}) wh$ and integrating in spacetime with respect to $dx\, dt$, we obtain
\begin{equation*}
\begin{split}
&\frac{1}{2}\int_0^t \int_{\Sigma_{t'}} \left(w	(e^{\gamma^{(n-1)}} N^{(n-1)}) \left((L^{(n)}_{\bA})^t \rd_t h^2+(L_{\bA}^{(n)})^i\rd_i h^2\right)\right)\, dx\, dt'\\
= & \int_0^t \int_{\Sigma_{t'}} f w (e^{\gamma^{(n-1)}} N^{(n-1)}) h \, dx\,dt' .
\end{split}
\end{equation*}
This yields
\begin{equation*}
\begin{split}
&\f 12\frac{d}{dt} \int_{\Sigma_{t}} w(L^{(n)}_{\bA})^t (e^{\gamma^{(n-1)}} N^{(n-1)}) h^2 \, dx \\
= & \f 12\int_{\Sigma_t}  \partial_t\left(e^{\gamma^{(n-1)}} N^{(n-1)}(L^{(n)}_{\bA})^t\right) w h^2 + \f 12\int_{\Sigma_t} \partial_i\left((L^{(n)}_{\bA})^i w e^{\gamma^{(n-1)}} N^{(n-1)}\right) h^2 \, dx \\
&+ \int_{\Sigma_{t}} f w e^{\gamma^{(n-1)}} N^{(n-1)} h \, dx .
\end{split}
\end{equation*}
The conclusion follows from the Cauchy--Schwarz inequality, the bounds \eqref{estN.small}, \eqref{estN.large}, \eqref{estgamma}{, Proposition~\ref{prpeogamma}} (and Proposition~\ref{holder}), and the observations that $e^{\gamma^{(n-1)}} N^{(n-1)} \ls (1+|x|^2)^{\f\ep{10}}$ (by \eqref{est.constraints.0}, \eqref{estN.small}, \eqref{estgamma}, Proposition~\ref{prpsmallness}) and $(L^{(n)}_{\bA})^t (e^{\gamma^{(n-1)}} N^{(n-1)})\gtrsim 1$ (by \eqref{estu.lower}).

\end{proof}

\begin{lm}\label{trans.est.higher}
Suppose $h$ and $f$ satisfy \eqref{gen.transport}. Then, for $\ell = 2,3$, $t\in [0,T]$, and for any $\sigma \in (-1,0)$, we have
$$\| h \|_{H^\ell_{\sigma}}(t) \ls  \| h \|_{H^\ell_{\sigma}}(0) + C(A_0, C_i) \int_0^T \|f \|_{H^\ell_{\sigma+\f{\ep}{10}}}(t') \, dt'.$$
\end{lm}
\begin{proof}
The $\ell=3$ case is harder, so we only consider that case. Let $\alp$ be a spatial multi-index with $|\alp|\leq 3$. Clearly, we have
\begin{equation}\label{trans.est.higher.0}
(L_{\bA}^{(n)})^\rho\rd_\rho(\nab^{\alp} h)= \nab^{\alp} f - [\nab^{\alp}, (L_{\bA}^{(n)})^\rho\rd_\rho]h.
\end{equation}
To compute the commutator, we consider separately\footnote{This is because in the $H^3_{\sigma}$ norm in the statement of the lemma, we only allow spatial derivatives.} the cases $\rho= t$ and $\rho = i$. Denoting by $\overline{L_{\bA}^{(n)}}$ the spatial part of $L_{\bA}^{(n)}$, we have
\begin{equation}\label{trans.est.higher.1}
\begin{split}
\left| [\nab^{\alp}, (L_{\bA}^{(n)})^t\rd_t] h \right| 
\ls &\sum_{|\alp_1| + |\alp_2| + |\alp_3| \leq |\alp| - 1} \left|\nab^{\alp_1}\nab \left( \log (L_{\bA}^{(n)})^t\right)\right| |\nab^{\alp_2} \overline{L_{\bA}^{(n)}}| |\nab^{\alp_3} \nab h| \\
& + \sum_{|\alp_1| + |\alp_2| \leq |\alp| - 1} \left|\nab^{\alp_1}\nab \left( \log (L_{\bA}^{(n)})^t\right)\right| |\nab^{\alp_2} f|,
\end{split}
\end{equation}
and 
\begin{equation}\label{trans.est.higher.2}
\left|[\nab^{\alp}, (L_{\bA}^{(n)})^i\rd_i] h \right|\ls \sum_{|\alp_1|+|\alp_2| = |\alp|-1} |\nab^{\alp_1}\nab \overline{L_{\bA}^{(n)}} | |\nab^{\alp_2} \nab h|.
\end{equation}
Here, in \eqref{trans.est.higher.1}, we have used the equation \eqref{gen.transportcoord}.

Hence, applying Lemma~\ref{trans.est} to $\nab^\alp h$ (instead of $h$) with $\sigma + |\alp|$ in place of $\sigma$ in the weight function $w$, using \eqref{trans.est.higher.0}, \eqref{trans.est.higher.1} and \eqref{trans.est.higher.2}, and summing over all $|\alp|\leq 3$, we obtain
\begin{equation}\label{trans.est.higher.main}
\begin{split}
 \|h\|_{H^3_{\sigma}}(t) 
\ls & \|h\|_{H^3_{\sigma}}(0) +\int_0^t \|f\|_{H^3_{\sigma+\f{\ep}{10}}} \, dt' + \underbrace{\int_0^t \left\|\sum_{|\alp_1|+|\alp_2| = |\alp|-1} \left|\nab^{\alp_1}\nab \left( \log (L_{\bA}^{(n)})^t\right) \right| |\nab^{\alp_2} f| \right\|_{L^2_{\sigma+\f{\ep}{10}+|\alp|}}(t')\, dt'}_{=:I}\\
&+\underbrace{\int_0^t \left\| \sum_{|\alp_1|+|\alp_2| = |\alp|-1} |\nab^{\alp_1}\nab \overline{L_{\bA}^{(n)}} | |\nab^{\alp_2} \nab h| \right\|_{L^2_{\sigma+\f{\ep}{10}+|\alp|}}(t')\, dt'}_{=:II}\\
& + \underbrace{\int_0^t \left\| \sum_{|\alp_1|+|\alp_2|+|\alp_3| = |\alp|-1} |\nab^{\alp_1} \nab \left( \log (L_{\bA}^{(n)})^t\right)||\nab^{\alp_2} \overline{L_{\bA}^{(n)}} | |\nab^{\alp_3} \nab h| \right\|_{L^2_{\sigma+\f{\ep}{10}+|\alp|}}(t')\, dt'}_{=:III}.
\end{split}
\end{equation}

To proceed, we note that using the estimate \eqref{estu}, together with the bounds for $N^{(n-1)}$ and $\gamma^{(n-1)}$ in \eqref{estN.large} and \eqref{estgamma}, and Lemma~\ref{der}, Propositions~\ref{holder} and \ref{produit2}, we have
\begin{equation}\label{L.der.est.pro}
\left\|\nab \left( \log (L_{\bA}^{(n)})^t\right)\right\|_{H^2_{\de'''+1}} + \|\nab \overline{L_{\bA}^{(n)}}\|_{H^2_{\de'''+1}} \ls C(A_0, C_i).
\end{equation}
Here, note in particular there are terms growing as $|x|\to \infty $ in $N^{(n-1)}$ and $\gamma^{(n-1)}$ so that we need to use $\de'''$ instead of $\de''$ in the weights. Therefore, using Proposition~\ref{produit},
\begin{equation*}
\begin{split}
I \ls & \|\nab \left( \log (L_{\bA}^{(n)})^t\right)\|_{H^2_{\de'''+1}} \| f\|_{H^3_{\sigma+\f{\ep}{10}}} \ls C(A_0, C_i)\| f\|_{H^3_{\sigma+\f{\ep}{10}}}.
\end{split}
\end{equation*}
Similarly, for $II$ and $III$, we can use \eqref{L.der.est.pro} and Proposition~\ref{produit} to get
$$ II \ls  \|\nab \overline{L_{\bA}^{(n)}}|_{H^2_{\de'''+1}} \| h\|_{H^3_{\sigma}} \ls C(A_0, C_i)\| h\|_{H^3_{\sigma}}$$
and
\begin{equation*}
\begin{split}
III \ls & \left\|\nab \left( \log (L_{\bA}^{(n)})^t\right) \right\|_{H^2_{\de'''+1}} (1 + \|\nab \overline{L_{\bA}^{(n)}}\|_{H^2_{\de'''+1}}) \| h\|_{H^3_{\sigma}} \ls C(A_0, C_i)\| h\|_{H^3_{\sigma}}.
\end{split}
\end{equation*}
Notice here that in $III$, there is a potentially growing factor of $\overline{L_{\bA}^{(n)}}$, but the weights are strong enough to handle it, as long as $\ep$ is sufficiently small. Plugging in the estimates for $I$, $II$ and $III$ into \eqref{trans.est.higher.main}, we thus obtain
$$\sup_{t\in [0,T]}\|h\|_{H^3_{\sigma}}(t) 
\ls  \|h\|_{H^3_{\sigma}}(0) +C(A_0, C_i) \int_0^T \left( \|f\|_{H^3_{\sigma+\f{\ep}{10}}} + \|h\|_{H^3_{\sigma}}(t)\right)(t) \, dt. $$
The conclusion therefore follows from Gr\"onwall's inequality.
\end{proof}

\begin{prp}[Estimates for $L^{(n+1)}_\bA$]\label{prpL}

	For $n\geq 2$, the following estimates hold:
	\begin{align}
	\left\|e^{2\gamma^{(n)}} (L^{(n+1)}_{\bA})^i+\overrightarrow{c_{\bA}}^i \right\|_{H^2_{\delta''}}+\left\|N^{(n)}e^{\gamma^{(n)}} (L^{(n+1)}_{\bA})^t - |\overrightarrow{c_{\bA}}| \right\|_{H^2_{\delta''}} & \ls C_{eik}, \label{prpL.0} \\
	\left\|e^{2\gamma^{(n)}} (L^{(n+1)}_{\bA})^i+\overrightarrow{c_{\bA}}^i \right\|_{H^3_{\delta''}}+\left\|N^{(n)}e^{\gamma^{(n)}} (L^{(n+1)}_{\bA})^t - |\overrightarrow{c_{\bA}}| \right\|_{H^3_{\delta''}} &\ls A_0 C_{i}, \label{prpL.1} \\
	\|\partial_t \left(e^{2\gamma^{(n)}} (L^{(n+1)}_{\bA})^i\right)\|_{H^2_{\delta'''}} + \|\partial_t \left(N^{(n)}e^{\gamma^{(n)}} (L^{(n+1)}_{\bA})^t\right)\|_{H^2_{\delta'''}} &\ls A_0 C_{i}.\label{prpL.2}
	\end{align}
\end{prp}
\begin{proof}
For this proof, it is convenient to write $L^{(n+1)}_\bA$ in the basis $\{e_0^{(n)}, \rd_i\}$. For this we use the notation 
$$L^{(n+1)}_\bA = ( \wht {L^{(n+1)}_\bA} )^0 e_0^{(n)} + ( \wht {L^{(n+1)}_\bA} )^i \rd_i.$$
One checks that
\begin{equation}\label{L.transform}
(L^{(n+1)}_\bA)^t = (\wht {L^{(n+1)}_\bA})^0,\quad (L^{(n+1)}_\bA)^i = (\wht {L^{(n+1)}_\bA})^i - (\bt^{(n)})^i (\wht {L^{(n+1)}_\bA})^0.
\end{equation}
We similarly decompose $L^{(n)}_\bA$ with respect to $\{e_0^{(n)}, \rd_i\}$ (instead of $\{e_0^{(n-1)}, \rd_i\}$) and define $\wht {L^{(n)}_\bA}$ analogously.

\textbf{Proof of \eqref{prpL.0} and \eqref{prpL.1}.} We first estimate the $(\wht {L^{(n+1)}_\bA})^0$ component, which satisfies
\begin{equation}\label{L.coord.eqn}
(\wht {L^{(n)}_\bA})^\alp e_{\alp}^{(n)} (\wht {L^{(n+1)}_\bA})^0= -(\wht \Gamma^{(n)})^{0}_{\alp\bt} (\wht {L^{(n)}_\bA})^\alp (\wht {L^{(n+1)}_\bA})^\bt,
\end{equation}
where $e_i^{(n)} = \rd_i$ and $(\wht \Gamma^{(n)})^{\mu}_{\alp\bt}$ is defined by $D^{(n)}_{e_{\alp}} e_{\bt}= (\wht \Gamma^{(n)})_{\alp\bt}^\mu e_{\mu}$, which are given by \eqref{connections}. 

According to \eqref{connections} (applied to $g^{(n)}$), and the estimates in \eqref{estN.small}, \eqref{estN.large}, \eqref{estbeta.small}, \eqref{estbeta.large}{,} \eqref{estgamma} {and Proposition~\ref{prpeogamma}}, the worst component {(from the point of view of the weights)} of $(\wht \Gamma^{(n)})^{\mu}_{\alp\bt}$ is $(\wht \Gamma^{(n)})^{0}_{00} = e_0^{(n)} \log N^{(n)}$; and all the remaining components have $H^3_{\de'}$ norm bounded above by $\ls C(A_1) C_i$. (To see this, simply notice that if there is a spatial (as opposed to $e_0^{(n)}$) derivatives of the metric components have better spatial decay, and that $e_0\gamma$ also has better spatial decay since $\alp$ is independent of $t$.)

Now, in \eqref{L.coord.eqn}, the worst component $(\wht \Gamma^{(n)})^{0}_{00}$ indeed appears on the RHS. Nevertheless, if we consider the equation instead for 
\begin{equation}\label{L.coord.eqn.renormalized}
\begin{split}
& (\wht {L^{(n)}_\bA})^\alp e_{\alp}^{(n)} \left( e^{\gamma^{(n)}} N^{(n)} (\wht {L^{(n+1)}_\bA})^0\right)\\
= & e^{\gamma^{(n)}} N^{(n)}\times \mbox{(RHS of \eqref{L.coord.eqn})} + (\wht {L^{(n)}_\bA})^\alp \left(e_{\alp}^{(n)} \left(e^{\gamma^{(n)}} N^{(n)}\right) \right)(\wht {L^{(n+1)}_\bA})^0,
\end{split}
\end{equation}
we cancel off the term $(e_0 \log N^{(n)}) (\wht {L^{(n)}_\bA})^0 (\wht {L^{(n+1)}_\bA})^0$ (and the other terms that are introduced also take the form of $(\wht {L^{(n)}_\bA})^\alp (\wht {L^{(n+1)}_\bA})^{\bt}$ multiplied by an $H^3_{\de'}$ function.)

Next, since $\overrightarrow c_\bA$ is a constant vector, we can rewrite \eqref{L.coord.eqn.renormalized} as
\begin{equation}\label{L.coord.eqn.renormalized.1}
(\wht {L^{(n)}_\bA})^\alp e_{\alp}^{(n)} \left( e^{\gamma^{(n)}} N^{(n)} (\wht {L^{(n+1)}_\bA})^0 - |\overrightarrow {c_{\bA}}|\right) = \mbox{(RHS of \eqref{L.coord.eqn.renormalized})}.
\end{equation} 
On the other hand, on the RHS we can write the components of $\wht {L^{(n)}_\bA}$ as
$$(\wht {L^{(n)}_\bA})^i = (L^{(n)}_\bA)^i + e^{-2\gamma^{(n-1)}}\overrightarrow {c_{\bA}}^i - e^{-2\gamma^{(n-1)}}\overrightarrow {c_{\bA}}^i + (\bt^{(n)})^i (L^{(n)}_\bA)^t,$$
 $$(\wht {L^{(n)}_\bA})^0 = ({L^{(n)}_\bA})^t - (N^{(n-1)})^{-1}e^{-\gamma^{(n-1)}}|\overrightarrow {c_{\bA}}| + (N^{(n-1)})^{-1}e^{-\gamma^{(n-1)}}|\overrightarrow {c_{\bA}}|;$$ and write the components of $(\wht {L^{(n+1)}_\bA})^\alp$ as $$(\wht {L^{(n+1)}_\bA})^i = (\wht {L^{(n+1)}_\bA})^i + e^{-2\gamma^{(n)}}\overrightarrow {c_{\bA}}^i - e^{-2\gamma^{(n)}}\overrightarrow {c_{\bA}}^i,$$
 $$(\wht {L^{(n+1)}_\bA})^0 = (\wht {L^{(n+1)}_\bA})^0 - (N^{(n)})^{-1}e^{-\gamma^{(n)}}|\overrightarrow {c_{\bA}}| + (N^{(n)})^{-1}e^{-\gamma^{(n)}}|\overrightarrow {c_{\bA}}|;$$
  and use the triangle inequality. Therefore, we conclude using the estimates for $(\wht \Gamma^{(n)})^{\mu}_{\alp\bt}$ we mentioned above and \eqref{estN.large}, \eqref{estbeta.large}, \eqref{estgamma}, \eqref{estu} that for $\ell=2,3$, the RHS of \eqref{L.coord.eqn.renormalized.1} is bounded above in the $H^\ell_{\de''+\f{\ep}{10}}$ norm as follows:
\begin{equation}\label{L.coord.eqn.renormalized.bound}
\begin{split}
& \left\| (\mbox{RHS of \eqref{L.coord.eqn.renormalized.1}})\right\|_{H^\ell_{\de''+\f{\ep}{10}}} \\
\ls & C(A_1, C_i) \left(1 + \|e^{2\gamma^{(n-1)}} (L_{\bA}^{(n)})^i + \overrightarrow{c_{\bA}}^i\|_{H^3_{\de''}}+\|N^{(n-1)} e^{\gamma^{(n-1)}} (L_{\bA}^{(n)})^t - |\overrightarrow{c_{\bA}}|\|_{H^3_{\de''}} \right.\\
& \qquad\left. + \|N^{(n-1)} e^{\gamma^{(n-1)}} (\bt^{(n)})^i (L_{\bA}^{(n)})^t \|_{H^3_{\de''}}\right) \\
&\times \left(1 + \|e^{2\gamma^{(n)}} (\wht{L_{\bA}^{(n+1)}})^i + \overrightarrow{c_{\bA}}^i\|_{H^\ell_{\de''}}+\|N^{(n)} e^{\gamma^{(n)}} (\wht{L_{\bA}^{(n+1)}})^0 - |\overrightarrow{c_{\bA}}|\|_{H^\ell_{\de''}}\right)\\
\ls & C(A_1, C_i)\left(1 + \|e^{2\gamma^{(n)}} (\wht{L_{\bA}^{(n+1)}})^i + \overrightarrow{c_{\bA}}^i\|_{H^\ell_{\de''}}+\|N^{(n)} e^{\gamma^{(n)}} (\wht{L_{\bA}^{(n+1)}})^0 - |\overrightarrow{c_{\bA}}|\|_{H^\ell_{\de''}}\right).
\end{split}
\end{equation}
Here, we used the weighted space $H^\ell_{\de''+\f{\ep}{10}}$ instead of $H^\ell_{\de'}$ to handle the logarithmically diverging terms $e^{-2\gamma^{(n-1)}}$, $e^{-2\gamma^{(n)}}$, etc.

By Lemma~\ref{trans.est.higher}, \eqref{L.init.bd}, \eqref{L.coord.eqn.renormalized.1} and \eqref{L.coord.eqn.renormalized.bound}, we have, for $\ell=2,3$,
\begin{equation}\label{Lest.final.1}
\begin{split}
\sup_{t\in [0,T]} & \left\|e^{\gamma^{(n)}} N^{(n)} (\wht {L^{(n+1)}_\bA})^0 - |\overrightarrow {c_{\bA}}| \right\|_{H^\ell_{\delta''}}(t) \ls  \left\|e^{\gamma^{(n)}} N^{(n)} (\wht {L^{(n+1)}_\bA})^0 - |\overrightarrow {c_{\bA}}| \right\|_{H^\ell_{\delta''}}(0) \\
& +  C(A_1, C_i) T\sup_{t\in [0,T]}\left(1 + \|e^{2\gamma^{(n)}} (\wht {L^{(n+1)}_\bA})^i + \overrightarrow{c_{\bA}}^i\|_{H^\ell_{\de''}}+\|N^{(n)} e^{\gamma^{(n)}} (\wht {L^{(n+1)}_\bA})^0 - |\overrightarrow{c_{\bA}}|\|_{H^\ell_{\de''}}\right)(t).
\end{split}
\end{equation}
An entirely analogous argument for the equation of $e^{2\gamma^{(n)}} (\wht {L^{(n+1)}_{\bA}})^i+\overrightarrow{c_{\bA}}^i$ instead of $e^{\gamma^{(n)}} N^{(n)} (\wht {L^{(n+1)}_\bA})^0 - |\overrightarrow {c_{\bA}}|$ implies that for $\ell=2,3$,
\begin{equation}\label{Lest.final.2}
\begin{split}
\sup_{t\in [0,T]} & \left\|e^{2\gamma^{(n)}} (\wht {L^{(n+1)}_{\bA}})^i+\overrightarrow{c_{\bA}}^i \right\|_{H^\ell_{\delta''}}(t) \ls \left\|e^{2\gamma^{(n)}} (\wht {L^{(n+1)}_{\bA}})^i+\overrightarrow{c_{\bA}}^i \right\|_{H^\ell_{\delta''}}(0) \\
& +  C(A_1, C_i) T\sup_{t\in [0,T]}\left(1 + \|e^{2\gamma^{(n)}} (\wht {L^{(n+1)}_\bA})^i + \overrightarrow{c_{\bA}}^i\|_{H^\ell_{\de''}}+\|N^{(n)} e^{\gamma^{(n)}} (\wht {L^{(n+1)}_\bA})^0 - |\overrightarrow{c_{\bA}}|\|_{H^\ell_{\de''}}\right)(t).
\end{split}
\end{equation}
Combining \eqref{Lest.final.1} and \eqref{Lest.final.2} and choosing $T$ sufficiently small give that for $\ell=2,3$,
\begin{equation}\label{Lest.final.3}
\begin{split}
&\sup_{t \in [0,T]} \left(\left\|e^{\gamma^{(n)}} N^{(n)} (\wht {L^{(n+1)}_\bA})^0 - |\overrightarrow {c_{\bA}}| \right\|_{H^\ell_{\delta''}}(t) + \left\|e^{2\gamma^{(n)}} (\wht {L^{(n+1)}_{\bA}})^i+\overrightarrow{c_{\bA}}^i \right\|_{H^\ell_{\delta''}} (t)\right) \\
\ls &\left\|e^{\gamma^{(n)}} N^{(n)} (\wht {L^{(n+1)}_\bA})^0 - |\overrightarrow {c_{\bA}}| \right\|_{H^\ell_{\delta''}}(0) + \left\|e^{2\gamma^{(n)}} (\wht {L^{(n+1)}_{\bA}})^i+\overrightarrow{c_{\bA}}^i \right\|_{H^\ell_{\delta''}} (0).
\end{split}
\end{equation}
To obtain \eqref{prpL.0} from \eqref{Lest.final.3}, we use \eqref{L.init.bd} to control the data term and note that 
\begin{itemize}
\item by \eqref{L.transform}, $(\wht {L^{(n+1)}_{\bA}})^0 = (L^{(n+1)}_{\bA})^t$; 
\item and that by \eqref{L.transform}, \eqref{est.constraints.0} (for $\alp$), \eqref{estbeta.small} (for $\bt^{(n)}$), \eqref{estgamma} (for $\wht \gamma^{(n)}$) and \eqref{Lest.final.3},
\begin{equation*}
\begin{split}
 \| e^{2\gamma^{(n)}} (\wht {L^{(n+1)}_{\bA}})^i - e^{2\gamma^{(n)}} (L^{(n+1)}_\bA)^i \|_{H^2_{\de''}}
\ls & \|e^{2\gamma^{(n)}} (\bt^{(n)})^i (\wht {L^{(n+1)}_\bA})^0 \|_{H^2_{\de''}} \ls C_{eik}.
\end{split}
\end{equation*}
\end{itemize}
Finally, to obtain \eqref{prpL.1} from \eqref{Lest.final.3}, we argue similarly except that
\begin{itemize}
\item we use Corollary~\ref{cor.data} instead of \eqref{L.init.bd} to estimate the initial data term; 
\item and that we need to use \eqref{estbeta.large} instead of \eqref{estbeta.small} to control $\bt^{(n)}$ in $H^3_{\de'}$. 
\end{itemize}
Note that these result in the estimate being linear in $A_0 C_i$.

\textbf{Proof of \eqref{prpL.2}.} To obtain \eqref{prpL.2}, we directly use the equation \eqref{L.coord.eqn.renormalized.1} and the corresponding equation for $e^{2\gamma^{(n)}} (\wht {L^{(n+1)}_{\bA}})^i$. For simplicity, let us just consider the bound for 
$\rd_t\left(e^{\gamma^{(n)}} N^{(n)} (\wht {L^{(n+1)}_\bA})^0 \right)$. For this, we express $\wht{L^{(n)}_\bA}$ in terms of $L^{(n)}_\bA$ and write \eqref{L.coord.eqn.renormalized.1} as follows:
\begin{equation}\label{L.coord.eqn.renormalized.2}
\rd_t \left( e^{\gamma^{(n)}} N^{(n)} (\wht {L^{(n+1)}_\bA})^0 \right) = - \f{(L^{(n)}_\bA)^i}{(L^{(n)}_\bA)^t} \rd_i\left( e^{\gamma^{(n)}} N^{(n)} (\wht {L^{(n+1)}_\bA})^0 - |\overrightarrow {c_{\bA}}|\right) + \f{\mbox{(RHS of \eqref{L.coord.eqn.renormalized})}}{(L^{(n)}_\bA)^t}.
\end{equation} 
The first term on the RHS of \eqref{L.coord.eqn.renormalized.2} can be estimated using \eqref{prpL.1} (which we just proved), \eqref{estu.lower}, \eqref{estN.small}, Proposition~\ref{prpsmallness}, Lemma~\ref{der} and Proposition~\ref{produit} as follows: (Recall here that the constant in $\ls$ can depend on $C_{eik}$)
\begin{equation*}
\begin{split}
&\left\|\f{(L^{(n)}_\bA)^i}{(L^{(n)}_\bA)^t} \rd_i\left( e^{\gamma^{(n)}} N^{(n)} (\wht {L^{(n+1)}_\bA})^0 - |\overrightarrow {c_{\bA}}|\right)\right\|_{H^2_{\de'''}}\\
\ls & \left(1+ \left\| e^{\gamma^{(n)}} N^{(n)} (L^{(n)}_\bA)^t - |\overrightarrow {c_{\bA}}| \right\|_{H^2_{\de''}}^2 + \left\| e^{2\gamma^{(n)}} (L^{(n)}_\bA)^i + \overrightarrow {c_{\bA}}^i \right\|_{H^2_{\de''}}^2 \right)\left\| e^{\gamma^{(n)}} N^{(n)} (\wht {L^{(n+1)}_\bA})^0 - |\overrightarrow {c_{\bA}}| \right\|_{H^3_{\de''}}\\
\ls & A_0 C_i.
\end{split}
\end{equation*}
Finally, for the second term on the RHS of \eqref{L.coord.eqn.renormalized.2}, we need to get an estimate better than \eqref{L.coord.eqn.renormalized.bound} (in terms of dependence on the constants), which is possible since we now only have up to $2$ derivatives. The key point is that the appropriately-weighted-$H^2$ norms for $\bt^{(n)}$, $\wht \gamma^{(n)}$, $\wht N^{(n)}$, $e^{2\gamma^{(n)}} (L^{(n+1)}_{\bA})^i+\overrightarrow{c_{\bA}}^i$ and $N^{(n)}e^{\gamma^{(n)}} (L^{(n+1)}_{\bA})^t - |\overrightarrow{c_{\bA}}|$ are bounded independently of $C_i$, $A_0$, $A_1$ or $A_2$. More precisely, by Lemma~\ref{lm.constraint}, \eqref{estN.small}, \eqref{estbeta.small}, \eqref{estu.lower}, \eqref{estu}, Proposition~\ref{prpsmallness}, \eqref{prpL.0} (which we just proved) and Proposition~\ref{produit}, we have 
\begin{equation*}
\begin{split}
& \left\| \f{\mbox{(RHS of \eqref{L.coord.eqn.renormalized})}}{(L^{(n)}_\bA)^t} \right\|_{H^2_{\de'''}}\\
\ls & \left(1 + |N^{(n)}_{asymp}| + \|\bt^{(n)}\|_{H^2_{\de'}} + \|\wht N^{(n)}\|_{H^2_{\de}} + \|\wht \gamma^{(n)}\|_{H^2_{\de}} \right)^2\\
& \times \left(1 + \|e^{2\gamma^{(n-1)}} (L_{\bA}^{(n)})^i + \overrightarrow{c_{\bA}}^i\|_{H^2_{\de''}}+\|N^{(n-1)} e^{\gamma^{(n-1)}} (L_{\bA}^{(n)})^t - |\overrightarrow{c_{\bA}}|\|_{H^2_{\de''}}\right)\\
& \times \left(1+ \left\|e^{2\gamma^{(n)}} (L^{(n+1)}_{\bA})^i+\overrightarrow{c_{\bA}}^i \right\|_{H^2_{\delta''}}+\left\|N^{(n)}e^{\gamma^{(n)}} (L^{(n+1)}_{\bA})^t - |\overrightarrow{c_{\bA}}| \right\|_{H^2_{\delta''}}\right) \ls C.
\end{split}
\end{equation*}
Note that here on the LHS we use $H^2_{\de'''}$ instead of $H^2_{\de''}$ to compensate for the factors growing as $|x|\to \infty$.

Combining the above estimates and plugging into \eqref{L.coord.eqn.renormalized.2} give \eqref{prpL.2} for $\rd_t \left( e^{\gamma^{(n)}} N^{(n)} (\wht {L^{(n+1)}_\bA})^0 \right)$. The other term can be dealt with similarly. \qedhere
\end{proof}

\begin{prp}[Lower bound for $N^{(n)} e^{\gamma^{(n)}} (L^{(n+1)}_\bA)^t$]
For $n\geq 2$, the following lower bound holds:
$$ \min_\bA \inf_{x\in \mathbb R^2} \left| N^{(n)} e^{\gamma^{(n)}} (L^{(n+1)}_\bA)^t\right|(x) \geq \f 12 C_{eik}^{-1}.$$
\end{prp}
\begin{proof}
By \eqref{L.init.bd.lower}, at $t=0$, we have
$$ \min_\bA \inf_{x\in \mathbb R^2} \left| N^{(n)} e^{\gamma^{(n)}} (L^{(n+1)}_\bA)^t\right|(0,x) \geq C_{eik}^{-1}.$$
The desired estimate therefore follows from the bound for $\rd_t\left(N^{(n)} e^{\gamma^{(n)}} (L^{(n+1)}_\bA)^t\right)$ in Proposition~\ref{prpL} together with Proposition~\ref{holder}, after choosing $T$ to be sufficiently small.
\end{proof}

In our next lemma, we show that $F_{\bA}^{(n+1)}$ is supported in an appropriate compact set.
\begin{lm}[Support of $F_{\bA}^{(n+1)}$]\label{lm:cpt.supp.2}
Choosing $C_s$ (from \eqref{lm:cpt.supp}) larger if necessary, for $\ep$, $T$ sufficiently small (depending on $R$) and $n\geq 2$, $F_{\bA}^{(n+1)}$ is supported in the set $\{(t,x)\in [0,T]\times \mathbb R^2: C_s(1+R^\ep) t - |x|\geq -R\} $. In particular, choosing $T$ smaller if necessary, the support $supp(F_{\bA}^{(n+1)}) \subset \{(t,x)\in [0,T]\times \mathbb R^2: |x| \leq 2R\}$.
\end{lm}
\begin{proof}
By the transport equation \eqref{Fn} for $F_{\bA}^{(n+1)}$, it suffices to show that any integral curve $L_{\bA}^{(n)}$ which at $t=0$ is in $\{(x \in \mathbb R^2: |x| \leq R \}$ remains in  the set $\{(t,x)\in [0,T]\times \mathbb R^2: C_s(1+R^\ep) t - |x|\geq -R\} $ for all time.

To see this, let us fix such an integral curve $\gamma$. By \eqref{estu}, (and \eqref{est.constraints.0}, \eqref{estN.small} and Proposition~\ref{prpsmallness},)
$$\f{\de_{ij} (L^{(n)}_\bA)^i (L^{(n)}_\bA)^j }{(L^{(n)}_\bA)^t (L^{(n)}_\bA)^t} \leq e^{-2\gamma^{(n-1)}} (N^{(n-1)})^2 + \f{C A_0}{(1+|x|^2)^{\f{\de''+1}{2}}}\ls (1+|x|^2)^{\f \ep{10}} + \f{A_0}{(1+|x|^2)^{\f{\de''+1}{2}}}.$$
We parametrize $\gamma$ by its $t$-value, and denote by $r(t)$ the $|x|$ value of $\gamma(t)$. The above inequality implies that
$$r(t) \leq R +  C\int_0^t \left((1+(r(\tau))^2)^{\f \ep{10}} + \f{A_0}{(1+(r(\tau))^2)^{\f{\de''+1}{2}}}\right) \, d\tau.$$
A simple continuity argument shows that for $C_s$, $\ep$, $T$ appropriately chosen,
$$r(t) \leq R+ C_s(1+R^\ep) t,$$
which is to be shown.
\end{proof}

\begin{prp}[Estimates for $F^{(n+1)}_\bA$ and $\chi_\bA^{(n+1)}$]\label{prpF}
	For $n\geq 2$, the following estimates hold:
	\begin{equation*}
	\begin{split}
	\left\|F^{(n+1)}_{\bA}\right\|_{H^3} \ls C_i,\quad
	\left\|\partial_t F^{(n+1)}_{\bA}\right\|_{H^{2}} \ls & C(A_0) C_i,\\
	\left\|\chi^{(n+1)}_{\bA}\right\|_{C^0_{\de'+1}}  \leq 2 C_{\chi}, \quad	\left\|\chi^{(n+1)}_{\bA}\right\|_{H^3_{\de}}  \ls & C_i.
	\end{split}
	\end{equation*}
\end{prp}

\begin{proof}
We now apply Lemma~\ref{trans.est.higher} to control $F^{(n+1)}_{\bA}$ and $\chi^{(n+1)}_{\bA}$ satisfying the equations \eqref{Fn} and \eqref{chin}. By the compact support of $F_{\bA}$ that we established in Lemma~\ref{lm:cpt.supp.2}, we can put in any weights in the bounds for terms in which $F_{\bA}$ appears. 

\textbf{Estimate for $F_\bA^{(n+1)}$.} By Lemma~\ref{trans.est.higher}, Proposition~\ref{produit}, \eqref{estchi} and \eqref{estF}, we have
$$ \sup_{t\in [0,T]}\| F_{\bA}^{(n+1)} \|_{H^3}(t) \lesssim \| F_{\bA}^{(n+1)} \|_{H^3}(0) + C(A_0, C_i) T \| F_{\bA}^{(n)} \|_{H^3}(t) \| \chi_{\bA}^{(n)} \|_{H^3_\de}(t) \ls C_i,$$
after choosing $T$ sufficiently small.

\textbf{Estimate for $\rd_t F_\bA^{(n+1)}$.} We use \eqref{Fn} to write $\rd_t F_{\bA}^{(n+1)}$ in terms of $\f{(L^{(n)}_\bA)^i \rd_iF^{(n+1)}_\bA}{(L^{(n)}_\bA)^t}$ and $\f{1}{(L^{(n)}_\bA)^t} \chi^{(n)}_\bA F_{\bA}^{(n+1)}$. In other words,
$$\|\rd_t F_{\bA}^{(n+1)}\|_{H^2} \ls \left\|\f{(L^{(n)}_\bA)^i \rd_iF^{(n+1)}_\bA}{(L^{(n)}_\bA)^t}\right\|_{H^2} + \left\| \f{1}{(L^{(n)}_\bA)^t} \chi^{(n)}_\bA F_{\bA}^{(n+1)}\right\|_{H^2}.$$
The first term is easily seen to obey
$$\left\|\f{(L^{(n)}_\bA)^i \rd_iF^{(n+1)}_\bA}{(L^{(n)}_\bA)^t}\right\|_{H^2}\ls C(A_0) C_i$$ 
using \eqref{estu.lower}, \eqref{estu.small} and the estimate for $\| F_{\bA}^{(n+1)} \|_{H^3}$ that we have just established above, and recalling our convention that $C$ can depend on $C_{eik}$.

For the second term, we use the fact that $supp(F_{\bA}^{(n+1)})\subset B(0,2R)$ and use\footnote{We refer the reader to Footnote~\ref{cpt.supp.product} on p.\pageref{cpt.supp.product} regrading the use of Proposition~\ref{product} when one of the factors is compactly supported.} Proposition~\ref{product} together with \eqref{estu.lower}, \eqref{estu.small}, \eqref{estchi} and the estimate for $\| F_{\bA}^{(n+1)} \|_{H^3}$ that we have just established above to obtain
\begin{equation*}
\begin{split}
& \left\| \f{1}{(L^{(n)}_\bA)^t} \chi^{(n)}_\bA F_{\bA}^{(n+1)}\right\|_{H^2}\\
\ls & \left\|\f{1}{(L^{(n)}_\bA)^t} \chi^{(n)}_\bA \right\|_{C^0(0,3R)} \left\|F_{\bA}^{(n+1)}\right\|_{H^2} + \left\|\f{1}{(L^{(n)}_\bA)^t} \chi^{(n)}_\bA \right\|_{H^2(0,3R)} \left\|F_{\bA}^{(n+1)}\right\|_{C^0}\ls C(A_0)C_i.
\end{split}
\end{equation*}
Here, we again recall that $C$ can depend on $C_{eik}$, and hence also $C_\chi$.

Combining the above estimates, we obtain
$$\|\rd_t F_{\bA}^{(n+1)}\|_{H^2} \ls C(A_0)C_i.$$

\textbf{Estimate for $\chi_\bA^{(n+1)}$ in $H^3_{\de}$.} For the $\chi_{\bA}^{(n)}$ estimate, notice that the only inhomogeneous term in \eqref{chin} that is not compactly supported in the $(\chi^{(n)}_\bA)^2$ term, which can be controlled using Proposition~\ref{produit} by 
$$\| (\chi_{\bA}^{(n)})^2 \|_{H^3_{\de}} \ls \| \chi_{\bA}^{(n)} \|_{H^3_{\de}}^2 \ls C(A_0)C_i^2.$$
The remaining terms, which can compactly supported, as easier to handle and can be treated in a similar manner as in the proof of Proposition~\ref{prpN}, namely, we have
\begin{equation*}
\left\| 2((L^{(n)}_{\bA})^\rho \partial_\rho \phi^{(n)})^2
+\sum_{\bf B} F_{\bf B}^2(g^{(n)}_{\mu\nu} (L^{(n)}_{\bA})^\mu (L^{(n)}_{\bf B})^\nu)^2\right\|_{H^3_\de} \ls \ep C(A_0)C_i.
\end{equation*}
Therefore, by Lemma~\ref{trans.est.higher}, we have
\begin{equation}\label{chi.top.pf}
 \sup_{t\in [0,T]}\| \chi_{\bA}^{(n+1)} \|_{H^3_\de}(t) \lesssim \| \chi_{\bA}^{(n+1)} \|_{H^3_{\de}}(0) + T(C(A_0)C_i^2 + \ep C(A_0)C_i) \ls C_i 
\end{equation}
after choosing $T$ to be sufficiently small.

\textbf{Estimate for $\chi_\bA^{(n+1)}$ in $C^0_{\de+1}$.} 
We first estimate $\rd_t \chi_{\bA}^{(n+1)}$ in a similar manner as we bound $\rd_t F_{\bA}^{(n+1)}$ above, namely, we use the equation \eqref{chin} to write $\rd_t\chi_{\bA}^{(n+1)}$ in terms of $\f{(L^{(n)}_\bA)^i \rd_i\chi^{(n+1)}_\bA}{(L^{(n)}_\bA)^t}$, $\f{(\chi^{(n)})^2}{(L^{(n)}_\bA)^t}$ and $\f{\mbox{RHS of \eqref{chin}}}{(L^{(n)}_\bA)^t}$. Using the bounds for the terms in \eqref{chin} we proved above and the estimate \eqref{chi.top.pf} above, we have
$$\| \rd_t\chi_{\bA}^{(n+1)}\|_{H^2_{\de'}} \ls C(C_i).$$
By Proposition~\ref{holder}, this implies that $\| \rd_t\chi_{\bA}^{(n+1)}\|_{C^0_{\de'+1}}\ls C(C_i)$, which, together with \eqref{chi.init.bd}, imply
$$\|\chi_\bA^{(n+1)}\|_{C^0_{\de'+1}}\leq 2C_\chi.$$\qedhere
\end{proof}

We conclude this subsection by noting that the combination of the propositions proved in this subsection show that we can recover all the estimates in \eqref{estN.small}--\eqref{estF} (in fact with better constants for most of the estimates) when replacing $(n)$ by $(n+1)$. As a consequence, the estimates in \eqref{estN.small}--\eqref{estF} hold for all $n$.

\subsection{Convergence of the sequence and solution to the reduced system}\label{sec:iteration.convergence}

Next, we show that the sequence we constructed in fact converges to a limit (in a larger functional space).

Define the following distances:
\begin{align}
\label{d1} d^{(n)}_1 := & \|\wht \gamma^{(n+1)} - \wht \gamma^{(n)}\|_{H^1_{\de'}} + \|\rd_t (\wht \gamma^{(n+1)} - \wht \gamma^{(n)})\|_{L^2_{\de'}}+ \|H^{(n+1)} - H^{(n)}\|_{H^1_{\de+1}} + \|\tau^{(n+1)} - \tau^{(n)}\|_{L^2_{\de''+1}} \nonumber\\
& + \sum_{\bA} \| e^{2\gamma^{(n)}} (L^{(n+1)}_\bA)^i - e^{2\gamma^{(n-1)}} (L^{(n)}_\bA)^i\|_{H^1_{\de''}}  + \sum_{\bA} \| N^{(n)}e^{\gamma^{(n)}}(L^{(n+1)}_\bA)^t - N^{(n-1)}e^{\gamma^{(n-1)}}(L^{(n)}_\bA)^t\|_{H^1_{\de''}} \nonumber\\
& + \|\rd (\phi^{(n+1)} - \phi^{(n)}) \|_{H^1} + \sum_{\bA} \|F_{\bA}^{(n+1)} - F^{(n)}_\bA\|_{H^1} + \sum_{\bA} \|\chi_{\bA}^{(n+1)} - \chi_{\bA}^{(n)}\|_{H^1_\de}, \\
d^{(n)}_2 := & |N^{(n+1)}_{asymp} - N^{(n)}_{asymp}| + \|\wht N^{(n+1)} - \wht N^{(n)}\|_{H^2_{\de}} + \|\beta^{(n+1)} - \beta^{(n)}\|_{H^2_{\de'}}, \\
d^{(n)}_3 := & \sum_{\bA} \| \rd_t(e^{2\gamma^{(n)}} (L^{(n+1)}_\bA)^i - e^{2\gamma^{(n-1)}} (L^{(n)}_\bA)^i)\|_{L^2_{\de'''}}  \nonumber\\
& + \sum_{\bA} \|\rd_t( N^{(n)}e^{\gamma^{(n)}}(L^{(n+1)}_\bA)^t - N^{(n-1)}e^{\gamma^{(n-1)}}(L^{(n)}_\bA)^t)\|_{L^2_{\de'''}} \nonumber\\
&+ \|\rd_t (\f{e_0^{(n)} \phi^{(n+1)}}{N^{(n)}} - \f{e_0^{(n-1)} \phi^{(n)}}{N^{(n-1)}})\|_{L^2} + \sum_{\bA} \|\rd_t (F_{\bA}^{(n+1)} - F^{(n)}_\bA)\|_{L^2}, \\
d^{(n)}_4 := & \|e_0^{(n+1)} H^{(n+1)} - e_0^{(n)} H^{(n)}\|_{H^1_{\de+1}},\\
d^{(n)}_5 := & \|\rd_t (\f{e_0^{(n)} \wht \gamma^{(n+1)}}{N^{(n)}} - \f{e_0^{(n-1)} \wht \gamma^{(n)}}{N^{(n-1)}})\|_{L^2_{\de'}} + \|\rd_t(\tau^{(n+1)} - \tau^{(n)})\|_{L^2_{\de''+1}},\\
\label{d6} d^{(n)}_6 := & |\rd_t (N^{(n+1)}_{asymp} - N^{(n)}_{asymp})| + \|\rd_t(\wht N^{(n+1)} - \wht N^{(n)})\|_{H^2_{\de}} + \|e_0^{(n)}\beta^{(n+1)} - e_0^{(n-1)}\beta^{(n)} \|_{H^2_{\de'}}.
\end{align}

Since we have already obtained uniform bound on the iterates, from now on we need not keep track of the constants $A_0$, $A_1$, $A_2$. They will henceforth be simply absorbed into constants depending $C_{eik}$, $k$, $\de$.

The following proposition gives estimates for the distances $d_i^{(n)}$. The estimates are easier than those required for uniform boundedness in the previous subsection, since
\begin{itemize}
\item we have already closed the nonlinear bootstrap argument,
\item and in the estimates for $d_i^{(n)}$, we only need bounds for lower order of derivatives.
\end{itemize}
The estimate we prove nevertheless crucially relies on the structure of the equations so that the distances can be controlled in a step-by-step manner such that at each step the RHS either consists terms bounded in the previous step or has appropriate smallness constant. It is exactly the same kind of structure that allowed us to prove the uniform boundedness statement in the previous subsection. We will only briefly indicate how these estimates are proven, but will refer the reader to the corresponding propositions in the previous subsection, where the analogous estimates for the corresponding quantities (without taking difference) were proven.
\begin{proposition}\label{dn.est}
For $n\geq 3$, the following inequalities hold for some fixed $C_*>1$ depending on $C_{eik}$, $\de$, $R$:
\begin{align}
d^{(n)}_1 \leq & C_*(C_i) T (d^{(n-1)}_1 + d^{(n-2)}_1 + d^{(n-1)}_2 + d^{(n-2)}_2 + d^{(n-1)}_3 + d^{(n-1)}_4 + d^{(n-1)}_5 + d^{(n-1)}_6),\\
d^{(n)}_2 \leq & C_* (d^{(n-1)}_1 + d^{(n-2)}_1) + C_* \ep (d^{(n-1)}_2 + d^{(n-2)}_2),\\
d^{(n)}_3 \leq & C_* (d^{(n-1)}_1 + d^{(n-2)}_1 + d^{(n-1)}_2 + d^{(n-2)}_2),\\
d^{(n)}_4 \leq & C_* C_i (d^{(n-1)}_1 + d^{(n-1)}_2),\\
d^{(n)}_5 \leq & C_* d^{(n-1)}_6 + C_* C_i ({ d^{(n-1)}_1} + d^{(n-1)}_2 + d^{(n-2)}_2 + d^{(n-1)}_3 + d^{(n-1)}_4),\\
d^{(n)}_6 \leq & C_* C_i (d^{(n-1)}_1 + d^{(n-1)}_2 + d^{(n-2)}_2 + d^{(n-1)}_3 + d^{(n-1)}_4) + C_* \ep (d^{(n-1)}_5 + d^{(n-1)}_6).
\end{align}
\end{proposition}
\begin{proof}
The basic strategy is to estimate these differences in a way similar to Section~\ref{sec:iteration.bdd}. {In particular, we use the structure of the equations in a similar manner.}

To control $d^{(n)}_1$ is the easiest. All of these terms are controlled using transport or wave type equations with $0$ initial data. Therefore, on the RHS we need to use all of $d^{(n)}_1, \dots d^{(n)}_6$ and $d^{(n-2)}_1$, $d^{(n-2)}_2$, the estimate comes with a small constant $T$ associated with a time integral (cf. estimates for the analogous quantities without taking differences in Propositions~\ref{prpgamma} {(for $\wht \gamma^{(n+1)} - \wht \gamma^{(n)}$ and $\rd_t (\wht \gamma^{(n+1)} - \wht \gamma^{(n)})$), \ref{prpH} ($H^{(n+1)} - H^{(n)}$), \ref{prpphi} (for $\rd (\phi^{(n+1)} - \phi^{(n)})$), \ref{prpL} (for $e^{2\gamma^{(n)}} (L^{(n+1)}_\bA)^i - e^{2\gamma^{(n-1)}} (L^{(n)}_\bA)^i$ and $N^{(n)}e^{\gamma^{(n)}}(L^{(n+1)}_\bA)^t - N^{(n-1)}e^{\gamma^{(n-1)}}(L^{(n)}_\bA)^t$), \ref{prpF} (for $F_{\bA}^{(n+1)} - F^{(n)}_\bA$ and $\chi_{\bA}^{(n+1)} - \chi_{\bA}^{(n)}$))}. {For} $\tau^{(n+1)}-\tau^{(n)}${, we estimate it directly by} integrating $\partial_t(\tau^{(n+1)}-\tau^{(n)})$, 
which can be controlled in terms of $d^{(n-1)}_1, d^{(n-1)}_2, d^{(n-2)}_2, d^{(n-1)}_5$ and $d^{(n-1)}_6$.

To control $d^{(n)}_2$, we consider the difference between the $(n+1)$-st iterates of \eqref{lapsen} and \eqref{shiftn} and their $n$-th iterates and perform elliptic estimates for the differences of $N^{(n+1)}- N^{(n)}$ and $\bt^{(n+1)}-\bt^{(n)}$. Arguing as in Propositions~\ref{prpN} and \ref{prpbeta}, for most terms there is a smallness constant $\ep$ in the coefficient. The only exception arise when controlling the difference $\bt^{(n+1)}-\bt^{(n)}$, on the right hand side there is a term $H^{(n+1)}-H^{(n)}$ with coefficients depending on $N^{(n)}e^{\gamma^{(n)}}$, $N^{(n-1)}e^{\gamma^{(n-1)}}$, $N^{(n-2)}e^{\gamma^{(n-2)}}$, which while not small, can be controlled by a constant independent of $C_i$.
To control $d^{(n)}_2$, we consider the difference between the $(n+1)$-st iterates of \eqref{lapsen} and \eqref{shiftn} and their $n$-th iterates and perform elliptic estimates for the differences of $N^{(n+1)}- N^{(n)}$ and $\bt^{(n+1)}-\bt^{(n)}$. Arguing as in Propositions~\ref{prpN} and \ref{prpbeta}, for most terms there is a smallness constant $\ep$ in the coefficient. The only exception arise when controlling the difference $\bt^{(n+1)}-\bt^{(n)}$, on the right hand side there is a term $H^{(n+1)}-H^{(n)}$ with coefficients depending on $N^{(n)}e^{\gamma^{(n)}}$, $N^{(n-1)}e^{\gamma^{(n-1)}}$, $N^{(n-2)}e^{\gamma^{(n-2)}}$, which while not small, can be controlled by a constant independent of $C_i$.

To control $d^{(n)}_3$, we control the differences for appropriates $n$'s of RHSs of \eqref{phin}, \eqref{un} and \eqref{Fn}. It is easy to check that to control this we only need to control the difference of terms appearing in $d^{(n)}_1$, $d^{(n-1)}_1$, $d^{(n-1)}_2$ and $d^{(n-2)}_2$, i.e., we do need to estimate the difference of the top $\rd_t$ derivative of any quantity. Moreover, since we do not need to take any derivatives of the RHS of these equations, we check using the estimates \eqref{estN.small}-\eqref{estF} and Proposition~\ref{prpsmallness} that the constant we have in the estimate can be chosen independent of $C_i$ (cf. Proposition~\ref{prpphi}, \ref{prpL}, \ref{prpF}).

To control $d^{(n)}_4$, we bound the difference $e_0^{(n+1)} H^{(n+1)} - e_0^{(n)} H^{(n)}$ and its spatial derivative by controlling the appropriate difference of of the RHS \eqref{hijn}. This is similar to the estimates for $d^{(n)}_3$, except that since we now need to take one spatial derivative, the constant may depend (linearly) on $C_i$ (cf. Proposition~\ref{prpH}).

To control $d^{(n)}_5$, we first estimate $\rd_t (\f{e_0^{(n)} \wht \gamma^{(n+1)}}{N^{(n)}} - \f{e_0^{(n-1)} \wht \gamma^{(n)}}{N^{(n-1)}})$by taking appropriate difference of \eqref{gamman}. The estimate for $\rd_t (\tau^{(n+1)} - \tau^{(n)})$ then follows easily using \eqref{taun}. There are two main observations. First, we note that the RHS does not depend on $d^{(n-1)}_5$, this follows easily from inspecting the RHS of \eqref{gamman}. Second, we note that when $d^{(n-1)}_6$ appears on the RHS, the constant is \emph{independent} of $C_i$. The relevant term here is $\frac{e^{2\gamma^{(n)}}}{2 N^{(n)}}e_0^{(n-1)} \left( \frac{e^{-2\gamma^{(n)}}}{N^{(n-1)}}div(\beta^{(n)})\right)$. The key point is that $e_0^{(n-1)}\bt^{(n)} - e_0^{(n-2)}\bt^{(n-1)} $ must be multiplied by at most two derivatives of $\gamma^{(n)}$, $\gamma^{(n-1)}$, $N^{(n)}$, $N^{(n-1)}$ and $N^{(n-2)}$. Hence, by \eqref{estN.small}, Proposition~\ref{prpsmallness}, these terms indeed can be bounded independent of $C_i$. The other terms can be controlled more roughly using the estimate in the previous section since we will allow the coefficient to depend on $C_i$.

To control $d^{(n)}_6$, we take $\rd_t$ derivatives of \eqref{lapsen} and \eqref{shiftn}, take the appropriate difference, and use elliptic estimates. The key is to observe that when $d^{(n-1)}_5$ and $d^{(n-1)}_6$ appear on the RHS, then there is a smallness constant $C_*\ep$. To see this, one can argue in a similar manner as in Propositions~\ref{prpN} and \ref{prpbeta}. The remaining terms can be controlled more roughly using the estimate in the previous section since we will allow the coefficient to depend on $C_i$.
\end{proof}

Proposition~\ref{dn.est}, together with a simple induction argument, imply the following estimates (we omit the easy proof):
\begin{cor}\label{cor.d}
Assume that for $n = 1,2$, we have the bound
$$d^{(n)}_1 \leq B,\quad d^{(n)}_2 \leq 8 C_* B,\quad d^{(n)}_3 \leq 2({3} C_* + 16 (C_*)^2) B,\quad d^{(n)}_4 \leq 2 C_i (C_* + 8(C_*)^2) B,$$
$$d^{(n)}_5 \leq 4 (1+C_*) C_i \left(2C_* + 16(C_*)^2 + 2C_* (2 C_* + 16 (C_*)^2) + 2 C_* C_i (C_* + 8(C_*)^2)\right) B,$$
$$d^{(n)}_6 \leq 4C_i \left(2C_* + 16(C_*)^2 + 2C_* (2 C_* + 16 (C_*)^2) + 2 C_* C_i (C_* + 8(C_*)^2)\right) B,$$
for some $B>0$, where $C_*$ is as in Proposition~\ref{dn.est} and $C_i$ is as in Corollary~\ref{cor.data}. (Note that this can always be achieved by taking $B$ larger if necessary).

Then, if $T$ and $\ep$ are sufficiently small (where $T$ may depend on $C_i$, $C_{eik}$, $\de$, $R$, but $\ep$ may only depend on $C_{eik}$, $\de$, $R$, but \underline{not} $C_i$), for every $n\geq 3$, the following bounds hold:
$$d^{(n)}_1 \leq 2^{-(n-3)} 2^{-1} B,\quad d^{(n)}_2 \leq 2^{-(n-3)} \cdot 4 C_* B,\quad d^{(n)}_3 \leq 2^{-(n-3)} \cdot ({3} C_* + 16 (C_*)^2) B,$$
$$d^{(n)}_4 \leq 2^{-(n-3)}C_i (C_* + 8(C_*)^2) B,$$ 
$$d^{(n)}_5 \leq 2^{-(n-3)} 2 (1+C_*) C_i \left(2C_* + 16(C_*)^2 + 2C_* ({3} C_* + 16 (C_*)^2) + 2 C_* C_i (C_* + 8(C_*)^2)\right) B,$$
$$d^{(n)}_6 \leq 2^{-(n-3)} 2C_i \left(2C_* + 16(C_*)^2 + 2C_* ({3} C_* + 16 (C_*)^2) + 2 C_* C_i (C_* + 8(C_*)^2)\right) B.$$
\end{cor}

The precise expression above is of course unimportant, but it shows that in the function spaces as in the definition of $d^{(n)}_1,\dots, d^{(n)}_6$, the sequence we constructed is Cauchy and therefore convergent. Using the regularity that we have obtained, it is easy to verify the limit indeed satisfies the system \eqref{tau}-\eqref{geodesic.reduced}, \eqref{F}-\eqref{chi}. Finally, define $u_{\bA}$ by \eqref{u}. It is easy to verify that we have indeed constructed a solution to \eqref{tau}--\eqref{chi}. Moreover, one easily checks that the solution is unique: indeed, if there are two solutions, we can control their difference using the distances \eqref{d1}-\eqref{d6}, then an argument as in Proposition~\ref{dn.est} and Corollary~\ref{cor.d} shows that these two solutions coincide. We summarize this discussion in the following theorem:
\begin{theorem}\label{thm:reduced.sys}
Given the initial conditions in Section~\ref{sec.data.constraints}, there exists a unique solution 
$$(N,\beta,\tau,H,\gamma, \phi, L_{\bA}, F_{\bA}, \chi_{\bA})$$
to the reduced system \eqref{tau}--\eqref{chi} such that
\begin{itemize}
\item $\gamma$ and $N$ admit the decompositions
$$\gamma=-\alpha \chi(|x|){\log}(|x|)+\wht \gamma,\quad N = {1+} N_{asymp} \chi(|x|){\log}(|x|) + \wht N,$$
where $\alp\geq 0$ is a constant, $N_{asymp}(t)\geq 0$ is a function of $t$ alone and 
$$\wht \gamma \in H^4_{\de'},\quad \f{e_0 \wht\gamma}{N}\in H^3_{\de'+1},\quad \rd_t\f{e_0 \wht\gamma}{N}\in H^2_{\de'+1},\quad \wht N\in H^5_\de,\quad \rd_t \wht N\in H^3_\de,$$
with estimates depending only on $C_i$, $C_{eik}$, $\de$ and $R$.
\item For all $\bA$, $\phi$, $\rd_t\phi$, $F_{\bA}$ are supported in\footnote{Here, $J^+$ denotes the causal future with respect to the metric $g$. Strictly speaking, we have only proved that $\phi$, $\rd_t\phi$, $F_{\bA}$ are supported in
$\{(t,x)\in [0,T]\times \mathbb R^2: C_s(1+R^\ep) t - |x|\geq -R\}$, but a posteriori, it is easy to check that the supports indeed lie in $J^+(\{t=0\}\cap B(0,R))$.}
$$J^+(\{t=0\}\cap B(0,R))$$
and satisfy
$$\nabla \phi, \, \partial_t \phi \in H^3,\quad \rd_t\f{e_0\phi}{N}\in H^2,\quad F_{\bA}\in H^3,\quad \rd_tF_{\bA}\in H^2,$$
with estimates depending only on $C_i$, $C_{eik}$, $\de$ and $R$.
\item $(\beta,\tau,H,L_{\bA}, \chi_{\bA})$ are in the following spaces (for all $\bA$):
	$$\beta,\, e_0\bt \in H^4_{\de'},\quad \tau \in H^3_{\de''+1},\quad \rd_t\tau \in H^2_{\de''+1},\quad  H,\, e_0 H \in H^{3}_{\delta +1},\quad \chi_\bA \in H^3_\de,$$
	$$e^{2\gamma} L_{\bA}^i+\overrightarrow{c_{\bA}}^i ,\,\, N e^{\gamma} L_{\bA}^t - |\overrightarrow{c_{\bA}}|\in H^3_{\delta''},\quad \partial_t (e^{2\gamma} L_{\bA}^i) ,\,\, \partial_t (N e^{\gamma} L_{\bA}^t)\in H^2_{\delta'''}, $$
	with estimates depending only on $C_i$, $C_{eik}$, $\de$ and $R$. 
\item $N e^\gamma L_\bA^t$ is bounded below by
				$$ \min_\bA \inf_x |N e^\gamma L_\bA^t|(x) \geq \f 12 C_{eik}^{-1}.$$
\item The smallness conditions in \eqref{estN.small}, \eqref{estbeta.small} and Proposition~\ref{prpsmallness} hold (without the $(n)$).
\end{itemize}
\end{theorem}

Before we end this subsection, it will be convenient to note that according to Remark~\ref{rmk:u}, for $u_{\bA}$ defined as above, we have
\begin{proposition}\label{u.eikonal}
Given a solution to \eqref{tau}--\eqref{chi}, $\forall \bA$, $u_{\bA}$ satisfies
$$L_{\bA}^\alp = -\gi^{\alp\bt}\rd_\bt u_{\bA}, \quad \gi^{\alp\bt}\rd_\alp u_{\bA} \rd_\bt u_{\bA} = 0.$$
\end{proposition}

\section{Local well-posedness for the system \eqref{back}}\label{sec.final}

In the previous section, we have shown that a unique local solution to the reduced system \eqref{tau}--\eqref{chi} exists (cf. Theorem~\ref{thm:reduced.sys}). We now show that given initial data satisfying the constraint equations \eqref{mom2} and \eqref{ham2}, as well as the initial conditions in Section~\ref{sec.data.constraints}, this solution is indeed a solution to the system \eqref{back}. For this purpose, it will be notationally convenient to denote\footnote{Note that this is different from the expression for $^{(4)}T_{\mu\nu}$.} as in Section~\ref{sec.T} that $T_{\mu\nu} = 2\rd_\mu\phi \rd_\nu\phi-g_{\mu\nu}\gi^{\alp\bt}\rd_\alp\phi\rd_\bt\phi+\sum_{\bA} F_{\bA}^2 \rd_\mu u_{\bA} \rd_\nu u_{\bA}$.

Given a solution to \eqref{tau}--\eqref{chi}, we have the geometric quantities $\gamma$, $\bt$, $N$, $H$ and $\tau$. Let $K_{ij}$ be defined according to \eqref{K} with $\tau$, $H$ and $\gamma$ as given. At this point, we do not yet know that (1) $K_{ij}$ is the second fundamental form and that (2) $H_{ij}$ is the traceless part of $K_{ij}$. On the other hand, by \eqref{tau}, we know that $\tau$ is the mean curvature of the constant-$t$ hypersurfaces.

We compute using \eqref{tau} and \eqref{shift} to get that
\begin{equation}
\begin{split}
K_{ij} :=& \f 12 e^{2\gamma}\tau\delta_{ij}+H_{ij} = (\f 12 e^{2\gamma}\tau-\f{1}{2N}e^{2\gamma}(\rd^k\beta_k))\delta_{ij}+\f{1}{2N} e^{2\gamma}(\partial_i \beta_j+\partial_i \beta_j)\\
=& -\f{1}{2N}e_0(e^{2\gamma})\delta_{ij}+\f{1}{2N} e^{2\gamma}(\partial_i \beta_j+\partial_i \beta_j).
\end{split}
\end{equation}
Hence, by \eqref{2nd.fund.form}, $K_{ij}$ is indeed the second fundamental form of the constant-$t$ hypersurfaces. Moreover, since we already know that $\tau$ is the mean curvature, it follows from \eqref{K} that $H_{ij}$ is indeed the traceless part (with respect to $\bar{g}$) of $K_{ij}$.

Therefore, it remains to show that $\tau=0$, $\Box_g u_{\bA}=\chi_{\bA}$, and that the first equation in \eqref{back} is satisfied. Let us note that the first equation can be rephrased as $G_{\mu\nu} = T_{\mu\nu}$, where $G_{\mu\nu}:= R_{\mu\nu}-\f 12 g_{\mu\nu} R$ is the Einstein tensor.

We first need some preliminary calculations for $G_{00}$ and $G_{ij}$.
\begin{prp}\label{prop:some.G.comp}
Given a solution to the reduced system \eqref{tau}--\eqref{chi}, the Einstein tensor (in the basis $\{e_0,\rd_i\}$) is given by
\begin{equation}\label{reduced.G}
G_{00}=\f{N}2(e_0\tau)+T_{00},\quad G_{ij}=T_{ij}+\f 12 \f{e^{2\gamma}(e_0\tau)}{N}\delta_{ij}.
\end{equation}
Moreover, $T$ satisfies the following propagation equation:
\begin{equation}\label{T.prop.prop}
\begin{split}
D^\mu T_{\mu\nu}=&\sum_{\bA}F_{\bA}^2 (\Box_g u_{\bA}-\chi_{\bA}) (\rd_\nu u_{\bA}).
\end{split}
\end{equation}
\end{prp}
\begin{proof}
First note that by Proposition~\ref{u.eikonal}, the computations in \eqref{sec.T} are applicable.

\textbf{Proof of first identity in \eqref{reduced.G}.} By \eqref{gamma}, \eqref{spatial.R.tr} and \eqref{T-trT.spatial.tr}, we have
\begin{equation}\label{spatial.R.tr.2}
\delta^{ij}R_{ij}=\delta^{ij}(T_{ij}-g_{ij}\mbox{tr}_g T).
\end{equation}
By \eqref{lapse}, \eqref{R00} and \eqref{T-trT.00}, we obtain
\begin{equation}\label{R00.2}
R_{00}=N(e_0\tau)+T_{00}- g_{00} \mbox{tr}_g T.
\end{equation}
\eqref{spatial.R.tr.2} and \eqref{R00.2} (and \eqref{g.form}) together give
\begin{equation}\label{R.rdtau}
R=-N^{-2}(Ne_0\tau+T_{00}- g_{00} \mbox{tr}_g T)+e^{-2\gamma}\delta^{ij}(T_{ij}-g_{ij}\mbox{tr}_g T)=-\f{1}{N}e_0\tau -2\mbox{tr}_g T,
\end{equation}
which then implies (using \eqref{R00.2} again)
\begin{equation}\label{G00.2}
G_{00}=\f{N}2e_0\tau+T_{00}.
\end{equation}

\textbf{Proof of second identity in \eqref{reduced.G}.} By \eqref{hij}, \eqref{Rij}, \eqref{spatial.R.tr}, \eqref{T-trT.ij}, \eqref{T-trT.spatial.tr}, \eqref{spatial.R.tr.2} and \eqref{R.rdtau}, we obtain
\begin{equation}\label{Gij.2}
R_{ij}=T_{ij}-g_{ij}\mbox{tr}_g T,\quad G_{ij}=T_{ij}+\f 12 \f{e^{2\gamma}e_0\tau}{N}\delta_{ij}.
\end{equation}

\textbf{Proof of \eqref{T.prop.prop}.} Finally, to derive \eqref{T.prop.prop}, we use \eqref{phi}, \eqref{u}, \eqref{F}, \eqref{chi} and Proposition~\ref{u.eikonal} to obtain
\begin{equation}\label{T.prop.2}
\begin{split}
D^\mu T_{\mu\nu} =& 2\Box_g\phi (\rd_\nu\phi)+\sum_{\bA} 2 F_{\bA} ( \rd_\alp F_{\bA}) \gi^{\alp\bt} (\rd_\bt u_{\bA}) (\rd_\nu u_{\bA}) + \sum_{\bA}F_{\bA}^2 (\Box_g u_{\bA}) (\rd_\nu u_{\bA}) \\
&+ \sum_{\bA}F_{\bA}^2 \gi^{\alp\bt} (\rd_\alp u_{\bA}) D_{\bt} (\rd_\nu u_{\bA})
= \sum_{\bA}F_{\bA}^2 (\Box_g u_{\bA}-\chi_{\bA}) (\rd_\nu u_{\bA}).
\end{split}
\end{equation}
\end{proof}

By Proposition~\ref{prop:some.G.comp}, in order to show that a solution to \eqref{tau}--\eqref{chi} is indeed a solution to \eqref{back}, it remains to show that
\begin{equation}\label{quantities.to.be.0}
\tau=0,\quad  \Box_g u_{\bA}=\chi_{\bA},\quad G_{0i}=T_{0i}.
\end{equation}
These will be shown simultaneously. We first derive the main equations used to prove \eqref{quantities.to.be.0}. 
\begin{prp}\label{prop.eqn.for.quantities.to.be.0}
The following coupled system hold for $\tau$, $(G_{0i}- T_{0i}-\f 12 N(\rd_i\tau))$ and $(\Box_g u_{\bA}-\chi_{\bA})$ (where we again use the basis $\{e_0,\rd_i\}$):
\begin{equation}\label{main.1.G}
\begin{split}
&-\f 1{N^2}(\rd_t-\bt^k\rd_k)(G_{0i}- T_{0i}) + \f 12 \f{\rd_i e_0\tau}{N} + \f 12 \f{(\rd_i N)(e_0\tau)}{N^2}\\
&+\f 1{N^2}\left(\f{(e_0 N)\delta^{ij}}{N}+\f 12 \left({-4}(e_0 \gamma)\delta_i^j + 2\de_i^k \rd_\ell \bt^\ell + 2(\rd_i\bt^j)\right)\right)(G_{0j}- T_{0j})\\
=&-\sum_{\bA}F_{\bA}^2(\Box_g u_{\bf A}-\chi_{\bf A})(\rd_i u_{\bf A}),
\end{split}
\end{equation}

\begin{equation}\label{main.2.G}
\begin{split}
& \f 12 (\rd_t-\bt^i\rd_i)\left(\frac{e_0\tau}{N}\right)-e^{-2\gamma}\delta^{ij}\rd_i(G_{j0}-T_{j0})\\
=&\f{{e^{-2\gamma}}\delta^{ij} (\rd_i N)}{N}(G_{0j}-T_{0j})-\left(2 e_0 \gamma - \rd_i\bt^i\right)\f {e_0\tau}{N}+\sum_{\bA}F_{\bA}^2(\Box_g u_{\bf A}-\chi_{\bf A})(e_0 u_{\bf A}),
\end{split}
\end{equation}
\begin{equation}\label{Boxu-chi}
\begin{split}
L_{\bA}^\rho\rd_\rho(\Box_g u_{\bA}-\chi_{\bA})=-(\Box_g u_{\bA}-\chi_{\bA})(\Box_g u_{\bA}+\chi_{\bA})-N (e_0\tau) (\wht {L_{\bA}})^0(\wht {L_{\bA}})^0 - 2(G_{0i}-T_{0i})(\wht {L_{\bA}})^0(\wht {L_{\bA}})^i,
\end{split}
\end{equation}
where $(\wht {L_\bA})^{\mu}$ denotes the components of $L_\bA$ with respect to the basis $\{e_0,\rd_i\}$.
\end{prp}
\begin{proof}

\textbf{Proof of \eqref{main.1.G}.} By \eqref{connections},
\begin{equation}\label{Bianchi.1.1}
\begin{split}
D_0 (G_{0i}- T_{0i})=&(\rd_t-\bt^k\rd_k)(G_{0i}- T_{0i})-\f{(e_0 N)}{N}(G_{0i}- T_{0i})-e^{-2\gamma}\delta^{jk}N\rd_kN (G_{ij}-T_{ij})\\
&-\f{\rd_i N}{N} (G_{00}-T_{00})-\f 12 (2(e_0 \gamma)\delta_i^j+(\rd_i\bt^j)-\delta_{ik}\delta^{j\ell}(\rd_\ell\bt^k))(G_{0j}-T_{0j})\\
=&(\rd_t-\bt^k\rd_k)(G_{0i}- T_{0i})-(\rd_i N)(e_0 \tau)\\
&-\left(\f{(e_0 N)\delta^{ij}}{N}+\f 12 \left(2(e_0 \gamma)\delta_i^j+(\rd_i\bt^j)-\delta_{ik}\delta^{j\ell}(\rd_\ell\bt^k)\right)\right)(G_{0j}- T_{0j}),
\end{split}
\end{equation}
where in the last line we have used \eqref{reduced.G}. On the other hand, by \eqref{connections}, we have
\begin{equation}\label{Bianchi.1.2}
\begin{split}
&(g^{-1})^{jk} D_j (G_{ki}-T_{ki})\\
=& e^{-2\gamma}\delta^{jk}\partial_j\left(\f 12 \frac{e^{2\gamma}(e_0 \tau)}{N}\delta_{ki} \right)-e^{-2\gamma}\delta^{jk}(\delta_j^{\ell}\partial_k \gamma + \delta_k^{\ell} \partial_j \gamma-\delta_{jk}\de^{\ell m} \partial_{m} \gamma )(G_{\ell i}-T_{\ell i})\\
&-\frac{\delta^{jk}}{2N^2}(2(e_0 \gamma)\delta_{jk}-\delta_{k\ell}\partial_j \beta^{\ell} -\delta_{j\ell}
\partial_k \beta^{\ell})(G_{0i}-T_{0i})\\
&-e^{-2\gamma}\delta^{jk}(\delta_i^{\ell} \partial_j \gamma + \delta_j^{\ell} \partial_i \gamma-\delta_{ij} \de^{\ell m} \partial_{m} \gamma)(G_{k\ell}-T_{k\ell})
-\frac{1}{2N^2}\delta^{jk}(2(e_0 \gamma)\delta_{ij}-\delta_{i\ell}\partial_j \beta^{\ell} -\delta_{j\ell}\partial_i \beta^{\ell})(G_{0k}-T_{0k})\\
=&\f 12 \f{\rd_i e_0\tau}{N}-\f 12 \f{(\rd_i N)(e_0\tau)}{N^2} + \f{(\rd_i\gamma)(e_0\tau)}{N} + (- \f{{3}(e_0\gamma)}{N^2}+\f{1}{N^2} \rd_\ell \bt^\ell)(G_{0i}-T_{0i})\\
&- \f{1}{2N^2} \left(-(\rd_j\bt^\ell)\delta^{jk}\delta_{i\ell}-(\rd_i\bt^{k})\right)(G_{0k}-T_{0k})- e^{-2\gamma}\rd_i\gamma \de^{k\ell}(G_{k\ell}-T_{k\ell})\\
=&\f 12 \f{\rd_i e_0\tau}{N}-\f 12 \f{(\rd_i N)(e_0\tau)}{N^2} + \f{1}{2N^2} \left(-{6}\de_i^k(e_0\gamma)+2\de_i^k \rd_\ell \bt^\ell + (\rd_j\bt^\ell)\delta^{jk}\delta_{i\ell} + (\rd_i\bt^{k})\right)(G_{0k}-T_{0k}) ,
\end{split}
\end{equation}
where in the last line we have used \eqref{reduced.G}. By \eqref{T.prop.prop} and the Bianchi equation $D^{\mu} G_{\mu i}=0$, we have 
$$\mbox{$(-\f 1{N^2}\times$\eqref{Bianchi.1.1}$)+$\eqref{Bianchi.1.2}$=-\sum_{\bA}F_{\bA}^2(\Box_g u_{\bf A}-\chi_{\bf A})(\rd_i u_{\bf A}) $}.$$
The LHS can be expanded, using the expressions in \eqref{Bianchi.1.1} and \eqref{Bianchi.1.2}, as
\begin{equation*}
\begin{split}
&-\f 1{N^2}(\rd_t-\bt^k\rd_k)(G_{0i}- T_{0i})+\f 1{N^2}(\rd_i N)(e_0 \tau)\\
&+\f 1{N^2}\left(\f{(e_0 N)\delta^{ij}}{N}+\f 12 \left(2(e_0 \gamma)\delta_i^j+(\rd_i\bt^j)-\delta_{ik}\delta^{j\ell}(\rd_\ell\bt^k)\right)\right)(G_{0j}- T_{0j})\\
&+\f 12 \f{\rd_i e_0\tau}{N}-\f 12 \f{(\rd_i N)(e_0\tau)}{N^2} + \f{1}{2N^2} \left(-{6}\de_i^k (e_0\gamma)+2\de_i^k \rd_\ell \bt^\ell + (\rd_j\bt^\ell)\delta^{jk}\delta_{i\ell} + (\rd_i\bt^{k})\right)(G_{0k}-T_{0k}) \\
=&-\f 1{N^2}(\rd_t-\bt^k\rd_k)(G_{0i}- T_{0i}) + \f 12 \f{\rd_i e_0\tau}{N} + \f 12 \f{(\rd_i N)(e_0\tau)}{N^2} \\
&+\f 1{N^2}\left(\f{(e_0 N)\delta^{ij}}{N}+\f 12 \left(-{4}(e_0 \gamma)\delta_i^j + {2}\de_i^k \rd_\ell \bt^\ell + 2(\rd_i\bt^j)\right)\right)(G_{0j}- T_{0j}),
\end{split}
\end{equation*}
which proves \eqref{main.1.G}.

\textbf{Proof of \eqref{main.2.G}.}
By \eqref{connections} and \eqref{reduced.G}, we have
\begin{equation}\label{Bianchi.2.1}
\begin{split}
&\gi^{00} D_0 (G_{00}-T_{00})\\
=&-\f 1{N^2}(\rd_t-\bt^i\rd_i)(G_{00}-T_{00})+\f 1{N^2}\left(\f{2(e_0 N)}{N}(G_{00}-T_{00}){+}2e^{-2\gamma}\delta^{ij}N(\rd_i N) (G_{0j}-T_{0j})\right)\\
=&-\f 1{2N^2}(\rd_t-\bt^i\rd_i)(N(e_0\tau))+\f{(e_0 N)(e_0\tau)}{N^2}{+}\f{2e^{-2\gamma}\delta^{ij} (\rd_i N)}{N}(G_{0j}-T_{0j}).
\end{split}
\end{equation}
On the other hand, using \eqref{connections} and \eqref{reduced.G}, we also have
\begin{equation}\label{Bianchi.2.2}
\begin{split}
&\gi^{ij}D_i (G_{j0}-T_{j0})\\
=&e^{-2\gamma}\de^{ij}\rd_i(G_{j0}-T_{j0})-\f 1{2N^2}\de^{ij}\left(2(e_0 \gamma)\delta_{ij}-(\rd_i\bt^k)\delta_{jk}-(\rd_j\bt^k)\delta_{ik}\right) (G_{00}-T_{00})\\
&-e^{-2\gamma}\de^{ij}\left(\delta_i^k\rd_j\gamma+\delta^k_j\rd_i\gamma-\delta_{ij}\delta^{k\ell}\rd_\ell\gamma\right)(G_{k0}-T_{k0})\\
&-e^{-2\gamma}\de^{ij}\f{\rd_i N}{N}(G_{j0}-T_{j0})-\f 12 e^{-{2}\gamma}\de^{ij}\left(2(e_0 \gamma)\delta^k_i-(\rd_i\bt^k)-\delta_{im}\delta^{k\ell}(\rd_\ell\bt^m)\right)(G_{jk}-T_{jk})\\
=&e^{-2\gamma}\de^{ij}\rd_i(G_{j0}-T_{j0})-\f 1{4 N}\left(4 (e_0 \gamma)-2 (\rd_i\bt^i)\right) (e_0\tau )\\
&{-} e^{-2\gamma}\de^{ij}\f{\rd_i N}{N}(G_{j0}-T_{j0})-\f 14 e^{-2\gamma}\left(2(e_0 \gamma)\delta^{jk}-\de^{ij}(\rd_i\bt^k)-\delta^{ik}(\rd_i\bt^j)\right)\f{e^{2\gamma}(e_0\tau)}{N}\de_{jk}\\
=&e^{-2\gamma}\delta^{ij}\rd_i(G_{j0}-T_{j0})-\left(2 e_0 \gamma - \rd_i\bt^i\right)\f {e_0\tau}{N}-\f{e^{-2\gamma}\delta^{ij}(\rd_i N)}{N}(G_{j0}-T_{j0}).
\end{split}
\end{equation}
By \eqref{T.prop.prop} and the Bianchi equation $D^{\mu} G_{\mu 0}=0$, we have 
$$\mbox{\eqref{Bianchi.2.1} $+$ \eqref{Bianchi.2.2}}= -\sum_{\bA}F_{\bA}^2(\Box_g u_{\bf A}-\chi_{\bf A})(e_0 u_{\bf A}) ,$$
which implies \eqref{main.2.G}.

\textbf{Proof of \eqref{Boxu-chi}.} By Proposition~\ref{u.eikonal}, $(L_{\bA})^\mu = -\gi^{\mu\nu}\rd_\nu u_{\bA}$. Computing as in Section \ref{sec.Ray}, we get
\begin{equation*}
\begin{split}
L_{\bA}^\rho\rd_\rho(\Box_g& u_{\bA}) = -(\Box_g u_{\bA})^2 - R_{L_{\bA}L_{\bA}}\\
=& -(\Box_g u_{\bA})^2-2(L_{\bA}^\rho \partial_\rho \phi)^2-\sum_{\bf B} F_{\bf B}^2(L_{\bA}^\rho (\rd_\rho u_{\bf B}))^2-N(e_0\tau)(\wht {L_{\bA}})^0(\wht {L_{\bA}})^0 - 2(G_{0i}-T_{0i})(\wht {L_{\bA}})^0(\wht {L_{\bA}})^i,
\end{split}
\end{equation*}
where in the last line we have used \eqref{R00.2} and \eqref{Gij.2}. 
Subtracting \eqref{chi} from this and using Proposition~\ref{u.eikonal}, we then obtain \eqref{Boxu-chi}.
\end{proof}

\begin{prp}\label{prop:final}
Suppose the solution to \eqref{tau}--\eqref{chi} as constructed in Section \ref{sec.iterative} arises from initial data with $\tau\restriction_{\Sigma_0}=0$, $(\Box_{g}u_\bA-\chi_\bA)\restriction_{\Sigma_0}=0$ and that the constraint equations are initially satisfied, then the solution satisfies
$$\tau=0,\quad  \Box_g u_{\bA}=\chi_{\bA},\quad G_{0i}-T_{0i}=0.$$
As a consequence, the solution to \eqref{tau}--\eqref{chi} is indeed a solution to \eqref{back}.
\end{prp}
\begin{proof}
We will consider the equations \eqref{main.1.G}, \eqref{main.2.G} and \eqref{Boxu-chi} as a linear system for the unknowns $\tau$, $\Box_g u_\bA-\chi_\bA$ and $G_{0i}-T_{0i}$. We will use the Gr\"onwall's inequality to simultaneously show that they are zero. In order to carry this out, we need to put them in an appropriately weighted $L^2$ space. To see that all the weights are compatible, we will use the following two facts without further comments:
\begin{enumerate}
\item According to the estimates proven in Section \ref{sec.iterative}, not all derivatives of the metric components decay. The only subtlety here are the logarithmic terms as $|x|\to \infty$ in $\gamma$ and $N$. Nevertheless, one notes that all the \underline{spatial} derivatives of any metric components decay as $|x|\to \infty$, and $\rd_t\gamma$, $\rd_t \beta^i$ decay as $|x|\to \infty$. and $\rd_t\log N$ is bounded as $|x|\to \infty$. (The decay of $\rd_t\gamma$ follows from the fact that in the asymptotic term $\alpha \chi(|x|){\log} (|x|)$, $\alpha$ is a constant; while the boundedness of $\rd_t\log N$ follows from the fact that $N_{asymp}(t)$ is a function of $t$ only and the ${\log} (|x|)$ terms cancel.) Moreover, for $\ep_{low}$ sufficiently small, for any term that decays, the decay rate is faster than any powers of ${\log}(|x|)$ and is also faster than $e^{2\gamma}$.
\item $F_\bA$ is compactly supported.
\end{enumerate}

We begin the equations in Proposition \ref{prop.eqn.for.quantities.to.be.0}. Consider the energy
$$E(t):=\int_{\mathbb R^2}\left(\f 1{2N^2}(e_0\tau)^2+2 e^{-2\gamma}|G_{0i}-T_{0i}|^2+  |G_{0i}-T_{0i}-\f 12 N\rd_i \tau|^2+ (\Box_g u_\bA-\chi_\bA)^2 \right)(t,x)\, dx.$$
Contracting \eqref{main.1.G} with $N^2\delta^{ij}(G_{0j}-T_{0j}-\f 12 N\rd_j\tau)$ and integrating by parts, we obtain
\begin{equation}\label{est.E.1}
\f{d}{dt}\int_{\m R^2} |G_{0i}-T_{0i}-\f 12 N\rd_i \tau|^2(t,x)\, dx\leq CE(t)
\end{equation}
for some constant $C>0$.

Next, we consider the estimates for $\tau$, which is slightly more subtle. We compute using \eqref{main.1.G} and \eqref{main.2.G}. In the computation to follow, notice that whenever a derivative falls on the metric components $N$ and $\gamma$, we obtain a term $O(E(t))$. For this we have in particular used the decay and boundedness properties of the derivatives of the metric components that we mentioned above. Also, we can freely commute $e_0$ and $\rd_i$ and the error terms are again $O(E(t))$. More precisely, we have
\begin{equation}\label{est.E.2}
\begin{split}
&\f{d}{dt}\int_{\mathbb R^2}\left(\f 1{2 N^2}(e_0\tau)^2+2e^{-2\gamma}|G_{0i}-T_{0i}|^2\right)(t,x)\, dx\\
=& \int_{\mathbb R^2}\left(\f 1{N^2}(e_0\tau)(e_0 e_0\tau)+ 4e^{-2\gamma}(G_{i0}-T_{i0})(e_0(G_{i0}-T_{i0}))\right)(t,x)\, dx +O(E(t))\\
\stackrel{\eqref{main.2.G}}{=}& \int_{\mathbb R^2}\left(\f {2e^{-2\gamma}}{N}(e_0\tau)(\rd_i(G_{i0}-T_{i0}))+ 4e^{-2\gamma}(G_{i0}-T_{i0})(e_0(G_{i0}-T_{i0}))\right)(t,x)\, dx +O(E(t))\\
=& \int_{\mathbb R^2}\left(-\f {2e^{-2\gamma}}{N}(e_0\rd_i\tau)(G_{i0}-T_{i0})+4e^{-2\gamma}(G_{i0}-T_{i0})(e_0(G_{i0}-T_{i0}))\right)(t,x)\, dx +O(E(t))\\
\stackrel{\eqref{main.1.G}}{=}& \int_{\mathbb R^2}\left(-4e^{-2\gamma}(G_{i0}-T_{i0})(e_0(G_{i0}-T_{i0})) +4e^{-2\gamma}(G_{i0}-T_{i0})(e_0(G_{i0}-T_{i0}))\right)(t,x)\, dx +O(E(t))\\
=&O(E(t)).
\end{split}
\end{equation}
By \eqref{Boxu-chi}, we get that for some $C>0$
\begin{equation}\label{est.E.3}
\f{d}{dt}\int_{\m R^2} (\Box_g u_\bA-\chi_\bA)^2 (t,x)\, dx\leq C E(t).
\end{equation}
Combining \eqref{est.E.1}, \eqref{est.E.2} and \eqref{est.E.3}, we thus obtain
$$\f{d}{dt}E(t)\leq C E(t).$$
Finally, by assumptions, we have, initially, $E(0)=0$. Therefore, by Gr\"onwall's inequality, we have for every $t\in [0,T]$
$$E(t)=0$$
The conclusion follows immediately.
\end{proof}

\section{Improved regularity}\label{sec:improved.regularity}

Finally, to conclude the proof of Theorem~\ref{lwp}, it remains to prove the bounds stated in Theorem~\ref{lwp}. Notice that some of the estimates are already obtained in Theorem~\ref{thm:reduced.sys}. The remaining task is therefore to improve some of the estimates in Theorem~\ref{thm:reduced.sys}, using now the fact that we know the solution also solves the original \eqref{back} and we can therefore use the elliptic equations for the metric components. More precisely, we have the following proposition:
\begin{proposition}
In the $k = 3$ case, taking $\ep_{low}$ smaller if necessary, all the estimates stated in Theorem~\ref{lwp} hold for the solution to \eqref{back} constructed in Theorem~\ref{thm:reduced.sys} and Proposition~\ref{prop:final}.
\end{proposition}
\begin{proof}
\textbf{Estimates for $\phi$ and $F_\bA$.} We first note that by Theorem~\ref{thm:reduced.sys}, we already have all the desired estimates for $\nab\phi$, $\rd_t\phi$, $\rd_t^2\phi$, $F_\bA$ and $\rd_t F_\bA$. It thus remains to show the estimate for $\rd_t^2 F_\bA$. It can be easily seen that $\rd_t^2 F_\bA \in H^1$ by differentiating the equation \eqref{F} by $\rd_t$, and then using the estimates in Theorem~\ref{thm:reduced.sys} and the fact that $F_\bA$ is compactly supported in $B(0,2R)$.

\textbf{Estimates for the metric components.} Now that we know that we have a solution to \eqref{back}, it follows that the metric components satisfy the following elliptic equations:
\begin{align}
&\label{elliptic.1}
\Delta \beta^j={\delta^{ik}}\delta^{j\ell}\rd_k\left(\log(Ne^{-2\gamma})\right)(L\beta)_{i\ell}-4 \delta^{ij}(e_0 \phi)( \partial_i \phi) - 2\de^{ij} \sum_{\bA} F_\bA^2 (e_0 u_\bA)(\rd_i u_\bA)\\
&\Delta \gamma = -|\nabla \phi|^2 -\frac{1}{2}\sum F_\bA^2|\nabla u_\bA|^2-\frac{e^{2\gamma}}{N^2}\left( (e_0 \phi)^2+\frac{1}{2}\sum_{\bA}F_{\bA}^2
( e_0 u_{\bA})^2\right) -\frac{e^{2\gamma}}{8N^2}|L\beta|^2,\label{elliptic.2}\\
&\Delta N = \f{e^{2\gamma}}{4N}|L\beta|^2 + \frac{e^{2\gamma}}{N}\left(2(e_0 \phi)^2+\sum_{\bA}F_{\bA}^2
(e_0 u_{\bA})^2\right).\label{elliptic.3}
\end{align}
All the computations are given in Appendix~\ref{sec.Ricci}. More precisely, \eqref{elliptic.1} follows from \eqref{R0j} and \eqref{beta}; \eqref{elliptic.2} follows from \eqref{G00}; \eqref{elliptic.3} follows from $R_{00} = T_{00} - g_{00} \mbox{tr}_{g} T$, which can be computed using \eqref{R00} and \eqref{T-trT.00}.

For the metric components (without $\rd_t$ derivatives), we need $\wht N, \wht \gamma\in H^5_\de$ and $\bt \in H^5_{\de'}$. The estimate for $\wht N$ is already proven in Theorem~\ref{thm:reduced.sys}. To improve the estimates for $\wht \gamma$ and $\bt$, we use \eqref{elliptic.1}, \eqref{elliptic.3}, the estimates in Theorem~\ref{thm:reduced.sys} and Corollary~\ref{coro} to obtain that $\wht \gamma\in H^5_\de$, $\bt \in H^5_{\de'}$. (Notice that $\wht N$ and $\wht\gamma$ can be put in a better weighted space compared to $\bt$ since on the RHS of \eqref{elliptic.2} and \eqref{elliptic.3}, the terms that are not compactly supported is quadratic in $L\bt$ and decay sufficiently fast. This is in contrast to, say, the term ${\delta^{ik}}\delta^{j\ell}\rd_k\left(\log(Ne^{-2\gamma})\right)(L\beta)_{i\ell}$ on the RHS of \eqref{elliptic.1}.)

\textbf{Estimates for $u_\bA$ and its $\rd_t$ derivatives.} First, by Proposition~\ref{u.eikonal} and \eqref{L.coord}, we have
$$L_\bA^t = \frac{1}{N^2}(e_0 u_\bA), \quad L_\bA^i=-\de^{ij} e^{-2\gamma}(\partial_j u_\bA) - \bt^i \frac{1}{N^2}(e_0 u_\bA).
	$$ 
Hence, derivatives of $u_\bA$ can be written in terms of components of $L_\bA$ and their derivatives. Therefore, by Theorem~\ref{thm:reduced.sys}, we have the desired estimates for $u_\bA$ when there are at most two $\rd_t$ derivatives on $u_\bA$. Moreover, the upper bound for $\left(\min_{\bA} \inf_{x\in \mathbb R^2} |\nab u_{\bA}|(x)\right)^{-1}$ holds due to the lower bound for $N e^\gamma L_\bA^t$ in Theorem~\ref{thm:reduced.sys}.

Finally, to bound the third $\rd_t$ derivatives for $u_\bA$, i.e., the term $\partial_t^2\left(e^{\gamma}N^{-1} e_0 u_\bA \right)$, we simply note that since $\gi^{\alp\bt}\rd_\mu u_{\bA} \rd_\nu u_{\bA} =0$ (cf.~Proposition~\ref{u.eikonal}), 
$$e_0 u_\bA  = N e^{-\gamma} |\nab u_\bA|.$$
Hence, we can write $e^{\gamma}N^{-1} e_0 u_\bA$ in terms of $\nab u_\bA$ and apply\footnote{Notice here that by lower bound on  $|\nab u_{\bA}|$, $|\nab u_\bA|$ is bounded away from $0$ and therefore has the same regularity properties as $\nab u_\bA$.} the estimates above.

\textbf{Estimates for the $\rd_t$ derivatives of the metric components.} To estimate the $\rd_t$ derivatives of the metric components, we again use \eqref{elliptic.1}, \eqref{elliptic.2} and \eqref{elliptic.3}. Differentiating \eqref{elliptic.1}, \eqref{elliptic.2}, \eqref{elliptic.3} by $\rd_t$ and using the estimates in Theorem~\ref{thm:reduced.sys} as well as Corollary~\ref{coro}, we have that for some $C_h'$ depending on $C_{eik}, C_{high}, k,\de,R$, the following estimates hold:
$$|\rd_t N_{asymp}|+ \|\rd_t\wht N\|_{H^2_{\de}} + \|\rd_t\bt\|_{H^2_{\de'}} + \|\rd_t \wht \gamma\|_{H^2_{\de}} \ls \ep^2 \left(|\rd_t N_{asymp}|+ \|\rd_t\wht N\|_{H^1_{\de}} + \|\rd_t\bt\|_{H^1_{\de'}} + \|\rd_t \wht \gamma\|_{H^1_{\de}}\right)+ C'_h ,$$
Hence, we have $\rd_t\wht N \in H^2_\de$, and $\rd_t\bt, \, \rd_t\wht\gamma\in H^2_{\de'}$. Now, by using again the equations \eqref{elliptic.1}, \eqref{elliptic.2}, \eqref{elliptic.3} differentiated by $\rd_t$, we can apply Corollary~\ref{coro} to iteratively improve the regularity until we obtain
$$\rd_t\wht N,\,\rd_t\wht\gamma \in H^4_\de,\quad \rd_t\bt\in H^4_{\de'}.$$

Similarly, to estimate the $\rd_t^2$ derivatives of the metric components, we differentiate \eqref{elliptic.1}, \eqref{elliptic.2}, \eqref{elliptic.3} by $\rd_t^2$. Using the estimates in Theorem~\ref{thm:reduced.sys} , the estimates for the $\rd_t$ derivatives of the metric components (that we just derived above), and also Corollary~\ref{coro}, we have that for some $C_h'$ depending on $C_{eik}, C_{high}, k,\de,R$, the following estimates hold:
$$|\rd_t^2 N_{asymp}|+ \|\rd_t^2\wht N\|_{H^2_{\de}} + \|\rd_t^2\bt\|_{H^2_{\de'}} + \|\rd_t^2 \wht \gamma\|_{H^2_{\de}} \ls \ep^2 \left(|\rd_t^2 N_{asymp}|+ \|\rd_t^2\wht N\|_{H^1_{\de}} + \|\rd_t^2\bt\|_{H^1_{\de'}} + \|\rd_t^2 \wht \gamma\|_{H^1_{\de}}\right)+ C''_h ,$$
This implies that $\rd_t^2\wht N, \,\rd_t^2\wht\gamma \in H^2_\de$, and $\rd_t^2\bt\in H^2_{\de'}$. 
As before, we now use the equations \eqref{elliptic.1}, \eqref{elliptic.2}, \eqref{elliptic.3} differentiated by $\rd_t^2$ and apply Corollary~\ref{coro} to iteratively improve the regularity until we obtain
$$\rd_t^2\wht N,\,\rd_t\wht\gamma \in H^3_\de,\quad \rd_t^2\bt\in H^3_{\de'}.$$
\end{proof}

This concludes the proof of Theorem~\ref{lwp} when $k=3$. As we mentioned earlier, in the case of larger $k$, one can easily show that higher regularity is propagated, and we will omit the proof.

\appendix

\section{Weighted Sobolev spaces}\label{weightedsobolev}
For sake of completeness, we collect some results about weighted Sobolev spaces. For this we recall the definitions in Definition \ref{def.spaces}. Unless otherwise stated, we will only be interested in weighted Sobolev spaces on $\mathbb R^2$. Most of the results can be found in \cite[Appendix I]{livrecb} (although we use slightly different notations).

The following is immediate from the definition.

\begin{lm}\label{der} Let $m \geq 1$, $p\in [1, \infty)$ and $\delta \in \mathbb{R}$. Then there exists $C>0$ such that for $j=1,2$, 
$$\| \partial_j u\|_{W^{m-1}_{\delta+1,p}} \leq C \|u \|_{W^m_{\delta,p}}.$$
Similarly, for $m \geq 1$, $\delta \in \mathbb{R}$, $j=1,2$, 
$$\| \partial_j u\|_{C^{m-1}_{\delta+1}} \leq C \|u \|_{C^m_{\delta}}.$$
\end{lm}

We have an easy embedding result, which is a straightforward application of the H\"older's inequality:
\begin{lm}\label{weight:emb}
If $1 \leq p_1 \leq p_2 \leq \infty$ and $\delta_2-\delta_1 > 2\left(\f{1}{p_1}-\f{1}{p_2}\right)$, 
then we have the continuous embedding
$$W^{0}_{\delta_2,p_2}\subset W^{0}_{\delta_1,p_1}.$$
\end{lm}

Next, we have Sobolev embedding theorems for weighted Sobolev spaces:
\begin{prp}\label{holder} 
Let $s,m \in \m N \cup \{0\}$, $1<p<\infty$. The following holds:
\begin{itemize}
\item Suppose $s >\frac{2}{p}$ and $\beta \leq \delta +\frac{2}{p}$. Then, the following continuous embedding holds
	$$W^{s+m}_{\delta,p}\subset C^{m}_{\beta}.$$
\item Suppose $s < \f 2p$. Then, the following continuous embedding holds
  $$W^{s+m}_{\delta,p}\subset W^{m}_{\delta+s, \f{np}{n-sp}}.$$
\end{itemize}
\end{prp}

We will also need a product estimate.
\begin{prp}\label{produit}
	Let $s,s_1,s_2 \in \m N \cup \{0\}$, $p \in [1,\infty]$, $\de, \de_1, \de_2 \in \m R$. Assume that $s\leq \min(s_1,s_2)$ and $s<s_1+s_2-\f 2 p$. Let $\delta < \delta_1 + \delta_2 + \f 2p$. Then $\forall (u,v) \in W^{s_1}_{\delta_1, p}\times W^{s_2}_{\delta_2, p}$,
	$$\|uv\|_{W^s_{\delta,p}} \lesssim \|u\|_{W^{s_1}_{\delta_1,p}} \|v\| _{W^{s_2}_{\delta_2, p}}.$$
\end{prp}

We also state another product estimate, which concerns unweighted $H^s$ spaces, and is especially convenient when handling compactly supported functions. See \cite[Appendix~A]{Tao} for a proof.
\begin{prp}\label{product}
 Let $s \in \m N$. Then $\forall (u,v)\in (H^s\cap L^\infty) \times (H^s\cap L^\infty)$,
   $$\| uv \|_{H^s} \ls \|u\|_{H^s} \|v\|_{L^\infty} + \|u\|_{L^\infty} \|v\|_{H^s}.$$
\end{prp}

The following simple lemma will be useful as well.

\begin{lm}\label{produit2}
	Let $\alpha \in \mathbb{R}$ and $g \in L^\infty_{loc}$ be such that
	$$|g(x)| \lesssim (1+|x|^2)^\alpha.$$
	Then the multiplication by $g$ maps $H^0_{\delta}$ to $H^0_{\delta -2\alpha}$ with operator norm bounded by $\sup_{x\in \m R^2} \f{|g(x)|}{(1+|x|^2)^{\alp}}$.
\end{lm}

{The next result, which is due to McOwen, concerns} the invertibility of the Laplacian on weighted Sobolev spaces. 
\begin{thm}\label{laplacien}(Theorem 0 in \cite{laplacien})
	Let $m\in \mathbb{Z}$, $m\geq 0$, $1<p<\infty$ and $-\frac{2}{p}+m<\delta<m+1-\frac{2}{p}$. The Laplace operator $\Delta :W^{2+m}_{\delta,p} \rightarrow W^{m}_{\delta+2,p}$ is an injection with closed range 
	$$\left \{f \in W^{m}_{\delta+2,p}\; | \;\int fv =0 \quad \forall v \in \cup_{i=0}^m \mathcal{H}_i \right \},$$
	where $\mathcal{H}_i$ is the set of harmonic polynomials of degree $i$.
	Moreover, $u$ obeys the estimate
	$$\|u\|_{W^{2+m}_{\delta,p}} \leq C(\delta,m,p)\|\Delta u\|_{W^{m}_{\delta+2,p}},$$
	where $C(\delta,m,p) > 0$ is a constant depending on $\delta$, $m$ and $p$.
\end{thm}

{The following is a corollary of Theorem~\ref{laplacien}}:
\begin{cor} \label{coro} Let $-1<\delta<0$ and $f\in H^0_{\delta+2}$. Then there exists a solution $u$ of 
	$$\Delta u =f$$
	which can be written 
	$$u=\frac{1}{2\pi}\left(\int f\right)\chi(|x|){\log}(|x|) +v,$$
	where $\chi$ is as in Definition~\ref{def.cutoff}, and $v \in H^2_{\delta}$ is such that
	$\|v\|_{H^2_\delta} \leq C(\delta)\|f\|_ {H^0_{ \delta+2}}$.
\end{cor}


\section{Computations in the elliptic gauge}\label{app.B}

In this section, we collect some computations for the spacetime metric in the elliptic gauge defined in Section~\ref{sec.elliptic.gauge}. We will frequently use conventions defined in Section~\ref{sec.elliptic.gauge} without further comment.

\subsection{Connection coefficients}\label{sec.connection}
We compute the connection coefficients for the metric \eqref{g.form} with respect to the frame $\{e_0, e_1, e_2\}$, where $e_0=\rd_t-\bt^i\rd_i$ and $e_i=\rd_i$. Notice that
$$[e_0,e_i]=[\rd_t-\bt^j\rd_j, \rd_i]=(\rd_i\bt^j)\rd_j,\quad [e_i,e_j]=0,\quad g(e_i,e_0)=0.$$
Using this, we compute
$$g(D_0 e_0, e_0)=\f12 e_0(g(e_0,e_0))=-\f 12 (e_0 N^2)=-N(e_0 N), $$
$$g(D_0 e_0, e_i)=-g(e_0,D_0 e_i)=-g(e_0,D_i e_0)-(\rd_i\bt^j)g(e_0, e_j)=-\f 12 \rd_i (g(e_0,e_0))=N\rd_i N,$$
$$g(D_i e_0, e_0)=\f 12 \rd_i(g(e_0,e_0))=-\f 12\rd_i N^2=-N\rd_i N, $$
\begin{equation}\label{II}
g(D_i e_0, e_j)=\f {e^{2\gamma}}{2}\left((2e_0 \gamma)\delta_{ij}-(\rd_i\bt^k)\delta_{jk}-(\rd_j\bt^k)\delta_{ik}\right),
\end{equation}
$$g(D_0 e_i, e_0)=-g(D_0 e_0, e_i)=-N\rd_i N,$$
$$g(D_0 e_i, e_j)=g(D_i e_0, e_j)+e^{2\gamma}(\rd_i \bt^k)\delta_{jk}=\f {e^{2\gamma}}2\left((2e_0 \gamma)\delta_{ij}+(\rd_i\bt^k)\delta_{jk}-(\rd_j\bt^k)\delta_{ik}\right),$$
$$g(D_i e_j, e_0)=-g(D_i e_0, e_j)=-\f {e^{2\gamma}}{2}\left((2 e_0 \gamma)\delta_{ij}-(\rd_i\bt^k)\delta_{jk}-(\rd_j\bt^k)\delta_{ik}\right),$$
$$g(D_i e_j, e_k)=e^{2\gamma}\left(\delta_{ik}\partial_j \gamma +\delta_{jk}\partial_i \gamma-\delta_{ij}\delta_k^{\ell}\partial_{\ell} \gamma\right).$$
Most of these are straightforward, let us just mention that \eqref{II} is derived using the symmetry $g(D_{e_i} e_0, e_j)=g(D_{e_j} e_0, e_i)$ so that
\begin{equation*}
\begin{split}
(e_0 e^{2\gamma})\delta_{ij}=e_0(g(e_i,e_j))=g(D_0 e_i, e_j)+g(D_0 e_j, e_i)=2g(D_i e_0, e_j)+e^{2\gamma}(\rd_i\bt^k)\delta_{jk}+e^{2\gamma}(\rd_j\bt^k)\delta_{ik}.
\end{split}
\end{equation*}
Note in particular that the computation for $g(D_i e_j, e_k)$ implies that the Christoffel symbols $\bar{\Gamma}_{ij}^k$ associated to the spatial metric $\bar{g}$ (cf. \eqref{g.form.0}) are given by
\begin{equation}\label{sp.Chr}
\bar{\Gamma}_{ij}^k=\delta_{i}{ }^k\partial_j \gamma + \delta_{j}{ }^k\partial_i \gamma-\delta_{ij}\de^{k\ell}\partial_{\ell} \gamma.
\end{equation}
From the above calculations, we then obtain
\begin{equation}\label{connections}
\begin{split}
D_0 e_0=&\f{e_0 N}{N} e_0+e^{-2\gamma}\delta^{ij}N \rd_i N e_j,\quad D_0 e_i=\f{\rd_i N}{N} e_0+\f 12 ((2 e_0 \gamma)\delta_i^j+(\rd_i\bt^j)-\delta_{ik}\delta^{j\ell}(\rd_\ell\bt^k))e_j,\\
D_i e_0=&\f{\rd_i N}{N} e_0+\f 12  ((2 e_0 \gamma)\delta_i^j-(\rd_i\bt^j)-\delta_{ik}\delta^{j\ell}(\rd_\ell\bt^k))e_j,\\
D_i e_j=&\f{e^{2\gamma}}{2N^2}\left((2 e_0 \gamma)\delta_{ij}-(\rd_i\bt^k)\delta_{jk}-(\rd_j\bt^k)\delta_{ik}\right)e_0+\left(\delta_{i}^k\partial_j \gamma +\delta_j^k\partial_i \gamma-\delta_{ij}\delta^{k\ell}\partial_{\ell} \gamma\right)e_k.
\end{split}
\end{equation}

\subsection{Decomposition of the Ricci tensor} \label{sec.Ricci}
\begin{proposition}
Given $g$ of the form \eqref{g.form}, the second fundamental form $K_{ij}$ (cf. \eqref{K}) is given by
\begin{equation}\label{2nd.fund.form}
K_{ij}=-\f 1{2N} e_0(e^{2\gamma})\delta_{ij} +\f{1}{2N} e^{2\gamma}(\partial_i \beta_j+\partial_i \beta_j),
\end{equation}
Moreover, its traceless part $H$ satisfies
\begin{equation}
\label{beta} 2Ne^{-2\gamma}H_{ij}=(L\beta)_{ij},
\end{equation}
where 
\begin{equation}\label{L.def}
(L\beta)_{ij}:= \de_{j\ell} \rd_i \bt^\ell + \de_{i\ell} \rd_j \bt^\ell -\de_{ij} \rd_k \bt^k
\end{equation}
as in Section~\ref{sec.notations}.
\end{proposition}
\begin{proof}
\eqref{2nd.fund.form} follows from \eqref{K}; and \eqref{beta} follows from \eqref{2nd.fund.form}.
\end{proof}

\begin{proposition}
Given $g$ of the form \eqref{g.form}, the components of the Ricci curvature in the basis $\{e_0=\rd_t-\bt^k\rd_k, \rd_i\}$ are given by
\begin{align}
\label{Rij} R_{ij} = &\delta_{ij}\left(-\Delta \gamma+\frac{\tau^2}{2}e^{2\gamma}-\frac{1}{2N}e^{2\gamma}e_0 \tau-\frac{1}{2N}\Delta N\right)
-\frac{1}{N}(\partial_t-\beta^k\partial_k)H_{ij}-2e^{-2\gamma} H_i{ }^{\ell} H_{j\ell} \\
\notag &+\frac{1}{N}\left(\partial_j \beta^k H_{ki}+\partial_i \beta^k H_{kj}\right)-\frac{1}{N}\left( \partial_i \partial_j N -\frac{1}{2}\delta_{ij}\Delta N -\left(\delta_i^k\partial_j \gamma +\delta_j^k\partial_i \gamma-\delta_{ij} \de^{\ell k}\partial_{\ell} \gamma\right)\partial_k N \right),\\
\label{R0j} R_{0j} =& N\left(\frac 12\partial_j \tau -e^{-2\gamma}\de^{ik}\partial_k H_{ij}\right),\\
\label{R00} R_{00} =& N\left(e_0 \tau -Ne^{-4\gamma}\left(|H|^2+ \frac{1}{2}e^{4\gamma}\tau^2\right)+ e^{-2\gamma}\Delta N\right).
\end{align}
Moreover, 
\begin{equation}\label{spatial.R.tr}
\delta^{ij}R_{ij}=2\left(-\Delta \gamma+\frac{\tau^2}{2}e^{2\gamma}-\frac{1}{2N}e^{2\gamma}e_0 \tau-\frac{1}{2N}\Delta N\right).
\end{equation}
\end{proposition}
\begin{proof}
From \cite[Chapter~6]{livrecb}, we have
\begin{align}
\label{Rij.CB} R_{ij} = & \bar{R}_{ij} + K_{ij}(\mbox{tr}_{\bar{g}} K) -2 (\bar{g}^{-1})^{ml} K_{im}K_{jl} - N^{-1}(  \q L_{e_0} K_{ij} + \bar{D}_i \partial_j N), \\
\label{R0j.CB} R_{0j} = & N(\partial_j (\mbox{tr}_{\bar{g}} K) -(\bar{g}^{-1})^{hk}\bar{D}_h {K}_{jk}),\\
\label{R00.CB} R_{00} = & N( e_0(\mbox{tr}_{\bar{g}} K) -N(\bar{g}^{-1})^{ii'}(\bar{g}^{-1})^{jj'}K_{ij}K_{i'j'} + \Delta_{\bar{g}} N),
\end{align}
where $\bar{D}$, $\bar{R}_{ij}$ and $\Delta_{\bar{g}}$ are defined with respect to $\bar{g}$.

\textbf{Proof of \eqref{Rij}.} First, by \eqref{K.tr.trfree} and \eqref{sp.Chr}, we have
\begin{equation}\label{Rij.0}
\begin{split}
R_{ij} = & -\delta_{ij}\Delta \gamma+\tau \left(H_{ij}+\frac{1}{2}e^{2\gamma}\delta_{ij}\tau\right)
 -2e^{-2\gamma}\left(H_{i}{ }^{\ell}+\frac{1}{2}e^{2\gamma}\delta_i^{\ell} \tau\right) \left(H_{j\ell}+\frac{1}{2}e^{2\gamma}\delta_{j\ell} \tau\right)\\
&-\frac{1}{N}\left( \q L_{e_0} K_{ij} +\partial_i \partial_j N - \left(\delta_i^k\partial_j \gamma +\delta_j^k\partial_i \gamma-\delta_{ij}\de^{k\ell} \partial_\ell \gamma\right)\partial_k N \right).
\end{split}
\end{equation}
To proceed, we compute $\q L_{e_0} K_{ij}$ by considering $H_{ij}$ and $\tau$. We calculate
\begin{align*}
 \q L_{e_0} H_{ij} 
=& (\partial_t-\beta^k\partial_k)H_{ij}-\partial_j \beta^k H_{ki}-\partial_i \beta^k H_{kj},\\
\q L_{e_0}(\tau \bar{g}_{ij}) =& e^{2\gamma}\delta_{ij}e_0 \tau -2N\tau K_{ij}.
\end{align*}
Therefore, using \eqref{K.tr.trfree} and plugging $\q L_{e_0} K_{ij}$ into \eqref{Rij.0}, we obtain \eqref{Rij}.

\textbf{Proof of \eqref{R0j}.} This follows from \eqref{R0j.CB} and the fact that (using \eqref{sp.Chr},) for any covariant symmetric $2$-tensor $A_{ij}$, 
$$(\bar{g}^{-1})^{ik}\bar{D}_k A_{ij}=e^{-2\gamma}\rd^i A_{ij}-(\partial_j\gamma)(\mbox{tr}_{\bar{g}}A).$$

\textbf{Proof of \eqref{R00}.} This is immediate from \eqref{R00.CB} and the conformal invariance of the Laplacian.

\textbf{Proof of \eqref{spatial.R.tr}.} Finally, to prove \eqref{spatial.R.tr}, we first note that 
$$\delta^{ij}\left(\partial_j \beta^k H_{ki}+\partial_i \beta^k H_{kj}\right)
= H_{ij}(L\beta)^{ij},$$
where $L$ is as in \eqref{L.def}. Combining this with \eqref{beta}, we obtain
$$\delta^{ij}\left(-2e^{-2\gamma}H_i{ }^{\ell} H_{j\ell} +\frac{1}{N}\left(\partial_j \beta^k H_{ki}+\partial_i \beta^k H_{kj}\right)\right)=0.$$
Taking the trace of \eqref{Rij} and using this identity yield \eqref{spatial.R.tr}. \qedhere
\end{proof}

\begin{proposition}
Given $g$ of the form \eqref{g.form}, the scalar curvature $R$ and the $G_{00} = G(e_0, e_0)$ component of the Einstein tensor are given by
\begin{align}
\label{scalarR} R =& -\frac{2}{N}e_0 \tau +\frac{3}{2}\tau^2+e^{-4\gamma}|H|^2-\frac{2e^{-2\gamma}}{N}\Delta N-2e^{-2\gamma}\Delta \gamma,\\
\label{G00} G_{00} =& N^2e^{-2\gamma}\left(-\Delta \gamma -e^{-2\gamma}\frac{1}{2}|H|^2+e^{2\gamma}\frac{1}{4}\tau^2\right).
\end{align}
\end{proposition}
\begin{proof}
By \eqref{g.form}, \eqref{R00} and \eqref{spatial.R.tr},
\begin{align*}
R=& -\frac{1}{N}\left(e_0 \tau -Ne^{-4\gamma}\left(|H|^2+ \frac{1}{2}e^{4\gamma}\tau^2\right)+ e^{-2\gamma}\Delta N\right)
+e^{-2\gamma}2\left(-\Delta \gamma+\frac{\tau^2}{2}e^{2\gamma}-\frac{1}{2N}e^{2\gamma}e_0 \tau-\frac{1}{2N}\Delta N\right).
\end{align*}
Simplifying this yields \eqref{scalarR}. By \eqref{g.form}, \eqref{R00} and \eqref{scalarR},
\begin{align*}
G_{00}=&\frac{1}{2}N\left(e_0 \tau -Ne^{-4\gamma}\left(|H|^2+ \frac{1}{2}e^{4\gamma}\tau^2\right)+ e^{-2\gamma}\Delta N\right)
+e^{-2\gamma}N^2\left(-\Delta \gamma+\frac{\tau^2}{2}e^{2\gamma}-\frac{1}{2N}e^{2\gamma}e_0 \tau-\frac{1}{2N}\Delta N\right).
\end{align*}
Simplifying this yields \eqref{G00}.\qedhere
\end{proof}

\subsection{Computations for the stress-energy-momentum tensor}\label{sec.T}

Definet $T_{\mu\nu}$ by
$$T_{\mu\nu}:=2\rd_\mu\phi\rd_\nu\phi - g_{\mu\nu}(g^{-1})^{\alp\bt}\rd_\alp\phi\rd_\bt\phi+\sum_{\bA} F^2_{\bA} \rd_\mu u_{\bA}\rd_\nu u_{\bA}.$$
If $\gi^{\mu\nu}\rd_\mu u_{\bA}\rd_\nu u_{\bA}=0$, then
$$\mbox{tr}_g T=- \gi^{\alp\bt}\rd_\alp\phi\rd_\bt\phi+\sum_{\bA} F^2_{\bA} \gi^{\mu\nu}\rd_\mu u_{\bA}\rd_\nu u_{\bA}=- \gi^{\alp\bt}\rd_\alp\phi\rd_\bt\phi.$$
This implies, with respect to the $\{e_0,\rd_i\}$ basis,
\begin{align}
\label{T-trT.00} T_{00} - g_{00}\mbox{tr}_g T =& 2(e_0\phi)^2+\sum_{\bA} F_{\bA}^2 (e_0 u_{\bA})^2,\\
\label{T-trT.ij} T_{ij}-g_{ij}\mbox{tr}_g T =& 2\rd_i\phi\rd_j\phi+\sum_{\bA} F_{\bA}^2(\rd_i u_{\bA})(\rd_j u_{\bA}),\\
\label{T-trT.spatial.tr} \delta^{ij}(T_{ij}-g_{ij}\mbox{tr}_g T) =& 2\delta^{ij}\rd_i\phi\rd_j\phi+\sum_{\bA} \f{e^{2\gamma}}{N^2} F_{\bA}^2(e_0 u_{\bA})^2.
\end{align}

\section{Computations regarding eikonal functions}\label{app.C}

\subsection{Geodesic equation in coordinates}
	Suppose $u$ satisfies the eikonal equation\footnote{In applications, $u$ will be $u_{\bA}$ as in \eqref{back}.} 
	$$\gi^{\mu\nu}\rd_\mu u \rd_\nu u=0.$$ 
	Observe that a consequence of the eikonal equation is the geodesic equation $D_{(du_\bA)^\sharp}(du_\bA)^\sharp=0$. As is well-known, they are in fact equivalent: one can solve for $D_L L=0$, $Lu=0$ for $L$ being an appropriately defined future-directed null vector field initially orthogonal to the hypersurface of $u$ so that $u$ satisfies the eikonal equation and $L=-(du)^\sharp$.
	
	In our setting, it is convenient to solve the eikonal equation using yet another (equivalent) set of equations. In this subsection, we derive these equations by performing some manipulations in coordinates, with $g$ given by \eqref{g.form}.

	Suppose we are given a solution $u$ to the eikonal equation, with $L=-(du)^\sharp$ future-directed. In terms of the basis $\{e_0, \rd_1,\rd_2\}$, $L$ is given by
	\begin{equation}\label{L.coord}
	L^t = \frac{1}{N^2}(e_0 u), \quad L^i=-\de^{ij} e^{-2\gamma}(\partial_j u) - \bt^i \frac{1}{N^2}(e_0 u).
	\end{equation}
By convention, we take $e_0 u>0$ so that $L$ is future-directed. In terms of $L$, the fact that $u$ satisfies the eikonal equation can be expressed as follows:		
  \begin{equation}\label{L.eikonal}
	N^2(L^t)^2= e^{2\gamma} \de_{ij} (L^i+\bt^i L^t) (L^j+\bt^j L^t).
	\end{equation}

\subsection{Raychaudhuri equation}\label{sec.Ray}

Let $u$ be a solution to the eikonal equation
$$\gi^{\alpha\beta}\partial_\alpha u \partial_\beta u=0.$$
Consider the vector field
$$L^\beta := -\gi^{\alpha \beta}\partial_\alpha u.$$
This vector field is null and geodesic. The second fact can be proven as follows:
\begin{align*}
D_L L^\beta &= L^\alpha D_\alpha L^\beta \\
&= \gi^{\alpha \rho}\partial_\rho u D_\alpha( \gi^{\beta \mu}\partial_\mu u)\\
&=\gi^{\alpha \rho}\gi^{\beta \mu}\partial_\rho u D_\alpha \partial_\mu u \\
&=\gi^{\alpha \rho}\gi^{\beta \mu}\partial_\rho u D_\mu\partial_\alpha u \\
&=\frac{1}{2}\gi^{\beta \mu}D_\mu (\gi^{\alpha \rho}\partial_\alpha u \partial_\rho u)\\
&=0.
\end{align*}
Let $\ell_{t_*,u_*}:=\{(t,x^1,x^2):t=t_*, u(t,x^1,x^2)=u_*\}$ and let $e_\theta$ be the unique unit (spacelike) vector field tangent to $\ell_{t,u}$. Let $\ba L$ be the unique null vector field which satisfies both $g(\ba L, e_\theta)=0$ and $g(L,\ba L)= -2$. Notice that $e_\theta$ verifies
$g(e_\theta,L)=g(e_\theta, \ba L)=0$ and $g(e_\theta,e_\theta)=1$.
Let
$$\chi:= \langle D_{e_\theta}L,e_\theta \rangle_g.$$
We write
$$D_{e_\theta }L=\chi e_\theta -\zeta  L,\quad D_L e_{\theta} =\ba \eta  L,$$
where $\zeta:=\f 12 \langle D_{e_\theta} L, \ba L\rangle_g$ and $\ba \eta:=-\f 12 \langle D_L e_{\theta}, \ba L \rangle_g$.
We calculate
\begin{equation}\label{Raychaudhuri}
\begin{split}
D_L \chi &= D_L  \langle D_{e_\theta}L,e_\theta \rangle_g\\
&=\langle D_L D_{e_\theta}L,e_\theta\rangle_g+\langle D_{e_\theta}L,D_L e_\theta \rangle_g\\
&=R_{L \theta L}{} ^{\theta} + \langle D_{e_\theta}D_L L,e_\theta\rangle_g + \langle D_{[L,e_\theta]}L,e_\theta\rangle_g + \langle D_{e_\theta}L,D_L e_\theta\rangle_g\\
&=R_{L \theta L}{} ^{\theta}+(\ba \eta+\zeta)\langle D_{L}L,e_{\theta}\rangle_g - \chi\langle D_{e_\theta }L,e_\theta\rangle_g.
\end{split}
\end{equation}
Consequently, using that $L$ is geodesic,
 $$L(\chi)+\chi^2= R_{L \theta L}{}^{\theta}=-R_{LL}.$$
Moreover, we calculate
\begin{equation}\label{boxu=chi}
\Box_g u = div(L)=\langle D_{e_\theta}L, e_{\theta}\rangle_g - \f 12 \langle D_{L} L,\ba L\rangle_g - \f 12 \langle D_{\ba L}L,L\rangle_g = \langle D_{e_\theta}L, e_{\theta}\rangle_g=\chi.
\end{equation}

\end{document}